\documentclass[11pt]{amsart}
\usepackage{amsfonts,amssymb,graphicx,mystyle,color,istgame}
\usepackage[colorlinks,linkcolor=blue]{hyperref} 
\usepackage{tikz}
\usepackage[round]{natbib}
\usepackage{hyperref}
\hypersetup{
    colorlinks = true,
    citecolor = {blue},
}

 \def\b{\beta} \def\d{\delta} \def\e{\varepsilon}   \def\l{\lambda} \def\s{\sigma}  
\def\D{\Delta} \def\G{\Gamma}  \def\S{\Sigma}
     \def\N{\mathcal N} 

\def\P{\mathcal P}   
\def\Re{\mathbb{R}}

\def\rho{\varrho}

\def\E{\mathcal E}

\usetikzlibrary{calc}
\usepackage{amssymb}
\usepackage{tikz-cd}
\usepackage{natbib}
\usepackage{enumitem}
\usepackage{setspace}
\usepackage{float}
\usepackage{commath}
\usepackage{caption}
\usepackage{subcaption}
\usepackage{multirow,array}

\newcommand{\n}{\noindent}
\def\E{{\rm \mathcal{E}}\,}

\def\mG{\mathbb{G}}

\begin{document}\openup 1\jot

\title[Index Theory]{Polytope-form games and Index/Degree Theories for Extensive-Form Games}
\author{Lucas Pahl}
\thanks{I am grateful to Paulo Barelli and Hari Govindan for their guidance and encouragement. I thank Sven Rady, Rida Laraki, Klaus Ritzberger, Carlos Pimienta, Eilon Solan, and Heng Liu for comments and suggestions. An older version of this paper circulated with the title ``Index Theory of Strategic-form Games with an application to Extensive-form Games''. I acknowledge financial support from the Hausdorff Center for Mathematics (DFG project no. 390685813)} 
\address{Institute for Microeconomics, University of Bonn, Adenauerallee 24-42, 53113 Bonn, Germany.}
\email{pahl.lucas@gmail.com}
\date{\today}
\maketitle

\begin{abstract}
We develop index and degree theories for extensive form games allowing the identification of equilibria that are robust to payoff perturbations directly from the extensive form. Our approach is based on index and degree theories for games where the strategy sets are polytopes (and not necessarily simplices) and payoff functions are multaffine. Polytope strategy sets arise naturally from topologically identifying equivalent mixed strategies of a normal form game.

\end{abstract}

\section{Introduction}

The index and degree theories of equilibria offer a selection criterion for equilibria in games which has wide applications both in the literature of dynamic as well as strategic stability (cf. \cite{RWW2022}, \cite{KR1994},  \cite{GLP2022}, \cite{GW2005} and \cite{DR2003}). 

The index of equilibria is essentially an integer number assigned to each connected component of equilibria of a finite game that measures whether an equilibrium is robust to payoff perturbations of this game. It can be readily defined using the characterization of equilibria as fixed points of the best-reply correspondence. The degree requires a bit more work to be properly defined. It is also an integer number that was shown to be identical to the index, thereby also capturing robustness of the equilibria to payoff perturbations. There are many distinct ways in which indices or degrees can be computed. Depending on the specifics of the problem considered, some formulas used to compute the degree might be more applicable than index-theoretic ones, or the opposite might be true. 

Index and degree theories of equilibria were formulated for normal-form games. This implies that whenever an extensive-form game is given, it has to be represented in normal form so that the index or degree can be computed. There are at least two computational inconveniencies that the representation of extensive-form games in normal form implies. First the number of pure strategies of a player (even in the reduced normal form of an extensive-form game) might be an exponential function of the number of terminal nodes of the game tree, causing this representation to be computationally intractable. Second, formulas used for the computation of the degree of equilibria in normal form involve considering small payoff perturbations of the original payoffs to nearby \textit{regular} ones, i.e., payoffs that define games with finitely many equilibria and where the equilibria are smooth and invertible functions of the payoffs of the game. Though regular games are generic, perturbing a given game to a generic one might involve varying the normal-form payoffs in a space whose dimension is exponentially larger than that of the terminal payoffs of the extensive form. This fact alone implies that the formulas for the degree in normal form are hard to compute, which ultimately implies that checking for robustness of equilibria to payoff perturbations is also computationally intractable.   

The first of these problems can be overcome by considering an alternative representation of extensive-form games with perfect recall, developed by \cite{BS1996}, called the \textit{sequence form}. In this alternative representation, the number of strategies of each player grows linearly with the terminal nodes of the game tree and the payoff functions are sufficiently well-behaved to allow standard algorithms for equilibrium computation to work (cf. \cite{KMS1996}).

This paper addresses the second problem listed above. We construct index and degree theories \textit{for extensive-form games}. We label these theories ``extensive-form'' because, ultimately, the formulas we provide for computation of the index or degree of equilibria in an extensive-form game do not rely on perturbations in the whole normal-form payoff space, but only on the terminal payoffs of the game tree. 

Showing these results will require developing index and degree theories for \textit{polytope-form games}, which are games in which the strategy set of a player is a polytope and his payoff function is multiaffine\footnote{``Multiaffine'' will be understood as affine in each coordinate.} in the product of these polytopes.\footnote{These games have been studied by \cite{JM2004}, and were called ``strategic-form games''. In order to avoid confusion with current terminology, we opted for the term ``polytope-form games''.} Developing these theories is conceptually interesting also because it shows that the degree and index of equilibiria are invariant to disposing of \textit{all} strategically redundant data of the game. More precisely, when \textit{all} equivalent mixed strategies of a player are identified topologically, the resulting space of strategies of the player might not be a simplex anymore, but is a polytope. This identification process does not alter the degree or index of the equilibria. This is not a simple matter: once redundant strategies of the players are identified in the normal form of a game, the dimension of the payoff space of the players decreases. It is not clear \textit{a priori} that in the game resulting from the identification of redundant strategies, the equilibria identified as robust by the polytope-form index/degree theory will also remain robust in the original normal-form game: the space of perturbations of the latter is typically much larger than the former.

 By showing the equivalence between polytope- and normal-form index theories, we show that indeed no information regarding robustness to payoff perturbation is lost by restricting to a polytope form of the game. For example, we show how to define a polytope form of an extensive game using the sequence form of \cite{BS1996}, which appears as one of many possible  representations of the strategies in an extensive game and the result of a particular identification of the mixed strategy set of the players. 


The representation of normal- or extensive-form games in polytope form will provide new formulas for computation of the degree and index of equilibria, by an application of the same reasoning to the one outlined above for normal form. As in this paper, we are particularly interested in the computation of degrees/indices of equilibria in extensive games, we are especially interested in formulas that can be computed from perturbation of the terminal payoffs of the game tree directly. Our second main result addresses this question. In  \cite{GW2002}, the authors provide (for an arbitrary game tree with perfect recall) an alternative construction of the strategy set of the players which is essentially equivalent to the sequence form in \cite{BS1996}. They call these strategy sets \textit{enabling strategies}. Under a perturbation of the enabling strategy sets of the players, the authors prove a structure theorem for the graph of equilibria in enabling strategy sets and terminal payoffs of this game tree. Without going into too many details about their arguments, the structure theorem of Govindan and Wilson enables the construction of a degree theory and gives general formulas for the computation of this degree which involve only perturbations of the terminal payoffs of the game tree. This theory is however inherently incomplete: the structure theorem only holds under perturbed strategy sets and therefore the degree theory it implies is not really applicable to verifying robustness of equilibria in extensive games. We show that this incompleteness of the theory can be overcome, and we can produce formulas for the degree of an equilibrium that rely simply on perturbations of the terminal payoffs of the tree. 

This paper is organized as follows. Section \ref{introdeg} gives an introduction to the degree theory of equilibria in normal form, without referencing any machinery in algebraic topology. This allows an intuitive introduction to the main tools and ideas that will be discussed in the context of polytope-form games in the subsequent sections. Section \ref{polytope-form games} develops the theory of polytope-form games: we construct index and degree theories for polytope-form games, and establish their relations with the normal-form theories. The main result is Theorem \ref{PFNFequivalence}. Section \ref{extensiveform} recalls the sequence form representation of the extensive-form game and shows that it is a \textit{reduction} of the normal-form representation of the extensive-form game. The main result is Theorem \ref{stabilization} and its main implications: a degree theory of equilibria defined for extensive-form games together with formulas for the computation of the degree that rely on perturbations of the terminal payoffs only. The Appendix contains the formal definitions of the concepts of index and degree, stated using homology, and additional technical material which is necessary for some of the proofs. The omitted proofs of the main text are also located in the Appendix.

\section{Index and Degree in Normal form: An Introduction}\label{introdeg}

The index and degree theories of equilibria are in principle quite different: the degree derives from an analysis of the graph of the Nash correspondence, i.e., the correspondence which assigns to every game its equilibria, whereas the index comes from an analysis of fixed-point problems generated from the best-reply correspondence of the players. An open problem in normal-form games was whether degree and index were identical - solved affirmatively in \cite{DG2000}. They showed therefore that the differences between the two concepts were just apparent. Depending on the type of problem at hand, however, either the index or the degree might be a more suitable tool. For example, in counting problems related to the number of equilibria in generic finite games, the index of equilibria has been used multiple times fruitfully (cf. \cite{GPS1993} or, more recently, \cite{CS2020}). For problems of computation of the degree or index of an equilibrium component, the formulas for the degree - as the ones we will present in subsection \ref{GWsec} or subsection \ref{polydegcomputation} - are more tractable computationally. 

	In this section we recall the construction of the degree and index theories of equilibria in normal form. As mentioned in the introduction, it could be said that the degree is conceptually more complicated to grasp than the index, which is why our introduction to it is lenghtier than for the index.

\medskip
\subsection{Preliminary Definitions and Notation} Given $X$ a topological space and $U$ a subset of $X$, let $\text{cl}_X(U)$ denote the closure of $U$ in $X$. When the underlying topological space is understood, we omit the subscript $X$ and write cl$(U)$, only. We denote by $\Vert \cdot \Vert$ the usual Euclidean norm. Let $\N = \{1,...,N\}$ be the set of \textit{players}, which from now on is fixed. A tuple $\mathbb{G} = (\N, (S_n)_{n \in \N}, (\mathbb{G}_n)_{n \in \N})$ is a \textit{finite normal-form game}, where $S_n$ is the set of \textit{pure strategies} of player $n$, $\mathbb{G}_n: S \to \Re$ is the \textit{payoff function} of player $n$, where $S \equiv \times_n S_n$. The set of mixed strategies of player $n$ is denoted $\S_n$ and is identified with the unit simplex of $\Re^{S_n}$, where the canonical vector $e_{s_n}$ (i.e., the vector with $1$ in the $s_n$ coordinate) is identified with $s_n \in S_n$. The function $\mathbb{G}_n$ is extended to $\S \equiv \times_n \S_n$ as follows (for notational convenience, we denote the extension also by $\mathbb{G}_n$): 

\begin{equation} 
\mathbb{G}_n(\s) \equiv \mathbb{G}_n (\s_1,...,\s_N) \equiv \sum_{s \in S}\mathbb{G}_n(s)\prod_{m \in \N}\s_{s_m},
\end{equation}

where $\s \equiv (\s_m)_{m \in \N}$ and $\s_m \equiv (\s_{s_m})_{s_m \in S_m}$.

The set of Nash equilibria in a finite game is described by a finite system of polynomial equalities and inequalities with real coefficients and real variables and is therefore a semi-algebraic set. This implies that there are finitely many connected components of solutions to the system, i.e., finitely many connected components of equilibria. (cf. \cite{BCR2013}). Fixing $\S \equiv \times_n \S_n$, the set of payoff functions of player $n$ over $\S$ is identified with $\Re^{|S_1|...|S_N|}$ and the set of payoffs for all players is $\mathcal{P} \equiv \Re^{N(|S_1|...|S_N|)}$. Given a finite normal-form game $\mathbb{G}$, we denote for notational convenience the vector of payoffs $(\mathbb{G}_n)_{n \in \N}$ by $\mathbb{G}$. For a fixed $\S$, the \textit{Kohlberg-Mertens equilibrium graph} ($KM$-equilibrium graph) is the set $\mathcal{E}^{KM} = \{ (\s, \mathbb{G}) \in \S \times \mathcal{P} \mid \s \text{ is an Nash equilibrium of } \mathbb{G}\}$. Recall that the topological space given by the one-point compactification of $\P$ is a sphere $\mathbb{S}$ of dimension equal to the dimension of $\P$.

\subsection{Index and Degree Theories in Normal form}

We start with an exposition on the degree theory of equilibria in normal form in subsection \ref{degtheory}. Right after, we define the theory of index of equilibria in normal form.  

\subsubsection{Structure Theorem, Degree Theeory and Robustness of equilibria}\label{degtheory}Let proj$: \mathcal{E}^{KM} \to \mathcal{P}$ be the map defined by $(\s, \mathbb{G}) \mapsto \mathbb{G}$. An equilibrium $\s \in \S$ of a game $\mathbb{G}$ obviously satisfies the following equation: 

\begin{equation}\label{eq graph}
\text{proj}(\s, \mathbb{G}) = \mathbb{G}
\end{equation}

Asking whether $\s$ is robust to small perturbations of $\mathbb{G}$ amounts to asking whether equation \eqref{eq graph} has a solution $\s'$ close $\s$, when $\mathbb{G}$ is perturbed to a sufficiently close $\mG'$. In other words, $\s$ is \textit{payoff-robust}\footnote{The terminology \textit{essential} has also been in the literature (cf. \cite{WJ1962}) in place of of \textit{payoff-robust}.} if for any $\varepsilon>0$, there exists $\d>0$ such that if $\Vert \mG - \mG' \Vert < \d$, there exist $\s'$ an equilibrium of $\mG'$ with $\Vert \s - \s' \Vert < \varepsilon$. 

The $KM$-structure theorem allows us to define a ($KM$)-\textit{degree theory of equilibria} which is a tool to identify which equilibria are robust to payoff perturbations. The $KM$-structure theorem has two parts: in the first part, a homeomorphism $\theta^{KM}: \mathcal{E}^{KM} \to \mathcal{P}$ is explicitly constructed. In the second part, proj$\circ (\theta^{KM})^{-1}: \mathcal{P} \to \mathcal{P}$ is shown to be (linearly) homotopic to the identity map $id_{\mathcal{P}}$ on $\P$, by a homotopy that extends to the one-point compactification $\mathbb{S}$ of $\mathcal{P}$. We explain how these two parts of the theorem play a role in defining a degree theory and ultimately help us in identifying robust equilibria.


The $KM$-homeomorphism is defined as follows: given $(\s, \mG) \in \mathcal{E}^{KM}$, the vector $\mG_n \in \Re^{|S_1|....|S_N|}$ is orthogonally decomposed as $\tilde \mG_n \oplus g_n$:  $\tilde \mG_n$ satisfies for each $s_n \in S_n$, $\sum_{t_{-n} \in S_{-n}}\tilde \mG_{n}(s_n, t_{-n}) =0$  and $g_n \in \Re^{|S_n|}$ lies is the orthogonal complement to $\tilde{\mathbb{G}}_n$. Then, for each $s \in S$

\begin{equation}
\theta^{KM}_{n}(\s, \mG) = \theta^{KM}_n(\s, \tilde \mG, g) = (\tilde \mG_n, z_n), \text{where}
\end{equation}

$g \equiv (g_n)_{n \in \N}$, $\mG_n(\s_{-n}) \equiv (\mG_n(s_n, \s_{-n}))_{s_n \in S_n}$, and  $z_n =\ \s_n + \mG_n(\s_{-n})$. Let $\theta^{KM} \equiv \times_{n \in \N}\theta^{KM}_n$. The inverse homeomorphism $(\theta^{KM})^{-1}: \mathcal{P} \to \E$ is then:

\begin{equation}
(\theta^{KM})^{-1}_n(\tilde \mG_n, z_n) = (\tilde \mG_n \oplus z_n - r_n(z_n) - \tilde \mG_n(r_{-n}(z_{-n})), r_n(z_n))
\end{equation}

 where $r_n: \Re^{S_n} \to \S_n$ is the nearest-point projection and $r_{-n} \equiv \times_{m \neq n}r_{m}$. 

Equation \eqref{eq graph} can then be rewritten as:

\begin{equation}\label{z-graph}
\text{proj} \circ (\theta^{KM})^{-1}(\tilde \mG, z) = \mathbb{G}.
\end{equation}

Equation \eqref{eq graph} is now rewritten in \eqref{z-graph} as an equation of a map defined from $\P$ to itself. A solution $z = (z_n)_{n \in \N}$ of the equation \eqref{z-graph} is such that $r(z) \equiv \times_n r_n(z_n)$ is an equilibrium of $\mG$. Evidently, payoff-robustness of an equilibrium $\s$ of $\mG$ with respect to payoff-perturbations from $\mG$ can now be rewritten in terms of robustness of solutions $z$ to small perturbations of $\mG$. Let us assume that $z_0$ is a solution of equation \eqref{z-graph}, that $\text{proj} \circ (\theta^{KM})^{-1}$ is smooth in an open neighborhood $U$ of $(\tilde \mG, z_0)$ and that the determinant Jacobian $\partial (\text{proj} \circ (\theta^{KM})^{-1})(\tilde \mG, z_0)$ of $\text{proj} \circ (\theta^{KM})^{-1}$ at $(\tilde \mG, z_0)$ is different from $0$. This guarantees the solution $z_0$ is isolated. The inverse function theorem implies that there exists an open neighborhood $W \subset U$ of $(\tilde \mG, z_0)$ and $O$ a neighborhood of $\text{proj} \circ (\theta^{KM})^{-1}(\tilde \mG, z_0)$ such that $\text{proj} \circ (\theta^{KM})^{-1}: W \to O$ is a diffeomorphism. It is now clear that for each $\e>0$, there exists $\d>0$ such that if $\mG' \in O$, $\Vert \mG - \mG' \Vert < \d$ there exists $(\tilde \mG', z') \in W$ satisfying \eqref{z-graph} with $\Vert (\tilde \mG, z) - (\tilde \mG',z') \Vert < \e$. Therefore, $\text{sign}(\partial(\text{proj} \circ (\theta^{KM})^{-1})(\tilde \mG, z_0)) \neq 0$ implies that $r(z_0)$ is a payoff-robust equilibrium of $\mG$.


Though $(\theta^{KM})^{-1}$ is not a smooth map, it can be shown that there exists an open and dense set $\mathcal{Q} \subset \mathcal{P}$ such that if $\mG \in \mathcal{Q}$, then equation \eqref{z-graph} has finitely many solutions in $z$ (cf. \cite{JH1973} or \cite{GW2003}), around each of which $\text{proj} \circ (\theta^{KM})^{-1}$ is smooth and on each of which $\text{proj} \circ (\theta^{KM})^{-1}$ has invertible determinant Jacobian. In this case,  the \textit{degree of a solution $z_0$ (of \ref{z-graph}) w.r.t. $\mathbb{G}$} is defined as $\text{sign}(\partial(\text{proj} \circ (\theta^{KM})^{-1})(\tilde \mG, z_0)) \in \{-1, +1\}$. Since $\s_0 \equiv r(z_0)$ is an equilibrium of $\mG$, the \textit{degree of $\s_0$ w.r.t. $\mG$}, denoted deg$_{\mG}(\s_0)$, is defined as $\text{sign}(\partial(\text{proj} \circ (\theta^{KM})^{-1})(\tilde \mG, z_0))$. The \textit{total degree of $\mG$} is the sum of the degree of all solutions. Intuitively, the total degree counts the number of solutions to the equation \eqref{z-graph} for a fixed $\mG$, together with their ``multiplicities'', that is, the signs of $+1$ or $-1$ associated to the determinant Jacobian at the solution. 

The definition of the degree of an equilibrium presented in the previous paragraph applies only to games which have isolated equilibria and where the Jacobian of proj$\circ (\theta^{KM})^{-1}$ around a solution is smooth and invertible. But in general games might have a continuum of equilibria in mixed strategies, a particularly prevalent feature in extensive-form games. Therefore, not all normal-form games have payoffs which belong to $\mathcal{Q}$. There are, however, finitely many connected components of solutions to \eqref{z-graph}, as there are finitely many connected components of equilibria for any finite normal-form game. The definition of the degree must then be refined to account for this situation. For games in $\mathcal{Q}$, all equilibria are isolated and therefore the degree is assigned to each of these isolated points. When there are nondegenerate connected components of equilibria, the degree is assigned to components of equilibria. The next proposition tells us how to do so. Denote $f^{KM} \equiv \text{proj} \circ (\theta^{KM})^{-1}$.

\begin{proposition}\label{degnormalform}
Let $X$ be an equilibrium component of a normal-form game $\mG$ and let $X_z \equiv \theta^{KM}(X, \mG)$. There exists $d \in \mathbb{Z}$, such that for any open neighborhood $U$ of $X_z$ with $\text{cl}(U) \cap \{ (\tilde \mG, z) \in \P \mid f^{KM}(\tilde \mG, z) = \mG \} = X_z$, there exists $\bar \e>0$  such that for each $\e \in (0,\bar \e)$ and $\mG' \in \mathcal{Q}$ satisfying $\Vert \mG' - \mG \Vert < \e$, 
\begin{equation}
\sum_{(\tilde \mG', z') \in U : f^{KM}(\tilde \mG', z') = \mG'} \text{sign}(\partial(\text{proj} \circ (\theta^{KM})^{-1})(\tilde \mG', z')) = d.
\end{equation} 
\end{proposition}

\begin{proof} The map $\text{proj} \circ (\theta^{KM})^{-1})|_{U}$ is proper over a ball $B$ around $\mG$ with sufficiently small radius. Proposition 5.12 in Chapter IV of \cite{D1972} implies that the local degree is constant for each $\mG' \in B$. For each $\mG' \in \mathcal{Q} \cap B$, the degree of $\text{proj} \circ (\theta^{KM})^{-1})|_{U}$ over $\mG'$ is constant in $\mG' \in B$ and equals the $\sum_{(\tilde \mG', z') \in U : f^{KM}(\tilde \mG', z') = \mG'} \text{sign}(\partial(\text{proj} \circ (\theta^{KM})^{-1})(\tilde \mG', z'))$.\end{proof}

Proposition \ref{degnormalform} tells us that the number of solutions $z$ of \eqref{z-graph} for a game $\mG'$ which is sufficiently close to $\mG$, counted with their multiplicities, is a constant $d$.	We can therefore define the \textit{degree of the component $X$ w.r.t. $\mG$} as the integer $d$, and denote it by $\text{deg}_{\mG}(X)$. Similarly to what happens in the generic case, if $\text{deg}_{\mG}(X) \neq 0$, then for any $\varepsilon>0$, there exists $\d>0$ such that for any game $\mG'$ with $\Vert \mG' - \mG \Vert < \d$, there exists an equilibrium $\s'$ of $\mG'$ with $d(X,\s')< \varepsilon$, where $d(\cdot,\cdot)$ is the set distance.  

The second part of the $KM$-structure theorem constructs a linear homotopy between $\text{proj} \circ (\theta^{KM})^{-1}$ and $\text{id}_{\mathcal{P}}$ and shows this homotopy can be extended to the one-point compactification $\mathbb{S}$ of $\mathcal{P}$. This second part implies an additional consequence for the degree of equilibrium components: \textit{the sum of degrees of the equilibrium components of a finite game $\mG$ is $+1$}. This allows us to say immediately that, for games $\mG \in \mathcal{Q}$, the number of equilibrium components must be odd, since their individual degrees are either $+1$ or $-1$, and their sum must be $+1$. Explaining precisely the reason for this implication requires a deeper dive into more demanding topological machinery. Indeed, a definition of the degree of a component of equilibria can be given directly using (singular) homology and the properties of the degree highlighted in this subsection follow seamlessly from this definition. In subsection \ref{degreetheory} of the Appendix, we present a short introduction to normal-form degree theory using this topological machinery, and highlight these useful properties.

\subsubsection{Index of Equilibria}\label{indextheory}In fixed point theory, the index of fixed points contains information about their robustness when the map is perturbed to a nearby map. (See \cite{D1972}, Ch. VII for an account of index theory) Since Nash equilibria are obtainable as fixed points, index theory applies directly to them. When an equilibrium is robust to perturbation of its associated map, then it is in particular robust to payoff perturbations (cf. \cite{KR1994}), which is why the index has an immediate interest to game theory. 

For simplicity, suppose $f: U \to \S$ is a differentiable map defined on a neighborhood $U$ of $\S$ in $\Re^{\sum_n |S_n|}$ and such that the fixed points of $f$ are the Nash equilibria of a game $\mG$. Let $d_{f}$ be the displacement of $f$, i.e., $d_{f}(\s) = \s - f(\s)$. Then the Nash equilibria of $\mG$ are the zeros of $d_{f}$.  Suppose now that the Jacobian of $d_f$ at a Nash equilibrium $\s$ of $\mG$ is nonsingular. Then we can define the {\it index} of $\s$ under $f$ as $\pm 1$ depending on whether the determinant of the Jacobian of $d_{f}$ is positive or negative.  As it happens with the degree of equilibria, we can obtain a definition of the index of a component of equilibria by considering a perturbation of the displacement that is differentiable and has finitely many zeros: the indices of the isolated zeros of any sufficiently small perturbation of the initial displacement sum to the same constant, which can be defined as the index of the component. As we did with the degree, the general definition of index of equilibria we will operate with in this paper will be done using singular homology and can be found in subsection \ref{indextheorygeneral} in the Appendix. 

A potential problem with the definition of index is the dependence of the definition on the function $f$. For a game $\mG$ there are many maps from $\S$ to $\S$ whose fixed points are the equilibria of game $\mG$ and one would like to know whether the index of each component of equilibria is the same under each map. Under some regularity assumptions on $f$, we can show that the index is independent of $f$. Specifically, consider the class of  continuous maps $f: \P \times \S \to \S$ with the property that  the fixed points of the restriction of $f$ to $\{\, \mG \, \} \times \S$ are the equilibria of $\mG$.  \cite{DG2000} show that the index of equilibria is independent of the particular map in this class that is used to compute it;  \cite{GW1997a} show that the degree is equivalent to the index computed using one of the maps in this class, the fixed-point map defined by \cite{GPS1993}. Thus, \textit{the index and degree of equilibria coincide}---see \cite{DG2000} for an alternate, more direct, proof of this equivalence.

\section{Polytope-form Games}\label{polytope-form games}

In this section, we introduce the notion of \textit{polytope-form games}, which are games in which the strategy sets of the players are polytopes. Examples throughout this section motivate the use of this notion. We introduce the notion of \textit{reduction} of a polytope-form game and prove a series of results about it, which will later be important in showing invariance properties of degree and index theories for polytope-form games.

\subsection{Polytope-form games and reductions}For each $n \in \N$, let $P_n$ be a polytope in some finite dimensional real vector space and denote $P \equiv \times_{n \in \N}P_n$. For each $n \in \N$, let $V_n: P \to \Re$ be an affine function in each coordinate $P_n$. The tuple $V = (\N, (P_n)_{n \in \N}, (V_n)_{n \in \N})$ defines a \textit{polytope-form game}, where $P_n$ is the \textit{strategy set} of player $n$ and $V_n$ the \textit{payoff function} of player $n \in \N$. Nash equilibrium is defined in the exact same fashion as for normal-form games, i.e., $p \in P$ is a \textit{Nash equilibrium} of $V$ if for each $n \in \N$ and each $t_n \in P_n$, $V_n(p) \geq V_n(t_n, p_{-n})$. We denote by $E(V)$ the set of equilibria of a polytope-form game $V$. Any normal-form game is evidently a polytope-form game, where the strategy set of each player is the unit simplex of some Euclidean space. 

\medskip
Let $V = (\N, (P_n)_{n \in \N}, (V_n)_{n \in \N})$ be a polytope-form game and, for each $n \in \N$, let $\bar P_n$ be a polytope. Let $q_n: P_n \to \bar P_n$ be an affine and surjective map satisfying the following condition:

\begin{equation}\label{identcondition}
 \forall p_n, p'_n \in P_n, q_n(p_n) = q_n(p'_n) \implies \forall m \in \N, \forall t_{-n} \in P_{-n}, V_m(p_n, t_{-n}) = V_{m}(p'_n, t_{-n}).
\end{equation}

Note that because $q_n$ is affine and surjective, it admits a right-inverse $j_n: \bar P_n \to P_n$, i.e., $q_n \circ j_n = id_{\bar P_n}$. Let $j \equiv \times_{n} j_n$. Given a polytope-form game $V$ and a map $q \equiv \times_{n \in \N} q_n$ whose coordinate maps satisfy \ref{identcondition}, we define the $q$-\textit{reduction $\bar V$ of $V$} as the polytope-form game $\bar V = (\N, (\bar P_n)_{n \in \N}, (\bar V_n)_{n \in \N})$, where $\bar V_n \circ q (p) = V_n(p)$. The map $q$ is called a \textit{reduction map of} $V$.  The reduction map $q$ defines a unique $q$-reduction of $V$: given $\bar p \in \bar P$, there exists some $p \in P$ such that $q(p) = \bar p$, and we can set, for each $n \in \N$, $\bar V_n(\bar p) \equiv V_n(p)$. By equation \eqref{identcondition}, $V_n$ is constant in the fibers of $q$ and so $\bar V_n$ is uniquely defined. 

\begin{definition}\label{reduction}
Two polytope-form games $V = (\N, (P_n)_{n \in \N}, (V_n)_{n \in \N})$ and $\bar V = (\N, (\bar P_n)_{n \in \N}, (\bar V_n)_{n \in \N})$ are \textit{equivalent} if there exist two reduction maps $q: P \to P'$ and $\bar q: \bar P \to P'$ and a polytope-form game $V' = (\N, (P'_n)_{n \in \N}, (V'_n)_{n \in \N})$ such that for each player $n$, $V'_n \circ q = V_n$ and $V'_n \circ \bar q = \bar V_n$.
\end{definition}

Put differently, Definition \ref{reduction} states that the two polytope-form games $V$ and $\bar V$ are equivalent if they have a common reduction $V' = (\N, (P'_n)_{n \in \N}, (V'_n)_{n \in \N})$.

\begin{example}For a normal-form game $\mG = (\N, (S_n)_{n \in \N}, (\mG_n)_{n \in \N})$, the reduced normal form $\mG^{r} = (\N, (S^r_n)_{n \in \N}, (\mG^r_n)_{n \in \N})$ (cf. \cite{KM1986}), where for all $n \in \N, S^r_n \subset S_n$, is a reduction of $\mG$. Define the reduction map $q^r_n: \S_n \to \S^r_n$, where $\S^r_n$ is the mixed strategy set of player $n$ in $\mathbb{G}^r$ and $\S_n$ the mixed strategy set of player $n$ in $\mathbb{G}$: let $q^r_n(s_n) \equiv s_n$, if $s_n \in S^r_n$. For $s_n \notin S^r_n$, there exists $\s^r_n \in \S^r_n$, such that $\forall m \in \N, \forall s_{-n} \in S_{-n}, \mathbb{G}_m(s_n, s_{-n}) = \mathbb{G}_{m}(\s^r_n, s_{-n})$. Define then $q^r_n(s_n) \equiv \s^r_n$. Extend the map $q^r_n$ to $\S_n$ by linear interpolation to obtain a reduction map. It follows immediately that for each $n \in \N$, $\mathbb{G}^r_n \circ q^r = \mathbb{G}_n$.\end{example}


\begin{remark}Definition \ref{reduction} indeed gives rise to an equivalence relation between polytope-form games: define $V \sim \bar V$ iff $V$ and $\bar V$ have a common reduction $V'$. The relation $\sim$ is an equivalence relation among polytope-form games. Symmetry and reflexivity are immediate from the definition. The only remaining property to check is transitivity: let $V \sim \bar V$  and $\bar V \sim \tilde V$. We claim that $V \sim \tilde V$.  Let $V'$ be the common reduction of $V$ and $\bar V$. Let $q$ be the reduction map from $V$ to $V'$ and $j$ the right-inverse of the reduction map from $\bar V$ to $V'$. Let $\tilde V'$ be the common reduction from $\bar V$ and $\tilde V$ and let $ q'$ be the reduction map from $\bar V$ to $\tilde V'$. Consider the map $\tilde q = q' \circ j \circ q$. The map $\tilde q$ is a reduction map from $V$ to $\tilde V'$, which shows that $V$ and $V'$ have a common reduction, i.e., $V \sim V'$. \end{remark}

\subsection{The reduced polytope form}Let $V = (\N, (P_n)_{n \in \N}, (V_n)_{n \in \N})$ be a polytope-form game and $\{v^n_1,...,v^n_{k_n}\}$ the set of vertices of $P_n$. Let $S_n = \{e_{1},...,e_{k_n}\}$ be the canonical basis of $\Re^{k_n}$. For each $n \in \N$, fix $h_n: S_n \to \{v_1,..., v_{k_n}\}$ a bijection. The map $h_n$ defines uniquely an affine and surjective map $q_n: \S_n \to P_n$, where $\S_n \equiv \D(S_n)$, such that for each $j \in \{1,..,k_n\}, q_n(e_j) = h_n(e_j)$.
Let $\mathbb{G} = (\N, (S_n)_{n \in \N}, (\mathbb{G}_n)_{n \in \N})$ be the normal-form game where payoff functions are defined by $\mathbb{G}_n (s_1,...,s_N) \equiv V_n (q_1(s_1),...,q_N(s_N))$. Clearly, $V$ is a reduction of $\mathbb{G}$. 

\begin{remark}We observe that when defining $\mathbb{G}$ from $V$ above, the only arbitrary choice made is the map $h_n$, as there are multiple ways of mapping the set $S_n$ to the vertices of $P_n$. Any two normal-form games $\mathbb{G}$ and $\mathbb{G}'$ obtained from from $V$ by considering two distinct maps $h_n$ and $h'_n$, respectively, are identical, modulo a relabeling of the pure strategies of the players. Therefore, up to relabeling, $\mathbb{G}$ and $\mathbb{G}'$ are the same game. \end{remark}

Each payoff function $\mathbb{G}_n$ can be uniquely extended to a multilinear function on $\times_n \mathbb{R}^{S_n}$, thus we still denote by $\mG_n$ the extension. We define an equivalence relation $\approx$ between points in $\mathbb{R}^{S_n}$: 

\begin{equation}\label{redeq}
x_n, y_n \in \mathbb{R}^{S_n}, x_n \approx y_n   \iff \forall m \in \N,  \forall s_{-n} \in S_{-n},  \mathbb{G}_m(x_n, s_{-n}) = \mathbb{G}_m(y_n, s_{-n}). 
\end{equation}

Let $\mathbb{R}^{S_n}/\approx$ be the set of equivalence classes given by equivalence relation $\approx$. We endow this set with the quotient topology given by this equivalence relation. Since for each $m \in \N$, $\mathbb{G}_m$ is multilinear, the quotient space $\mathbb{R}^{S_n}/\approx$ can be given a real vector space structure from operations in $\mathbb{R}^{S_n}$: for addition, if $[x_n], [y_n] \in \mathbb{R}^{S_n}/\approx$, define $[x_n] \oplus [y_n] := [x_n  + y_n]$; for scalar multiplication, let $\alpha \in \mathbb{R}$ and define $\alpha[x_n] := [\alpha x_n]$.  Because $\mathbb{R}^{S_n}$ is finite dimensional, so is $\mathbb{R}^{S_n}/\approx$. Therefore, $\mathbb{R}^{S_n}/\approx$ is a finite-dimensional vector space  and we denote it by $J_n \equiv \mathbb{R}^{S_n}/\approx$.  

Consider now the partition mapping $\pi^r_n: \mathbb{R}^{S_n} \rightarrow J_n$ from $\approx$ given by  $x_n \mapsto [x_n]$ and its restriction to $\Sigma_n$ denoted $\pi^r_n|_{\Sigma_n}$. By construction, $\pi^r_n|_{\Sigma_n}$ is an affine and surjective map which implies that $\pi^r_n(\Sigma_n)$ is a polytope. The map $\pi^r|_{\S^r}$ is called the \textit{maximal reduction map of $V$} and is obviously a reduction map. For notational convenience we will denote $\pi^r|_{\S^r}$ by $\pi^r$. The polytope-form game $V^r = (\N, (P^r_n)_{n \in \N}, (V^r_n)_{n \in \N})$ associated to this map is a reduction of $\mathbb{G}$, which is uniquely defined from $V$ (modulo the the labeling map $h_n$ defined in the previous paragraph). We call $V^r$ the \textit{reduced polytope form of $V$}.

\begin{ex}\label{running example} In Figure \ref{first game tree} we depict an extensive-form game $G$ and show that the strategy polytope of player 2 obtained from the identification of his equivalent mixed strategies is not a simplex.

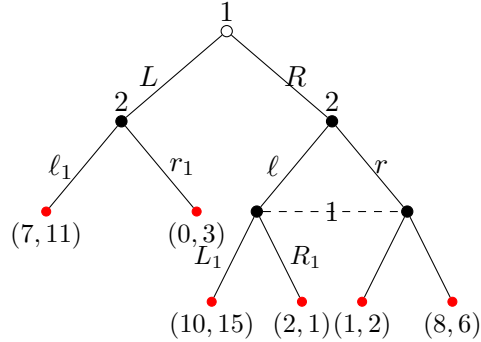
\begin{figure}[t]
\centering{}%
\caption{Extensive-form game $G$}
\label{first game tree}
\tikzstyle{hollow node}=[circle,draw,inner sep=1.5]
                         \tikzstyle{solid node}=[circle,draw,inner sep=1.5,fill=black]
                         \tikzset{
                           red node/.style={circle,draw=red,fill=red,inner sep=1.2},
                           blue node/.style={rectangle,draw=blue,inner sep=2.5}
                         }

\begin{tikzpicture}[scale = 0.8]
 
\tikzstyle{level 1}=[level distance=15mm,sibling distance=35mm]
    \tikzstyle{level 2}=[level distance=15mm,sibling distance=25mm]
 		\tikzstyle{level 3}=[level distance=15mm,sibling distance=15mm]
 
  \node(0)[hollow node]{}
     child{node (0-1)[solid node]{}
     child{node (1-1)[red node]{} edge from parent node[left]{$\ell_1$}}
     child{node (1-2)[red node]{} edge from parent node[right]{$r_1$}}
	edge from parent node[left,xshift=-2]{$L$}    
     }    
     child{node(0-2)[solid node]{}     
       child{node(2-0)[solid node]{} 
       	child{node(2-1)[red node]{} edge from parent node[left]{\begin{small}$L_1$\end{small}}}
       	child{node(2-2)[red node]{} edge from parent node[right]{\begin{small}$R_1$\end{small}}}
		edge from parent node[left,xshift=-2]{$\ell$}      
       }
       child{node(3-1)[solid node]{}
       	child{node(3-2)[red node]{}}
       	child{node(3-3)[red node]{}}
		edge from parent node[right,xshift=-2]{$r$}      	
      	}
     edge from parent node[right,xshift=-2]{$R$} 	
     };
     \draw[dashed](2-0)to(3-1);
     \node[above] at (0){$1$};
     \node[above] at (0-1){$2$}; 
	 \node at($(2-0)!.5!(3-1)$){$1$};
	 \node[above] at (0-2){$2$};
	 \node[below] at (2-1){\begin{small}$(10,15)$\end{small}};
	 \node[below] at (2-2){\begin{small}$(2,1)$\end{small}};
	 \node[below] at (3-2){\begin{small}$(1, 2)$\end{small}};
	 \node[below] at (3-3){\begin{small}$(8, 6)$\end{small}};
	 \node[below] at (1-1){\begin{small}$(7, 11)$\end{small}};
	  \node[below] at (1-2){\begin{small}$(0, 3)$\end{small}};

\end{tikzpicture}
\end{figure}

\medskip
\medskip

\begin{table}[b] 
\centering{}%
\caption{Normal-form $\mG$ game of $G$}
\setlength{\extrarowheight}{2pt}
\scalebox{0.85}{
\begin{tabular}{cc|c|c|c|c|}
  & \multicolumn{1}{c}{} & \multicolumn{4}{c}{$2$} \\
  & \multicolumn{1}{c}{} & \multicolumn{1}{c}{$\ell_1\ell$}  & \multicolumn{1}{c}{$\ell_1 r$}  & \multicolumn{1}{c}{$r_1 \ell$} & \multicolumn{1}{c}{$r_1 r$} \\ \cline{3-6}
           		 & $LL_1$ & $7, 11$ & $7, 11$ & $0, 3$  & $0, 3$ \\ \cline{3-6}
$1$  & $LR_1$ & $7, 11$ & $7, 11$ & $0, 3$ & $0, 3$ \\ \cline{3-6}
           		 & $RL_1$ & $10,15$ & $1, 2$ & $10,15$ & $1, 2$ \\ \cline{3-6}
				 & $RR_1$ & $2,1$ & $8, 6$ & $2,1$ & $8, 6$ \\ \cline{3-6}

\end{tabular}}
\end{table}

We can identify the pure strategies $LL_1$ with $LR_1$ of player $1$, since they are equivalent, which gives us the reduced normal form of Table \ref{first ex normal form}. In the reduced normal form, the strategy sets of the players are still simplices. Now, notice that the equal mixture of pure strategies $\ell_1r$ and $r_1\ell$ of player $2$ is equivalent to the equal mixture of $\ell_1 \ell$ and $r_1r$. When identified, these mixed strategies give rise to a parallelogram, as illustrated in Figure \ref{FIG:identification}. 

\begin{table}[t]
\centering{}%
\caption{Reduced normal form of $\mG$}
\label{first ex normal form}
\setlength{\extrarowheight}{2pt}
\scalebox{0.8}{
\begin{tabular}{cc|c|c|c|c|}
  & \multicolumn{1}{c}{} & \multicolumn{4}{c}{$II$} \\
  & \multicolumn{1}{c}{} & \multicolumn{1}{c}{$\ell_1\ell$}  & \multicolumn{1}{c}{$\ell_1 r$}  & \multicolumn{1}{c}{$r_1 \ell$} & \multicolumn{1}{c}{$r_1 r$} \\ \cline{3-6}
 				 & $LR_1$ &  $7, 11$ & $7, 11$ & $0, 3$ & $0, 3$ \\ \cline{3-6}
           $I$ & $RL_1$ & $10,15$ & $1, 2$ & $10, 15$ & $1, 2$ \\ \cline{3-6}
				 & $RR_1$ & $2,1$ & $8, 6$ & $2,1$ & $8, 6$ \\ \cline{3-6}

\end{tabular}}
\end{table}

\begin{figure}[t]
\centering{}%
\caption{Identification}
\label{FIG:identification}
\medskip
\includegraphics[scale=0.6]{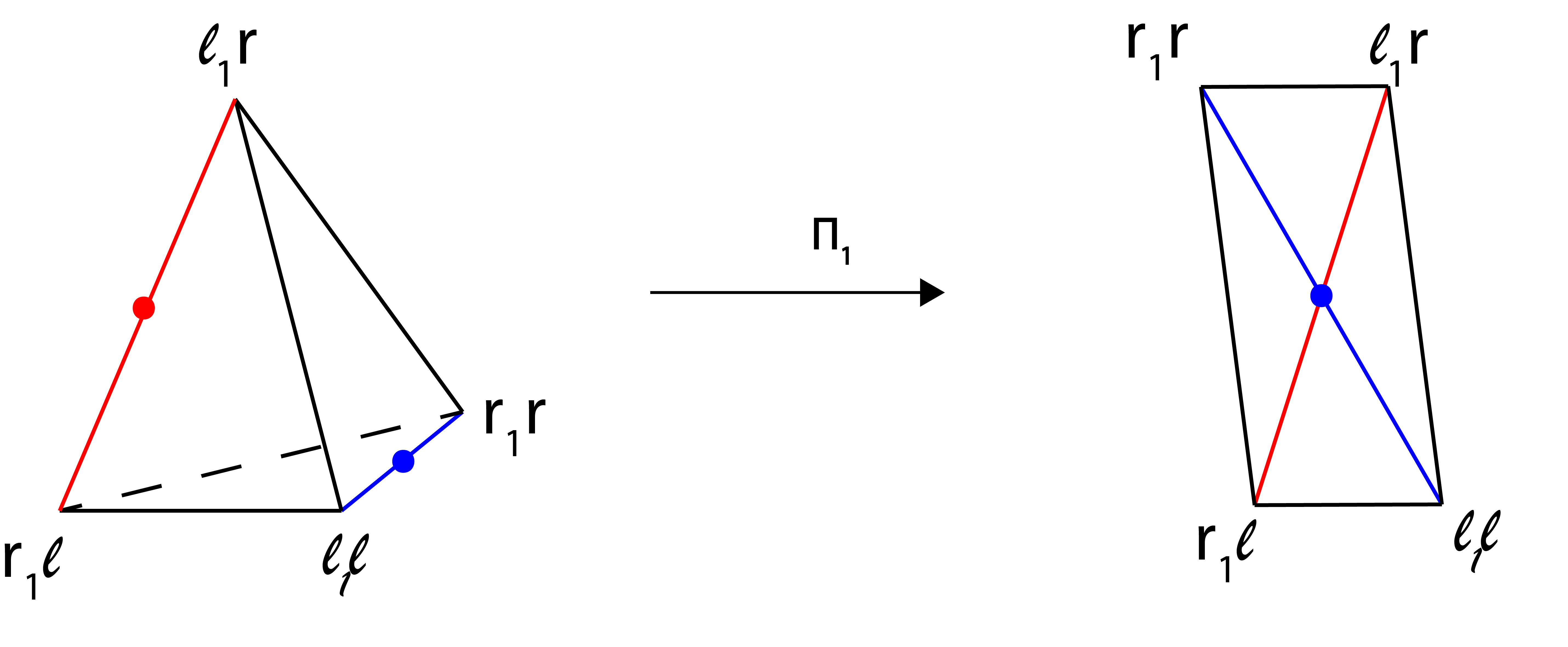}

\end{figure}

\end{ex}

\begin{proposition}\label{isoreduc}Let $V = (\N, (P_n)_{n \in \N}, (V_n)_{n \in \N})$ and $\bar V = (\N, (\bar P_n)_{n \in \N}, (\bar V_n)_{n \in \N})$ be two equivalent polytope-form games, and let $V^r = (\N, (P^r_n)_{n \in \N}, (V^r_n)_{n \in \N})$ and $\bar V^r = (\N, (\bar P^r_n)_{n \in \N}, (\bar V^r_n)_{n \in \N})$ be reduced polytope-form games of $V$ and $\bar V$ respectively. Then, for each $n \in \N$, $P^r_n$ is affinely isomorphic to $\bar P^r_n.$\end{proposition}

\begin{proof}Let $\pi^r$ be the maximal reduction map from $V$ and $\bar \pi^r$ the maximal reduction map from $\bar V$. Let $h: \S \to P$ be the reduction map from the normal-form game $\mathbb{G}$ to $V$ and, analogously, let $\bar h: \S \to \bar P$ be the reduction map from the normal-form game $\bar{\mathbb{G}}$ to $\bar{V}$, as defined in the begining of this subsection. Let $\tilde V$ be the common reduction of $V$ (with reduction map $q$) and $\bar V$ (with reduction map $\bar q$). Note that from equations \eqref{identcondition} and \eqref{redeq}, for each $n \in \N$, $\pi^r_n$ is constant in the fibers of $q_n \circ h_n$. So we can define  $\tilde \pi^r_n: \tilde P_n \to P^{r}_n$  by $\tilde \pi^r(\tilde p) \equiv \pi^r(\s)$, for all $\s \in \S$ with $(q \circ h) (\s) = \tilde p$.  The map $\tilde \pi^r_n(\cdot)$ is the unique map that satisfies $\pi^r_n = \tilde \pi^r_n \circ (q_n \circ h_n)$, and is obviously affine and surjective. Similarly, for each $n \in \N$, there exists a unique map $\tilde{\bar \pi}^r_n : \bar P_n \to \bar P^r_n$ that satisfies $\bar \pi^r_n = \tilde{\bar{\pi}}^r_n \circ (\bar q_n \circ \bar h_n)$, which implies $\tilde{\bar{\pi}}^r_n$ is affine and surjective. We claim that for each $\tilde p, \tilde p' \in \tilde P$, the following holds:  

\begin{equation} 
\tilde{\bar \pi}^r(\tilde p) = \tilde{\bar \pi}^r(\tilde p') \iff \tilde{\pi}^r(\tilde p) =  \tilde{\pi}^r(\tilde p').
\end{equation}  

We assume the claim for now just to conclude the proof, and provide a proof of the claim right after. From this claim, we have that for each $n \in \N$, there exists a unique (affine and surjective) map $\bar g^r_n: \bar P^{r}_n \to  P^r_n$ such that the following equation is satisfied: $\tilde \pi^r_n = \bar g^r_n \circ \tilde{\bar{\pi}}^r_n$. Similarly, there exists an affine and surjective map $g^r_n: P^r_n \to \bar P^r_n$ such that $\tilde{\bar{\pi}}^r_n = g^r_n \circ \tilde{\pi}^r_n$. Note that $g^r_n \circ (\bar g^r_n \circ \tilde{\bar{\pi}}^r_n) = g^r_n \circ \tilde{\pi}^r_n = \tilde{\bar{\pi}}^r_n$. Therefore, $g^r_n \circ \bar g^r_n = \text{id}_{\bar P^r_n}$. By the same reasoning, $\bar g^r_n \circ  g^r_n = \text{id}_{P^r_n}$. Therefore, $P^r_n$ and $\bar P^r_n$ are affinely isomorphic. 


We now prove the claim. Let $\tilde p$ and $\tilde p'$ be elements of $\tilde P$ and fix $\s$ and $\s'$ in $\S$ such that $(q \circ h)(\s) = \tilde p$ and $(q \circ h)(\s') = \tilde p'$. Note that $\tilde{\pi}^r(\tilde p) =  \tilde{\pi}^r(\tilde p')$ if and only if $\tilde{\pi}^r((q \circ h)(\s)) =  \tilde{\pi}^r((q \circ h)(\s'))$ if and only if $\pi^r(\s) =  \pi^r(\s')$ if and only if for each $n, m \in \N, s_{-n} \in S_{-n}, \mG_m(\s_n, s_{-n}) = \mG_m(\s'_n, s_{-n})$ if and only if for each $n, m \in \N, s_{-n} \in S_{-n}, \tilde V_m((q_n \circ h_n)(\s_n), (q_{-n} \circ h_{-n})(s_{-n})) = \tilde V_m((q_n \circ h_n)(\s'_n), (q_{-n} \circ h_{-n})(s_{-n}))$ if and only if for each $n, m \in \N, \tilde{t}_{-n} \in \tilde P_{-n}, \tilde V_m(\tilde p_n, \tilde{t}_{-n}) = \tilde V_m(\tilde p'_n, \tilde{t}_{-n})$ if and only if $\tilde{\bar{\pi}}^r(\tilde p) =  \tilde{\bar{\pi}}^r(\tilde p')$, which concludes the proof.\end{proof}
 
Proposition \ref{isoreduc} tells us that equivalent polytope-form games have ``isomorphic'' reduced polytope forms, in the sense that the polytope-form strategy sets of the players are affinely isomorphic. Therefore, reduced polytope forms of two equivalent polytope-form games are unique (up to isomorphism).

\subsection{Polytope-form Degree Theory}\label{polytope-form degree} In order to define a degree theory for polytope-form games, we follow the same path adopted for normal form: we first establish a structure theorem for polytope-form games. We start with some preliminary definitions. 
For each $n \in \N$, fix $P_n \subset \Re^{d_n}$ a polytope and let $P \equiv \times_{n}P_n$. Let $(\Delta_n)$ be the affine space generated by the unit simplex $\Delta_n$ in $\mathbb{R}^{d_n}$. The polytope $P_n$ is called \textit{standard} if: 
\begin{enumerate}

\item $P_n \subset (\Delta_n)$;

\item $P_n$ has dimension $d_n-1$.

\end{enumerate}

Given $V = (\N, (P_n)_{n \in \N}, (V_n)_{n \in \N})$ a polytope-form game with dim$(P_n)= d_n - 1$, there exists an affine an bijective mapping $e_n: P_n \to P^s_n \subset \Re^{d_n}$, where $P^s_n$ is a standard polytope. Let $e \equiv \times_n e_n$ and $V^{s}_n \equiv V_n \circ e^{-1}|_{P^s} $, where $P^s \equiv \times_n P^s_n$. The polytope-form game $V^{s} = (\N, (P^s_n)_{n \in \N} ,(V^{s}_n)_{n \in \N})$ is called a \textit{standard polytope form of $V$}. A standard polytope form of a game $V$ is unique up to the isomorphism $e$ which is used to ``standardize'' $V$. Because each $V^s_n$ is defined over $P^s$, each $V^s_n$ is uniquely defined over $\times_n (\Delta_n) \subset \times_n \mathbb{R}^{d_n}$. This implies it has an unique extension to a multilinear functional over $\times_n \mathbb{R}^{d_n}$. Then $V^s_n$ can be represented as vector $(V^s_n(x_1,...,x_N))_{x_i \in \mathbb{R}^{d_i}, i =1,..,N} \in \mathbb{R}^{D}$, where $D = d_1  ...  d_N$ and  $x_i \in \mathbb{R}^{d_i}$ denotes a canonical vector of $\mathbb{R}^{d_i}$. Let $\mathcal{E}^{s} \equiv \{ (V^s, \s) \in \Re^{DN} \times P^s \mid  \text{$\s$ is an equilibrium  $V^s$}\}$ be the graph of equilibria for standard polytope-form games over $P^s$. Let $\text{proj}_{\mathbb{R}^{DN}}: \mathcal{E}^s \rightarrow \mathbb{R}^{DN}$ be defined by $\text{proj}_{\mathbb{R}^{DN}}(V^s, \s) = V^s$. We will denote by $V^{s,n}(\sigma_{-n})$ the vector ($V^s_n(x_n,\sigma_{-n}))_{x_n \in \mathbb{R}^{d_n}}$, where $$V^s_n(x_n,\sigma_{-n}) \equiv \sum_{x_m \in \mathbb{R}^{d_m} : m \neq n}V^s_{n}(x_n, x_{-n})\Pi_{m \neq n}\sigma_{m x_m},$$ where $\sigma_m \equiv (\sigma_{m x_m})_{x_m \in \mathbb{R}^{d_m}} \in \mathbb{R}^{d_m}$.

Let $\sigma^{*}_n$ be the uniform distribution over the vertices of the polytope $P^{s}_n$ and $V^s \in \mathbb{R}^{DN}$. Defining $g_{nx_n}\equiv V^s_n(x_n, \sigma^{*}_{-n})$, we have that $V^s_n(x_1,...,x_N)=\bar{V}^s_n(x_1,..,x_N) + g_{nx_n}$, where $\bar{V}^s_n(x_1,...,x_N) \equiv V^s_n(x_1,...,x_N)-g_{nx_n}$. Letting $g_n \equiv (g_{nx_n})_{x_n \in \Re^{d_n}}$, the decomposition of $V^s_n$ in $\bar{V}^s_n$ + $g_n$ is unique and we denote it by $\bar{V}^s_n \oplus g_n$. Let $\sigma \in P^{s}$ be an equilibrium of the polytope-form game $V$. Define $\theta^s: \mathcal{E}^s \rightarrow \mathbb{R}^{DN}$ by $\theta^s_{nx_n}(\bar V^s,g,\s) \equiv \bar{V}^s_n(x_n,\sigma_{-n}) + g_{nx_n} + \sigma_{nx_n}$. Lemma \ref{ST} is the structure theorem for the graph of equilibria for standard polytope-form games over $P^s$. Its proof is entirely similar to the proof of the $KM$-structure theorem in \cite{KM1986}. We include it in the subsection \ref{ProofST} in the Appendix for completeness.

\begin{lemma}\label{ST}
The map $\theta^s: \mathcal{E}^s \rightarrow \mathbb{R}^{DN}$ is a homeomorphism. Moreover, there is a homotopy between $\text{proj}_{\mathbb{R}^{DN}} \circ (\theta^s)^{-1}$ and the identity function on $\mathbb{R}^{DN}$ and this homotopy extends to a homotopy on the one-point compactification of $\mathbb{R}^{DN}$.\end{lemma}

From Lemma \ref{ST} we can then obtain a structure theorem for the equilibrium graph over products of polytopes $P \equiv \times_n P_n$ which are possibly not standard. The space $A(P)$ of \textit{multiaffine functions over $P$} is an $D$-dimensional linear space, where the linear space structure is given by pointwise addition and scalar multiplication. Let $\mathcal{E}^{PF} \equiv \{ (V,p) \in \times_{n}A(P) \times P \mid  \text{$p$ is an equilibrium of $V$}\}$. The linear space $\times_n A(P)$ is a $DN$-dimensional Euclidean space and we denote its one-point compactification by $\overline{\times_{n} A(P)}$. Recall that the one-point compactification $\overline{\times_{n} A(P)}$ is homeomorphic to the sphere $\mathbb{S}^{DN}$. Let proj$_{V}: \E^{PF} \to \times_{n} A(P)$ be the natural projection over the payoff coordinate.

For the next proposition fix, for each $n \in \N$, $e_n: P_n \to P^s_n$ an affine isomorphism between a polytope $P_n$ and a standard polytope $P^s_n$. Let $e \equiv \times_n e_n$. Let $\bar{e}: A(P) \rightarrow \mathbb{R}^{D}$ be defined as $\bar{e}(V_n) = V^s_n$, where $V^s_n \equiv V_n \circ e^{-1}|_{P^s}$. It is easy to check that $\bar{e}$ is a linear isomorphism. Therefore, for $V \in \mathbb{R}^{DN}$ the mapping $T \equiv \times_{n}\bar{e}: \times_n A(P) \rightarrow \mathbb{R}^{DN}$ is also a linear isomorphism. 

\begin{proposition}\label{NSST}
There exists a homeomorphism $\theta^{PF}: \mathcal{E}^{PF} \rightarrow \times_n A(P)$. Moreover, there exists a homotopy between $\text{proj}_V \circ (\theta^{PF})^{-1}$ and the identity function on $\times_n A(P)$ which extends to a homotopy on the one-point compactification of $\times_n A(P)$.\footnote{\cite{AP2009} has investigated structure theorems for more general Nash-graphs than the ones explored in this paper, where the strategy sets of the players are compact convex sets and the the payoff functions are not necessarily multiaffine in the product of these sets.}
\end{proposition}

\begin{proof}Define the following mapping $e^{PF}: \mathcal{E}^{PF} \to \E^{s}$ by $e^{PF}(V,p) = (T(V), e(p))$. The map $e^{PF}$ is a homeomorphism. This implies that there exists a homeomorphism between $\mathcal{E}^{PF}$ and $\times_n A(P)$ given by $\theta^{PF} \equiv  T^{-1} \circ \theta^s \circ e^{PF}: \mathcal{E}^{PF} \rightarrow \times_n A(P)$. Let now $H: [0,1] \times \times_n A(P)\rightarrow \times_n A(P)$ be defined by $H(t,V) = T^{-1} \circ H(t, T(V))$, where $H$ is the homotopy from Lemma \ref{ST}. Since $T$ is a proper mapping, it follows that the homotopy $H$ also extends to a homotopy in the one-point compactification of $\times_n A(P)$ continuously. Notice now that $\text{proj}_V \circ (e^{PF})^{-1} = T^{-1} \circ \text{proj}_{\Re^{ND}}$ and from the definition of $\theta^{PF}$ we have that $(\theta^{PF})^{-1} = (e^{PF})^{-1} \circ (\theta^s)^{-1} \circ T$. Combining these two facts, it implies $\text{proj}_V \circ (\theta^{PF})^{-1} = \text{proj}_{V} \circ ((e^{PF})^{-1} \circ (\theta^s)^{-1} \circ T) = T^{-1} \circ \text{proj}_{\Re^{ND}} \circ (\theta^s)^{-1} \circ T$. This in turn implies that $H(0, \cdot) = \text{proj}_{V} \circ (\theta^{PF})^{-1}$ and  $H(1, \cdot) = \text{id}_{\times_nA(P)}$. \end{proof}

\begin{remark}The definition of the standard polytope form of a polytope-form game is not canonical, because it is based on the arbitrary choice of $e_n$, for each $n \in \N$. There are many different affine isomorphisms between $P_n$ and $P^s_n$ that could be used for this purpose. The choice of one of them defines a homeomorphism $\theta^{PF}$ and therefore generates a structure theorem which depends on those isomorphisms. As explained in section \ref{introdeg}, a structure theorem defines a degree theory for polytope-form games. Because this structure theorem depends on the arbitrary choice of $e$ so does the degree theory associated to it and therefore the immediate question is whether the degree of a component of equilbria $Q$ of a polytope-form game $V$ originated from a particular homeomorphism $\theta^{PF}$ - denoted $\text{deg}^{PF}_V(Q)$ - depends on the choice of the isomorphism $e_n$ for each $n$. Intuitively, this should not happen because $e$ simply defines a new ``representation'' of the polytope form game $V$ (a standard polytope form $V^s$), i.e., it does not create any new strategic possibilities for the players. Proposition \ref{identicaltheories}, which we will later present, implies that the polytope-form degree does not depend on the arbitrary choice of $e$. It actually shows more: the polytope-form degree is invariant to reductions.\end{remark}

\subsection{Polytope-form Index Theory}\label{polytope-form index}Defining an index theory of equilibria in polytope-form games is an immediate extension of the same exercise in normal-form theory. For specifics, see subsection \ref{indextheorygeneral} in the Appendix. We recall the main points of the construction for completeness. 

Let $V = (\N, (P_n)_{n \in \N}, (V_n)_{n \in \N})$ be a polytope-form game and let $P \equiv \times_n P_n$. The procedure to define the index is exactly analogous to the procedure we introduced for the normal form (cf. subsection \ref{indextheory}). For simplicity, suppose $f: U \to P$ is a differentiable map defined on a an open neighborhood $U$ of $P$ and such that the fixed points of $f$ are the Nash equilibria of game $V$. Let $d_{f}$ be the displacement of $f$. Then the Nash equilibria of $V$ are the zeros of $d_{f}$.  Suppose now that the Jacobian of $d_f$ at a Nash equilibrium $p$ of $V$ is nonsingular. Then we can define the {\it index} of $p$ under $f$ as $\pm 1$ depending on whether the determinant of the Jacobian of $d_{f}$ at $p$ is positive or negative. We can obtain a definition of the index of a component of equilibria by considering a perturbation of its displacement that is differentiable and has finitely many zeros: the indices of the isolated zeros of any sufficiently small perturbation of the initial displacement sum to the same constant, which can be defined as the index of the component.

Similarly to what we presented in subsection \ref{indextheory} for normal form, we could consider Nash-maps for polytope-form games. Fixing the cartesian product of polytopes $P \equiv \times_n P_n$  and the space $A(P)$, $f: A(P) \times P \to P$ is a Nash-map if it is continuous and the fixed points of $f(V, \cdot)$ are the equilibria of $V$. As in the normal-form case, one potential problem with the definition of index is the dependence of the computation on the function $f$. We show that the index is independent of $f$. A proof of the result for polytope-form games can be found in subsection \ref{indextheorygeneral} in the Appendix  (Proposition \ref{inv of nashmap index}).  In this Appendix, we also present a contruction of index theory for polytope-form games using general topological tools. This is done not only for completeness of exposition but because some of the details in the proofs in the polytope-form environment change when compared to normal form, and this requires verification. Similarly, the formal definition of index will be important in the proofs of Proposition \ref{identicaltheories} and Theorem \ref{PFNFequivalence}. Notationwise, if $Q$ is a component of equilibria of $V$ we denote the index of $Q$ w.r.t. $V$ by $\text{ind}_V(Q)$.

A Nash map which is particularly important for the remainder of the paper is the G\"ul, Pearce and Stachetti map ($GPS$ map, for short. Cf. \cite{GPS1993}). This Nash map can be defined as follows. Let $V^s = (\N, (P^s_n)_{n \in \N}, (V^s_n)_{n \in \N})$ be a standard polytope-form game. Let $\text{w}^n_{V^s}: P^s \rightarrow \mathbb{R}^{d_n}$ be given by $\text{w}^n_{V^s}(\sigma) = \s_n  + V^{s,n}(\sigma_{-n}) \in \mathbb{R}^{d_n}$. For each $n \in \N$, let $r_n:\mathbb{R}^{d_n} \to P^s_n$ be the closest-point retraction, $r \equiv \times_n r_n$ and $\text{w}_{V^s} \equiv \times_n \text{w}^n_{V^s}$. Let $\Phi_{V^s} \equiv r \circ \text{w}_{V^s}: P^s \rightarrow \times_{j}\mathbb{R}^{d_j} \rightarrow P^s$. We claim that $\sigma$ is an equilibrium of $V^s$ if and only if it is a fixed point of $\Phi_{V^s}$. The variational inequality \ref{closestpoint} characterizes the nearest-point retractions $r_n(z_n)$ of $z_n \in \Re^{d_n}$:

\begin{equation}\label{closestpoint}
 \langle\tau_n - r_n(z_n), z_{n} - r_{n}(z_n) \rangle \leq 0, \forall \tau_n \in P^{s}_n.
 \end{equation} 

If $\s \in P^s$ is a fixed point of $\Phi_{V^s}$, then $\s$ satisfies $\forall n \in \N$, $\langle \s'_n - \s_n, \s_n + V^{s,n}(\s_{-n}) - \s_n \rangle \leq 0, \forall \s'_n \in P^s_n$, which implies $\langle \s'_n - \s_n, V^{s,n}(\s_{-n}) \rangle \leq 0$. Note that $\s_n \cdot V^{s,n}(\s_{-n})$ is precisely the payoff to player $n$. Therefore, $\langle \s'_n - \s_n, V^{s,n}(\s_{-n}) \rangle \leq 0, \forall \s'_n \in P^s_n$ shows that $\s_n$ is indeed a best reply and that $\s$ is an equilibrium. Conversely, $\langle \s'_n - \s_n, V^{s,n}(\s_{-n}) \rangle \leq 0,\forall \s'_n \in P^s_n$ implies that $\s_n$ is the nearest-point retraction of $\s_n + V^{s,n}(\s_{-n})$ for each player $n$. Therefore, $\s$ is a fixed point of $\Phi_{V^s}$. The map $\Phi : \Re^{ND} \times P^s \to P^s$, given by $\s \mapsto \Phi_{V^s}(\s)$, is trivially continuous. This shows that $\Phi$ is indeed a Nash map. 

We will also make use of another map associated to $\Phi_{V^s}$. Let $\Psi_{V^s} \equiv \text{w}_{V^s} \circ r$. This is the \textit{commuted GPS map} of $V^s$. The commutativity property of the index (see \cite{D1972}, Chapter VII, Theorem 5.14) gives us that the sets of fixed points of $\Phi_{V^s}$ and of the permuted map $\Psi_{V^s}$ are homeomorphic and their indices agree.


\subsection{Equivalence of Index and Degree in Polytope-form Games}\label{equivdegindtheories}In this section we show that the degree and the index of an equilibrium component are identical and invariant under reductions. 

Let $V = (\N, (P_n)_{n \in \N}, (V_n)_{n \in \N})$ and $\bar V = (\N, (\bar P_n)_{n \in \N}, (\bar V_n)_{n \in \N})$ be two equivalent polytope-form games, and let the polytope-form game $V' = (\N, (P'_n)_{n \in \N}, (V'_n)_{n \in \N})$ be a common reduction. Let $q^{V}: P \to P'$ and $q^{\bar V}: \bar P \to  P'$ be the two reduction maps of $V$ and $\bar V$, respectively. Suppose $\bar Q$ is a component of equilibria of $\bar V$ and $Q$ is a component of equilibria of $V$, with $Q' \equiv q^{V}(Q) = q^{\bar V}(\bar Q)$ an equilibrium component of $V'$.

\begin{proposition}\label{identicaltheories}
The following statements hold: 
\begin{enumerate}
\item $\text{deg}^{PF}_{V'}(Q') = \text{deg}^{PF}_V(Q) = \text{deg}^{PF}_{\bar V}(\bar Q)$;
\medskip
\item $\text{ind}_{V'}(Q') = \text{ind}_{V}(Q) = \text{ind}_{\bar V}(\bar Q)$;
\medskip
\item $\text{deg}^{PF}_{V'}(Q') =\text{ind}_{V'}(Q')$.
\end{enumerate}
\end{proposition}

Here are the main ideas of the proof of Proposition \ref{identicaltheories}. The proof is divided in a few steps, which together prove the three items of the proposition at once. To start, we prove that $\text{deg}^{PF}_V(Q) = \text{deg}^{PF}_{T(V)}(e(Q))$, which shows that the degree is invariant to (any) standartization. The proof itself is a technical exercise using the formal definition of the degree in terms of homology. It is expected that such a result would hold because $e^{PF}$ (cf. proof of Proposition \ref{NSST}) should be understood as just a reparametrization of the graph $\E^{PF}$ and therefore should not affect the degree. In normal form, invariance to standartization is immediate (since the unit simplex is by definition a standard polytope), so the result has a point only in polytope form. The second step of the proof is to show that $\text{deg}^{PF}_{T(V)}(e(Q)) = \text{ind}_{T(V)}(e(Q))$. This step follows essentially the same strategy used in normal-form games to prove the identity between the $KM$-degree and the index of a component assigned by the $GPS$ map,  and links the index assigned by this map to $e(Q)$ to the degree $\text{deg}^{PF}_{T(V)}(e(Q))$. Using the commuted $GPS$ map and the commutativity property of the index, one then establishes that $\text{deg}^{PF}_{V}(Q) = \text{ind}_{T(V)}(e(Q))$. We then use Proposition \ref{inv of nashmap index}, to ascertain that the index of $e(Q)$ with respect to $T(V)$ is invariant under any Nash map.  With these arguments, we establish $(3)$. Finally, once the link between the degree and the index is established, we prove that the index of $Q$ with respect to $V$ is invariant under reductions: this immediately establishes $(2)$, and given that we proved $(3)$, we finally have $(1)$. 

Theorem \ref{PFNFequivalence} below establishes that the $KM$ and $PF$ index and degree theories are identical and therefore capture the same robustness to perturbations of payoffs. Moreover, it establishes that robustness ultimately depends on the reduced polytope form. Let $\mathbb{G} = (\N, (S_n)_{n \in \N}, (\mathbb{G}_n)_{n \in \N})$ be a normal-form game. Let $V = (\N, (P_n)_{n \in \N}, (V_n)_{n \in \N})$ be the reduced polytope form of $\mathbb{G}$ and $\pi^{\mathbb{G}}: \S \to P$ the maximal reduction map of $\mathbb{G}$.

\begin{theorem}\label{PFNFequivalence}
Let $Q \subseteq P$ be a component of equilibria of $V$. Then $X \equiv (\pi^{\mathbb{G}})^{-1}(Q)$ is a component of equilibria of $\mathbb{G}$ and 
$$
\text{ind}_V(Q)  = \text{deg}^{PF}_V(Q) = \text{deg}^{KM}_{\mathbb{G}}(X) = \text{ind}_{\mathbb{G}}(X)
$$
\end{theorem}

\begin{proof}  The first equality follows from (3) in Proposition \ref{identicaltheories} and the last equality is known. We show that $\text{ind}_{\mathbb{G}}(X)$ equals $\text{ind}_V(Q)$. This equality is a consequence of the fact that $\pi^{\mathbb{G}}: \S \to P$ is a reduction map and therefore (2) of Proposition \ref{identicaltheories} applies.\end{proof}

The proof of Theorem \ref{PFNFequivalence} is an immediate application of Proposition \ref{identicaltheories}. On a practical level, the Theorem tells us that if one wants to check for robustness of an equilibrium component of a normal-form game, it is sufficient to check that the index of the component is non-zero in its associated maximal reduction. In other words, the identifiation of duplicate pure and mixed strategies from the normal-form does not alter the index of a component of equilibria, eventhough the dimensions of the payoff space decrease in this process and therefore make `robustness to payoff perturbations' a seemingly less stringent requirement.

\subsection{Formulas for Computation of the Degree}\label{polydegcomputation}
Let $V = (\N, (P_n)_{n \in \N}, (V_n)_{n \in \N})$ be a polytope-form game and let $Q$ be a component of equilibria of $V$. For each $n \in \N$, let $e_n: P_n \to P^s_n$ be an affine bijection of $P_n$ with a standard polytope $P^s_n$, defining a standartization $V^s_n$ of $V_n$. From Proposition \ref{identicaltheories} we have $\text{deg}^{PF}_{V}(Q) = \text{deg}^{PF}_{V^{s}}(Q^{s})$, where $Q^s \equiv (\times_n e_n)(Q)$. Thus the problem of computing the degree of the equilibrium component $Q$ of the arbitrary polytope-form game $V$ can be reduced to the problem of computing the degree of $Q^{s}$ w.r.t. $V^s$.

The payoff functions $V^s_n$ for each player $n \in \N$ are now multilinear and can be identified with a vector in $\mathbb{R}^{ND}$ where $D = d_1... d_N$, and dim$(P_n)$ = dim($P^s_n$)  = $d_n - 1$. We reduce now the problem of computation of the degree of $Q$ even further. 

Let $V ^s\bigoplus g$ be the standard polytope-form game whose payoff functions are given by $V^s_n(\s_n, \s_{-n}) + \s_n \cdot g_n$, $g_n \in \mathbb{R}^{d_n}, \s \in P^s$ and $\mathcal{E}_{V^s} \equiv \{(g, \s) \in \times_n \mathbb{R}^{d_i} \times P^s \mid$ $\s$ is an equilibrium of $V^s \bigoplus g \}$. Because $V^s$ is standard, the payoffs $V^s_n(\s_n, \s_{-n})$ can be written as $\s_n \cdot V^{s,n}(\s_{-n})$, where $V^{s,n}(\s_{-n}) = \nabla_{\s_n}V_n(\s_n,\s_{-n})$.  Now notice that there is a structure theorem for $\mathcal{E}_{V^s}$ that is exactly analogous to the one obtained in Lemma \ref{ST}: let $\theta_{V^s}: \mathcal{E}_{V^s} \rightarrow \times_n \mathbb{R}^{d_n}$ be defined by $(\theta_{V^s})_n(g,\s) = \s_n + V^{s,n} (\s_{-n}) + g_n \in \mathbb{R}^{d_n}$. It follows that $\theta_{V^s}$ is continuous and it can be verified that it has a continuous inverse: $(\theta_{V^s})^{-1}(z) = (h(z), r(z))$, where $h_n(z) = z_n - \s_n - V^{s,n}(\s_{-n})$ and $\s_m = r_m(z_m)$, $\forall m \in \N$, where $r_m$ is the closest-point retraction to the standard polytope $P^s_m$. Also, $\text{proj}_{g} \circ \theta^{-1}_{V^s}$ is homotopic to the identity function in $\times_n \mathbb{R}^{d_n}$, where $\text{proj}_g: \mathcal{E}_{V^s} \rightarrow \times_n\mathbb{R}^{d_n}$ is the projection over the first coordinate. 

The degree of $Q$ can be computed from the function $\text{proj}_{g} \circ \theta^{-1}_{V^s}$ as follows: let $\theta_{V^s}(0,Q) = K$. If $U$ is an open neighborhood of $K$ such that its closure cl$(U)$ contains no other $z$ with $h(z) = 0$ besides those in $K$, then $\text{deg}_{V^s}(Q)$ equals the local degree of $\text{proj}_{g} \circ \theta^{-1}_{V^s}|_{U}$ over $0$. The problem of computing the degree of $Q$ is therefore reduced to computing the degree over $0$ of the map $\text{proj}_{g} \circ \theta^{-1}_{V^s}|_{U}$. We will show now how to calculate the local degree of $\text{proj}_{g} \circ \theta^{-1}_{V^s}|_{U}$ over $0$ by approximating the game $V$ with ``generic'' games. 

We first define precisely what genericity means in this context. Let $\mathbb{V}(P^s_n)$ be the vertex set of the standard polytope $P^s_n$ and let $T_n \subset \mathbb{V}_n(P^s_n)$ be the subset of vertices that generates a face $[T_n]$ of the polytope $P^s_n$. The restriction of $\text{proj}_{g} \circ \theta^{-1}_{V^s}$ to the set of $z \in \times_n \mathbb{R}^{d_n} $ such that for each $n$, $r_n(z_n)$ is in the relative interior of the face $[T_n]$ is a polynomial map of degree $N-1$. Let $g$ be \textit{generic} if $g$ is a regular value of each polynomial map obtained through this restriction. The set of regular values is then open and dense in $\times_n \mathbb{R}^{d_n}$, by Sard's Theorem (cf. \cite{BCR2013}). We say that a standard game $\bar V^s \equiv V^s \bigoplus g$ is generic if $g$ is generic. As a consequence of the inverse function theorem and our definition of genericity, it follows that a generic standard polytope-form game has finitely many equilibria.  

For a generic standard game $\bar V^s$, let $\s$ be an equilibrium of this game and let $\theta_{V^s}(0,\s) = z$. Then the local degree of $\text{proj}_{g} \circ \theta^{-1}_{V^s}|_{U}$ over $0$ is given by $\text{sign} (\text{det} [D(\text{proj}_g \circ \theta^{-1}_{V^s})(z)] )$. The computation of the degree of $\s$ can therefore be done explicitly through the computation of the determinant of the Jacobian matrix $D(\text{proj}_g \circ \theta^{-1}_{V^s})(z)$ and then checking its sign. This Jacobian is a square matrix of dimension $d_1 + ... + d_n$.

Fix now a nongeneric and standard polytope-form game $V^s$. Assume $Q$ is a nondegenerate component of equilibria of $V^s$ and $U$ an open neighborhood of $K = \theta_{V^s}(0,Q)$, defined as before. From Proposition 5.12, Chapter IV, in \cite{D1972}, it follows the the $\text{deg}_{V^s}(Q)$ is locally constant in $V^s$. This implies that for a generic perturbation $g \in \times_n \mathbb{R}^{d_n}$ sufficiently close to $0$, the game $V^s \bigoplus g$ has finitely many equilibria and the local degree of $\text{proj}_g \circ \theta^{-1}_{V^s}|_{U}$ over $g$ equals the local degree of $\text{proj}_{g} \circ \theta^{-1}_{V^s}|_{U}$ over $0$. The additivity property of the degree (see \cite{D1972}, Proposition 5.8, Chapter IV) now implies that the local degree of $\text{proj}_g \circ \theta^{-1}_{V^s}|_{U}$ over $g$ is the summation: $$\sum_{z \in U: z \in h^{-1}(g)} \text{sign} (\text{det} [D(\text{proj}_g \circ \theta^{-1}_{V^s})(z)] ).$$

This shows therefore that: 

$$ \text{deg}_{V^s}(Q) = \sum_{z \in U: z \in h^{-1}(g)} \text{sign} (\text{det} [D(\text{proj}_g \circ \theta^{-1}_{V^s})(z)] ).$$

The formula above shows how the computation of the degree of a nondegenerate component of equilibria depends on the dimension of the strategy polytopes: the dimension of the Jacobian matrix $[D(\text{proj}_g \circ \theta^{-1}_{V^s})(z)]$ at $z$ is $d_1 + ... + d_n$. Typically the number of pure strategies of a player in the normal-form representation of an extensive-form game grows exponentially with the size of the tree. If the formula above is used for computation of the degree in the normal form of an extensive-form game, the dimension of the Jacobian matrix is therefore typically exponential in the number of terminal nodes of the tree.

When we perform reductions of the strategy sets, the dimension of the polytope strategy set of each player decreases. Maximal reductions, therefore, imply a formula for the computation of the degree where the sum $\sum_{i}d_i$ is the smallest. But if an extensive game is given and one is interested in applying the formula described above to compute the index of an equilibrium component in the reduced polytope form, one must compute the reduced polytope form of the normal-form representation of the extensive game, which involves computing the normal form of the extensive game and then performing the (maximal) reduction. Again, this might be intractable, since it requires the computation of the normal form. In section \ref{extensiveform} we provide an alternative, more straightforward formula for the degree computation from extensive form data that circumvents this problem.  

\section{Polytope form of Extensive-form Games}\label{extensiveform}

In this section we present formulas for computing the degree of equilibria from extensive form data.  We introduce the \textit{enabling form of an extensive-form game}, which is a polytope-form game derived from the \textit{sequence form} (a concept defined in \cite{BS1996}). We show that the enabling form is a reduction of the normal formal of the extensive form that is particularly convenient to compute from the extensive form. The formulas for computation of the degree established in subsection \ref{polydegcomputation} then apply immediately to equilibrium components in the enabling form. In subsection \ref{GWsec} we show an alternative formula for computation that can also be applied more directly from the extensive form. 

We start with some preliminary definitions for extensive-form games. We fix from now on $\G \equiv  (T, \prec,U,\N, P_{*} )$ a game-tree with perfect recall. The set $T$ is the set of nodes and $\prec$ is the irreflexive binary relation of precedence in the tree $(T, \prec)$; that is, the relation $\prec$ is acyclic and totally orders the predecessors $\{t^{'} | t^{'} \prec t \}$ of $t$. The subset of terminal nodes -- those with no successors -- is $Z \subset T$,  $U$ is a partition of $T \setminus Z$ into information sets of players and Nature. The set $U_n \subset U$ is the collection of information sets for player $n \in \N$ and $A_n(u)$ is $n$'s set of actions available at his information set $u \in U_n$. Let $A_n = \cup_{n \in \N}A_n(u)$ be the entire set of $n$'s actions. Write $u \prec z$ if $t \prec z \in Z$, for some node $t \in u$, and write $(u,i) \prec z$ if there exists $t \prec t^{'} \preceq z$ for some node $t^{'}$ that follows $t \in u$ and action $i \in A_n(u)$. Perfect recall implies that each $(U_n, \prec)$ is a tree. Player $n$ set of pure strategies is $S_n \equiv \{s : U_n \rightarrow A_n | s(u) \in A_n(u) \}$. \cite{HK1950} shows that in  a game tree with perfect recall each player $n$ can implement a mixture of pure strategies by a payoff-equivalent behavior strategy $b_n = (b_n(u))_{u \in U_n}$ in which each $b_n(u) \in \Delta(A_n(u))$ is a mixture of actions in $A_n(u)$; i.e., $b_n(i|u)$ is the conditional probability at $u$ that $n$ chooses $i$.

The space of payoffs of $\G$ will be denoted $\mathcal{G} \equiv \mathbb{R}^{N|Z|}$. An element $G \in \mathcal{G}$ defines a payoff $G_n(z)$ to player $n$ at final node $z$. The space of outcomes is $\Omega = \Delta(Z)$, where an outcome $F \in \Omega$ assigns probability $F(z)$ to $z$. The probability $F_{*}(z) > 0$ is the probability that Nature's actions do not exclude the final node $z$. The probability $F_{*}(z) > 0$ is formally defined as follows: consider Nature as a player that plays a fixed behavior strategy. Then fix any mixed strategy $\sigma_* \in \Delta(S_*)$ that is equivalent to this behavior strategy, where $S_*$ are Nature's ``pure strategies'' . Let $S_*(z) = \{s \in S_* | (u,i) \prec z \Rightarrow s(u) = i  \}$. Then $F_*(z) \equiv \sum_{s_* \in S_*(z)}\sigma_*(s_*)$.


\subsection{Sequence- and Enabling-form strategy sets} 

The \textit{sequence form} of an extensive game is an alternative representation of an extensive game introduced by \cite{BS1996} in order to compute equilibria of extensive games more efficiently. A thorough discussion of the advantages of the sequence form for this purpose can be found also \cite{KMS1996}. Our purpose in this paper is to use this new representation of the game in order to compute degrees or indices of equilibria, thus obtaining methods to identify equilibria which are robust to payoff perturbations. More specifically, we aim at obtaining formulas for computation of degrees of equilibria that can be derived from the extensive-form data directly, without computing the normal-form representation. For completeness, we recall the main definitions. 


\begin{definition}[Definition 3.1 in \cite{BS1996}]
For each player $n$, a \textit{sequence} of choices of player $n$ \textit{defined} by a node $t$ of the game tree $\G$ is the set of actions of player $n$ on the path from the root to $t$. The set of sequences is denoted $\mathfrak{S}_n$.
\end{definition}

Every node in an information set $u$ of player $n$ defines the same sequence of actions for that player from the root to that information set (due to perfect recall). This sequence is denoted $\mathfrak{s}_u$ and is called the \textit{sequence leading to} $u$. An action $a_n \in A_n(u)$ and the sequence $\mathfrak{s}_u$ define another sequence $\mathfrak{s}_u \cup \{a_n\}$. This extended sequence is denoted $\mathfrak{s}_u a_n$. Therefore, a nonempty sequence of player $n$ is defined by its last action $a_n$ and the set of sequences can be represented as 

$$
\mathfrak{S}_n = \{\emptyset\} \cup \{\mathfrak{s}_u a_n \mid u \in U_n, a_n \in A_n(u)\}
$$

Consider now the space $\Re^{\mathfrak{S}_n}$. Each coordinate of a vector $r^n \in \Re^{\mathfrak{S}_n}$ is viewed as indexed by an element $\mathfrak{S}_n$. The realization plans $r^n$ of player $n$ are the set of solutions $r^n \in \Re^{\mathfrak{S}_n}$ to the following system of linear equations: 

\begin{equation}
r^n_{\emptyset} = 1
\end{equation}

\begin{equation}
- r^n_{\mathfrak{s}_u} + \sum_{a_n \in A_n(u)}r^n_{\mathfrak{s}_u a_n} = 0, \text{ for $u \in U_n$. }
\end{equation}

\begin{equation}\label{behaviortosequence}
r^n_{\mathfrak{s}_n} \geq 0, \text{ for $\mathfrak{s}_n \in \mathfrak{S}_n$. }
\end{equation}

It can be easily checked that the set of solutions satisfying the system above forms a polytope of $\Re^{\mathfrak{S}_n}$. We denote the polytope of realization plans by $\mathfrak{P}_n$. 

A behavior strategy $b_n$ of player $n$ defines a unique realization $r^n$ for player $n$ as follows (cf. 3.1 \cite{BS1996}): for each sequence $\mathfrak{s}_n \in \mathfrak{S}_n$, we have 

\begin{equation}\label{sequencetobehavior}
r^n_{\mathfrak{s}_n} = \prod_{a_n \in \mathfrak{s}_n}b_{n}(a_n|u).
\end{equation}

Conversely, the proof of Proposition 3.4 in \cite{BS1996} shows that each realization plan $r^n$ of player $n$ defines a collection of behavior strategies as follows: for each information set $u \in U_n$ and $a_n \in A_n(u)$, define 
\begin{equation}
b_n(a_n|u) = \frac{r^n_{\mathfrak{s}_u a_n}}{r^n_{\mathfrak{s}_n}},
\end{equation}
 if $r^n_{\mathfrak{s}_n}>0,$ and arbitrarily if $r^n_{\mathfrak{s}_n} =0$. Nature is usually considered as player $0$, who plays a fixed behavior strategy. Nature's realization plan is derived as in equation \eqref{sequencetobehavior} from its fixed behavior strategy. 


\textit{Enabling strategy sets} were introduced in \cite{GW2002} in order to obtain structure theorems for game trees. The set of enabling strategies of player $n$ can be defined by considering the natural projection of the set of realization plans of player $n$ over those coordinates which correspond to sequences defined by terminal nodes of the game tree. Formally, let $L_n \equiv \{ \mathfrak{s}_n \in \mathfrak{S}_n \mid \mathfrak{s}_n \text{ is defined by a terminal node $z$} \}$ and $\text{proj}_{L_n}: \Re^{\mathfrak{S}_n} \to \Re^{L_n}$ be defined by $r^n \mapsto \text{proj}_{L_n}(r^n) = (r^n_{\mathfrak{s}_n})_{\mathfrak{s}_{n} \in L_n}$. Note that $L_n = \emptyset$ iff $n$ is a dummy player and that each sequence defined by a terminal node $z$ has a unique \textit{last action} $\ell_n(z)$ which identifies that sequence. Therefore, we can view $L_n$ as the set of player $n$'s last actions: $i_n \in L_n \subset A_n$ if there exists $z \in Z$ such that $i$ is the $\prec$-maximal element in $A_n(z) \equiv \{i ^{'}_n \in A_n | i_n^{'} \prec z \}$. That is, $i_n = \ell_n(z) = \text{argmax} A_n(z)$.

\begin{definition}\label{enabling}
The \textit{enabling strategy set} of player $n$ is $\text{proj}_{L_n}(\mathfrak{P}_n)$ and is denoted $C_n$. 
\end{definition}

\begin{remark}Because $\text{proj}_{L_n}$ is affine and $\mathfrak{P}_n$ is a polytope, $C_n$ is also a polytope. In \cite{GW2002}, enabling stratagies are defined using the mixed strategy set of each player, instead of using behavior strategies, in an entirely analogous, but different procedure to the one showed above. \end{remark}

Let $L \equiv \times_n L_n$. We can now define a polytope-form game $V^{e} = (\N, (C_n)_{n \in \N}, (V^{e}_n)_{n \in \N})$, called the \textit{enabling form of $G$}, by defining payoffs as follows: for each $n \in \N$, let $g_n: L \to \Re$ be defined by $g_n(i) \equiv G_n(z)$, if for each $m \in \N$, the sequence (with last action) $i_m$ is defined by $z$. It is easy to see that each $z \in Z$ defines a unique sequence $i_m$ for each player $m$. If otherwise, then set $g_n(i) \equiv 0$. We can now define the payoff function $V^e_n$ as 

\begin{equation}\label{enablingpayoff}
V^{e}_n(r) \equiv \sum_{i \in L}g_n(i)\prod^{N}_{m=0}r^m(i_m).
\end{equation}

Note that this is precisely the same way \cite{BS1996} defines the payoffs of its sequence form. The payoff function $V^e_n$ is affine in each coordinate $r^m$ and indeed defines a polytope-form game. We denote $C \equiv \times_n C_n$.  

\begin{remark} In the Appendix, subsection \ref{addextensiveform}, we include additional results about the payoff space of the enabling form of an extensive game. We show there that the space of payoffs associated to the enabling form of an extensive game is a linear subspace of the space of multiaffine functions over $C$. The dimension of this linear subspace could be strictly lower than the space of multiaffine functions over $C$ but has dimension $N|Z|$.\end{remark}

\begin{proposition}\label{enablingreduction}
Let $\mathbb{G}$ be the normal form of an extensive-form game $G \in \mathcal{G}$. The polytope-form game $V^e$ is a reduction of $\mathbb{G}$. 
\end{proposition}

\begin{proof} Define for each $i \in L_n$, $Z_n(i) \equiv  \ell_n^{-1}(i) = \{z | \ell_n(z) = i\}$ and for each $z \in Z$, $S_n(z) \equiv \{s \in S_n | (u,i) \prec z$ implies $ s(u) = i\}$. If  $L_n = \emptyset$, then $n$ is a dummy player. Note that for each $z,z' \in Z_n(i), S_n(z) = S_n(z')$. Hence, we define $s_n(i) \equiv S_n(z), z \in Z_n(i)$. For each $n \in \N$, define the map $q_n: \S_n \to C_n$ given by $q_n(\s_n) \equiv \sum_{s_n \in s_n(i)}\s_n(s_n) = r^n_i$. The map $q_n$ is affine and surjective, so is a reduction map. Moreover, for each player $n$, $V^{e}_n \circ q = \mathbb{G}_n$.\end{proof}

Given an extensive game $G \in \mathcal{G}$, we denote by $q^{e}: \S \to C$ the reduction map from mixed to enabling strategies of this game. The next proposition compares the polytopes of the reduced polytope form of the extensive game and the enabling form. 

\begin{proposition}\label{enablingtoreduced}
Let $G \in \mathcal{G}$ be an extensive game and $\mathbb{G}$ the normal form of the extensive game. Let $V^r$ be the reduced polytope form of $\mathbb{G}$ and $V^e$ the enabling form. Then for each $n \in \N$, there exists an affine and surjective map $\bar{\pi}^{\mathbb{G}}_n: C_n \to P^r_n$ with $\bar{\pi}^{\mathbb{G}} \equiv \times_n \bar{\pi}^{\mathbb{G}}_n$ such that $V^{r} \circ \bar{\pi}^{\mathbb{G}} = V^e$. \end{proposition}

\begin{proof}Let $\pi^{\mathbb{G}}_n : \S_n \to P^r_n$ be the maximal reduction map of game $\mathbb{G}$ to $V^r$. Let $q^e_n: \S_n \to C_n$ be the reduction map of $\S_n$ to enabling strategies $C_n$. Note that $\pi^{\mathbb{G}}_n$ is constant in the fibers of $q_n$. Therefore there exists an affine and surjective mapping $\bar{\pi}^{\mathbb{G}}_n: C_n \to P^r_n$ such that $\bar{\pi}^{\mathbb{G}}_n \circ q_n = \pi^{\mathbb{G}}_n$.\end{proof}

An immediate consequence of Proposition \ref{enablingtoreduced} is that $\text{dim}(C_n) \geq \text{dim}(P^r_n)$, and it is not hard to see that the weak inequality can be strict for certain extensive games, because $C_n$ by definition is a quotient space from $\S_n$ that produces identifications entirely based on outcome-equivalences of strategies that arise from the game tree only, without taking payoffs into account. The maximal reduction map $\pi^{\mathbb{G}}_n$, however, takes also payoffs into account, and so produces ``more'' identifications than $q_n$. 
\begin{ex}
\begin{figure}[h]
\centering{}%
\caption{Game-Tree of Example \ref{running example}}\label{FIG:running}

\tikzstyle{hollow node}=[circle,draw,inner sep=1.5]
                         \tikzstyle{solid node}=[circle,draw,inner sep=1.5,fill=black]
                         \tikzset{
                           red node/.style={circle,draw=red,fill=red,inner sep=1.2},
                           blue node/.style={rectangle,draw=blue,inner sep=2.5}
                         }

\begin{tikzpicture}[scale=0.78]
 
\tikzstyle{level 1}=[level distance=15mm,sibling distance=35mm]
    \tikzstyle{level 2}=[level distance=15mm,sibling distance=25mm]
 		\tikzstyle{level 3}=[level distance=15mm,sibling distance=15mm]
 
  \node(0)[hollow node]{}
     child{node (0-1)[solid node]{}
     child{node (1-1)[red node]{} edge from parent node[left]{$l_1$}}
     child{node (1-2)[red node]{} edge from parent node[right]{$r_1$}}
	edge from parent node[left,xshift=-2]{$L$}    
     }    
     child{node(0-2)[solid node]{}     
       child{node(2-0)[solid node]{} 
       	child{node(2-1)[red node]{} edge from parent node[left]{$L_1$}}
       	child{node(2-2)[red node]{} edge from parent node[right]{$R_1$}}
		edge from parent node[left,xshift=-2]{$l$}      
       }
       child{node(3-1)[solid node]{}
       	child{node(3-2)[red node]{}}
       	child{node(3-3)[red node]{}}
		edge from parent node[right,xshift=-2]{$r$}      	
      	}
     edge from parent node[right,xshift=-2]{$R$} 	
     };
     \draw[dashed](2-0)to(3-1);
     \node[above] at (0){$1$};
     \node[above] at (0-1){$2$}; 
	 \node at($(2-0)!.5!(3-1)$){$1$};
	 \node[above] at (0-2){$2$};
	 \node[below] at (2-1){$z_3$};
	 \node[below] at (2-2){$z_4$};
	 \node[below] at (3-2){$z_5$};
	 \node[below] at (3-3){$z_6$};
	 \node[below] at (1-1){$z_1$};
	  \node[below] at (1-2){$z_2$};

\end{tikzpicture}
\end{figure}
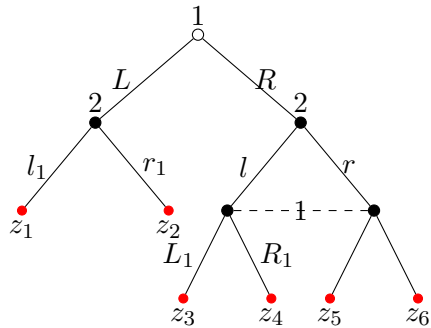

We show how to define payoff functions for player 1 in the polytope-form of the extensive-form game depicted in Figure \ref{FIG:running}, which is the same game tree as that of Example \ref{running example}. The payoff functions for player 2 can be defined using the same procedure. Since the identification of mixed strategies resulting in the enabling strategy set does not rely on payoffs, but only in the game tree, we substitute the specific terminal payoffs of Example \ref{running example}, by arbitrary terminal payoffs. Let $(G_{1}(z_i))_{i =1,..,6}$ be a vector of payoffs of player 1 over terminal nodes defined for the game tree in Figure \ref{FIG:running}. Then define: 

\begin{enumerate}

\item $V_{1}(L ,\ell_1) \equiv G_{1}(z_1)$ and $V_{1}(L, r_1) \equiv G_{1}(z_2)$;

\item $V_{1}(\ell, L_1) \equiv G_{1}(z_3)$ and $V_{1}(\ell, R_1) \equiv G_{1}(z_4)$;

\item $V_{1}(r, L_1) \equiv G_{1}(z_5)$ and  $V_{1}(r, R_1) \equiv G_{1}(z_6)$;

\end{enumerate}
Thus, $$V_1(p) = V_{1}(L ,l_1)p_1(L)p_2(l_1) + V_{1}(L,r_1)p_1(L)p_2(r_1) + V_{1}(l,L_1)p_1(L_1)p_2(l) +$$ $$V_{1}(l,R_1)p_1(R_1)p_2(l) + V_{1}(r, L_1)p_1(L_1)p_2(r) + V_{1}(r,R_1)p_1(R_1)p_2(r),$$ where $p_1 \in C_1$ and $p_2 \in C_2$. In this game tree, player 1 has 4 pure strategies, so his mixed strategy space is a 3-simplex. But the enabling strategy space $C_1$ is a 2-simplex. Now, player 2 also has 4 pure strategies but the enabling strategy set $C_2$ is 2-dimensional (it is actually the ``paralelogram'' of Example \ref{running example}). The strategy space $\Sigma_1 \times \Sigma_2$ in normal form has dimension 6 whereas $C_1 \times C_2$ has dimension 4. \end{ex}

As \cite{BS1996} points out, representing the extensive game in sequence form is convenient because it can be done directly from data of the extensive form, without computing the normal form. For the problem of computation of the degree of an equilibrium component, however, one still has to compute the payoffs $V^e$ in order to apply the formulas of subsection \ref{polydegcomputation}. We provide an alternative way to compute the degree of a component that depends on perturbing terminal payoffs of the game tree directly.

\subsection{The GW-Structure Theorem}\label{GWsec} \cite{GW2002} present a structure theorem for extensive games where the payoff space of the Nash-graph is the space of terminal payoffs $\mathcal{G}$ of a fixed game tree $\Gamma$. This structure theorem over terminal payoffs of the tree has however a limitation: the strategy space considered in this Nash-graph is a perturbation of the enabling strategy set to its relative interior. Govindan and Wilson show the perturbation is necessary: without it, there are game trees for which there are no structure theorems (cf. \cite{GW2002}). Nevertheless, this formulation of the structure theorem is the natural one for extensive-form games, because it involves terminal payoffs of the tree. We show that this particular structure theorem could be used to define a degree theory that allows us to compute the degree of equilibria with unperturbed enabling strategy sets. We start with some preliminary notation. 


For $\e>0$ and $n \in \N$, let $C^{\e}_n$ denote a subset of the enabling strategy $C_n$ which is a polytope, is contained in the relative interior of $C_n$ and is such that the Hausdorff-distance $d(C_n, C^{\e}_n) \leq \e$. Since $C^{\e}_n$ is a polytope that is a subset of $C_n$, the payoff functions $V^e_n$ are also defined over $C^{\e} \equiv \times_n C^{\e}_n$. We refer to the equilibria of the game $V^{e}$ with polytope strategy sets $C^{\e}_n$ for each player $n$ as \textit{$\e$-restricted equilibria of $G$}. Let $\mathcal{E}^{GW}_{\e} \equiv \{ (G,p) \in \mathcal{G} \times C^{\epsilon} \mid$ $p$ is an $\e$-restricted equilibrium of $G$ $\}$ be the \textit{$GW$ $\e$-equilibrium graph} of payoffs over terminal nodes of the tree (cf. \cite{GW2002}). This is the graph of the ``equilibria'' of the extensive-form game defined by $G \in \mathcal{G}$, when we restrict players to choose enabling strategies in $C^{\e}_n$.


\begin{remark}\label{remark}
For each player $n$, let $(G_n(z))_{z \in Z} \in \mathcal{G}$ be the vector of payoffs of player $n$ associated to the terminal nodes of the tree. We will use a similar notation to \cite{GW2002} and write $V^e_n(p_n, p_{-n}) = p_n \cdot \nu_n(p_{-n}) + \nu_{n}(\emptyset)$ where $\nu_n(p_{-n}) \equiv (\nu_n(i, p_{-n}))_{i \in L_n}$. Formally, given a vector of enabling strategies $(p_n)_{n \in \N} \in \times_n C_n$ and induced distribution $F \in \Delta(Z)$, we have $$V^e_n(p_n, p_{-n}) = \sum_{z \in Z}G_n(z)F(z) =$$ 
$$\sum_{i \in L_n}p_n(i)\sum_{z \in Z_n(i)}G_n(z)F^{n}(z) + \sum_{z| A_n(z) = \emptyset}G_n(z)F^{n}(z) = p_n \cdot \nu_n(p_{-n}) + \nu_n(\emptyset),$$where $F^{n}(z) = \Pi_{j \neq n}F_j(z)$.  Hence $p_n \cdot \nu_n(p_{-n})$ corresponds to the part of the multiaffine function $V^e_n$ that depends on the last actions of player $n$, whereas $\nu_n(\emptyset)$ depends exclusively on $p_{-n}$.
\end{remark}

Let $\text{proj}_{\mathcal{G}} : \mathcal{E}^{GW}_{\e} \rightarrow \mathcal{G}$ be defined by $\text{proj}_{\mathcal{G}}(G,p) = G$. For each fixed $\e>0$, there is a structure theorem for the graph $\E^{GW}_{\e}$: first, there is a homeomorphism $\Theta^{GW}_{\e}: \mathcal{E}^{GW}_{\e} \to \mathcal{G}$; second, $\text{proj}_{\mathcal{G}} \circ (\Theta^{GW}_{\e})^{-1}: \mathcal{G} \to \mathcal{G}$ is homotopic to the identity map on $\mathcal{G}$, by a homotopy that extends to the one-point compactification of $\mathcal{G}$. As we showed in section \ref{introdeg} for normal-form games, this structure theorem generates a degree theory which allows us to verify the robustness of connected components of solutions (in the variable $p \in C^{\e}$) of the equation \eqref{GWequation} to perturbations of the parameter $G \in \mathcal{G}$:

\begin{equation}\label{GWequation}
	\text{proj}_{\mathcal{G}}(G,p) = G.
\end{equation}
\medskip

Note that for a fixed $G$, there are finitely many $\e$-restricted equilibrium components $Q$ of $G$, because the set of $\e$-restricted equilibria of $G$ is semi-algebraic. Since the degree theory that immediately follows from this $GW$-structure theorem (we refer to this degree theory as \textit{GW-degree theory}) requires the $\e$-perturbation $C^{\e}$ of the strategy set $C$, it cannot be applied to verify the robustness of equilibrium components to perturbations of the terminal payoffs $G$. We therefore extend this theory to provide a method capable of executing this verification. 

We recall the definition of the $GW$-homeomorphism $\Theta^{GW}_\e: \E^{GW}_{\e}\to \mathcal{G}$ constructed in \cite{GW2002}. Denote by $F^u$ the uniform distribution over terminal nodes and let $E[ \cdot | \cdot ]$ be the conditional expectation operator for $F^u$. Define $\Theta^{GW}_\e(G,p) = H$ by: for the case $A_n(z) \neq \emptyset$,  $H_n(z) \equiv G_n(z) - g_n(\ell_n(z)) + p_n(\ell_n(z)) + \nu_n(\ell_n(z), p_{-n})$, with $g_n(i_n) \equiv E[G_n | Z_n(i_n)]$. In case $A_n(z) = \emptyset$, define $H_n(z) \equiv G_n(z)$. Theorem 5.2 in \cite{GW2002} shows this map is a homeomorphism by constructing an explicit inverse. Analogously to the $KM$ and $PF$ structure theorems we defined before, we can view $\Theta^{GW}_\e$ as fixing a subspace of $\mathcal{G}$ as follows: let $\tilde G_n(z) \equiv G_n(z) - g_n(\ell_n(z))$ if $\ell_n(z) \neq \emptyset$ and $\tilde G_n(z) \equiv G_n(z)$, if otherwise. Then $G_n(z) = \tilde G_n(z) + g_n(\ell_n(z))$, for all $z \in Z$ and the decomposition $G_n = (\tilde G_n, g_n)$ is actually unique. Using this decomposition the $n$-th coordinate of the homeomorphism can be rewritten as: $(\Theta^{GW}_\e)_n(G,p) = (\Theta^{GW}_\e)_n(\tilde G, g, p)  = (\tilde G_n(z), p_n(\ell_n(z)) + \nu_n(\ell_n(z), p_{-n}))_{z \in Z}$. Therefore, analogously to the previously presented structure theorems, the $GW$-homeomorphism acts only on the pairs of bonus and strategies $(g,p)$.

If $Q \subset C^{\epsilon}$ is an $\e$-restricted equilibrium component of $G$, we denote the degree of this component by $\text{deg}^{\e}_G(Q)$. Let $E(G)$ denote the set of equilibria of the extensive game $G$ in enabling strategies.

The next theorem provides a formula for computing the polytope- or normal-form degree in terms of the $GW$-degree theory. It shows ultimately that if $\e>0$ is taken sufficiently small, than the $GW$-degree theory can be used to compute the degree of the equilibria in $KM$ (or $PF$) degree theory. 

For the statement of the next result, recall that $q^{e}: \S \to C$ denotes the reduction map from mixed to enabling strategies.

\begin{Theorem}\label{stabilization}
Let $Q \subset C$ be an equilibrium component of the extensive-form game $G$. Let $W \subset C$ be an open neighborhood (in $C$) of $Q$ such that cl$_{C}(W) \cap E(G) = Q$ and let $X \equiv (q^{e})^{-1}(Q)$. For each $\e>0$, let $W^{\e} \equiv W \cap C^{\e}$. There exists $\bar \e>0$, such that for each $\e \in (0, \bar \e)$, $W^{\e}$ has no $\e$-restricted equilibria in its boundary and the following holds: 
\begin{equation}\label{formula}
\sum_{Q'}\text{deg}^{\e}_{G}(Q') = \text{deg}^{KM}_{\mathbb{G}}(X) = \text{deg}^{PF}_{V^{e}}(Q),
\end{equation}

where the sum above is over the connected components $Q' \subset W^{\e}$ such that $(G,Q') \in \mathcal{E}^{GW}_{\e}$.

\end{Theorem}

The first difficulty with the proof of Theorem \ref{stabilization} is to show that for $\e>0$ small $\sum_{Q'}\text{deg}^{\e}_{G}(Q')$ is constant. Though this sum is an integer by construction, it could in principle be that as $\e$ tends to $0$, this integer oscilates and no limit exists. Our proof strategy is to show that for each $\e>0$, the degree of an $\e-$restricted equilibrium component $Q'$ of $G$ is equal to the index of an associated Nash map on the restricted polytope $C^{\e}$. From this we again use the commutativity property of the index (cf. subsection \ref{polytope-form index}) to show that for a small $\e>0$, the sum of the indices of the equilibrium components in $W^{\e}$ is equal to the index of $Q$. This argument is done using the notion of best-reply index, which is defined in subsection \ref{indextheorygeneral}. The rest of the proof is an exercise in relating the best-reply index to the $GW$-degree theory in similar fashion to the proof of Proposition \ref{identicaltheories}.

\begin{remark}Given an equilibrium component $X$ in mixed strategies of the normal form of an extensive-form game $G \in \mathcal{G}$, the formula of Theorem \ref{stabilization} gives us an alternative way for computing the $KM$-degree of $X$ by using the $GW$-degree. The $GW$-degree can be computed using a similar procedure to the one we described in section \ref{polydegcomputation}, but now applied to terminal payoffs of the game tree, instead of polytope-form payoffs. Though the procedure is similar to the one explained in section \ref{polydegcomputation}, the details are different, so we would like to highlight these differences. 

Let $\e>0$ and $Q'$ an $\e$-restricted equilibrium component in enabling strategies of $G \in \mathcal{G}$. We compute $\text{deg}^{\e}_{G}(Q')$.
As already observed we can write $G = (\tilde G, g) \in \mathcal{G}$. This allows us to write an associated homeomorphism to $\Theta^{GW}_\e$: let $$\E^{\e}_{\tilde{G}} = \{ (g',p) \in \times_n\Re^{L_n} \times C^{\e} \mid  \text{ $p$ is an $\e$-restricted equilibrium of $\tilde G \oplus g'$}\}$$ and define 
$\theta^{\e}_{\tilde{G}}: \E^{\e}_{\tilde G} \to \times_n \Re^{L_n}$ as follows: let $\theta^{\e}_{\tilde G} \equiv \times_n (\theta^{\e}_{\tilde G})_n$, with $(\theta^{\e}_{\tilde G})_n(g,p) = (p_n(\ell_n(z)) + \nu_n(\ell_n(z), p_{-n}))_{z \in Z}$. Note that $(\Theta^{GW}_\e)_n(\tilde G, g, p) = (\tilde G, \theta_{\tilde{G}}(g,p))$. 

Let $\text{proj}_{g}: \E^{\e}_{\tilde{G}} \to \times_n \Re^{L_n}$ be the projection over the ``$g$'' coordinates and let $f^{\e}_{\tilde{G}} \equiv \text{proj}_{g} \circ (\theta^{\e}_{\tilde{G}})^{-1}$. As before, let $K \equiv \theta^{\e}_{\tilde{G}}(g, Q')$ and $B$ be an open neighborhood of $K$ which contains no solution $x$ of $f^{\e}_{\tilde{G}}(x) = g$ in its boundary. Then $\text{deg}^{\e}_{G}(Q')$ equals the local degree of $f^{\e}_{\tilde{G}}|_{B}$ over $g$. The local degree of $f^{\e}_{\tilde{G}}|_{B}$ over $g$ is locally constant in $g$, i.e., for $g'$ sufficiently close to $g$, the local degree of $f^{\e}_{\tilde{G}}|_{B}$ over $g'$ is well-defined and identical to that of $g$. This is an immediate consequence of the map $f^{\e}_{\tilde G}$ being proper (cf. \cite{D1972}, Proposition 5.12). For a generic choice of $g'$ the equation $f^{\e}_{\tilde{G}}|_{B}(x) = g'$ has finitely many solutions in $x$ with the map $f^{\e}_{\tilde{G}}|_{B}(x)$ being a diffeomorphism around each of the solutions.\footnote{Genericity here can be defined in analogous fashion to the procedure in section \ref{polydegcomputation}. The map $f^{\e}_{\tilde G}$ is smooth, except in a finite collection of closed, lower-dimensional subsets of $\times_n\Re^{L_n}$. When restricted to the complement of this union, $f^{\e}_{\tilde G}$ is a smooth map which, by Sard's Theorem, has a residual set of regular values. A game $\tilde G \oplus g$ is then \textit{generic} if $g$ is a regular value of $f^{\e}_{\tilde G}$.} The local degree of $f^{\e}_{\tilde{G}}|_{B}$ over $g'$ is therefore the sum of signs of the determinant Jacobian of $f^{\e}_{\tilde G}|_{B}$ at each solution $x$.  Hence, $\text{deg}^{\e}_{G}(Q')$ equals the sum of signs of the determinant Jacobian of $f^{\e}_{\tilde G}|_{B}$ at each solution $x$ of $f^{\e}_{\tilde{G}}|_{B}(x) = g'$, for generic $g'$ chosen sufficiently close to $g$.

Note that the map $f^{\e}_{\tilde{G}}$ is a map defined from $\Re^{|L_1| +...+ |L_n|}$  to itself. Hence the formula for $\text{deg}^{\e}_{G}(Q')$ involves computing the sign of the determinant of Jacobian matrices of dimension $|L_1| + ... +|L_N|$, which can then be used to compute the leftmost formula of equation \ref{formula}. The rightmost formula of \ref{formula}, as can be seen from subsection \ref{polydegcomputation} applied to normal form, involves computing the sign of the determinant of Jacobian matrices of dimension $|S_1|+ ... + |S_N|$. But $|S_n|$ is in general exponentially larger than $|L_n|$, since the latter is at most the number of terminal nodes of the game tree. From this perspective, the computation of the leftmost formula of \ref{formula} is more tractable than the rightmost. 

Intuitively, the leftmost formula in equation \ref{formula} is a sum of $GW$-degrees, which is a tool to verify robustness to payoff perturbations of terminal payoffs of the extensive game; the rightmost formula of \ref{formula} involves $KM$-degrees, which is a tool to verify robustness to payoff perturbations of the normal form. Though the $GW$ and $KM$-degree theories are formulated in very different spaces, they are, in the precise sense of \ref{formula}, equivalent. This motivates the following definition.

\end{remark}

\begin{definition}\label{indth}
Let $Q \subset C$ be an equilibrium component of $G \in \mathcal{G}$. Let $W \subset C$ be an open neighborhood (in $C$) of $Q$ such that $\text{cl}_{C}(W) \cap E(G) = Q$. 

$$
 \text{deg}_{G}(Q) \equiv \lim_{\e \downarrow 0}\sum_{Q'}\text{deg}^{\e}_{G}(Q'),
$$
where the sum above is over the connected components $Q' \subset W^{\e}$ such that $(G,Q') \in \mathcal{E}^{GW}_{\e}$.
\end{definition}

The number $\text{deg}_{G}(Q)$ is well-defined because of Theorem \ref{stabilization}. This number is also indepedent of the specific neighborhood $W$ of $Q$: for any open neighborhood $W' \subset C$ of $Q$ such that $\text{cl}_{C}(W') \cap E(G) = Q$, the limit above identical. We note a few important properties implied by Definition \ref{indth}.

\medskip
\begin{itemize}
\item[(P1)] \textbf{Payoff Robustness}: If $\text{deg}_{G}(Q) \neq 0$, then for sufficiently small perturbations $G'$ of the terminal payoffs of $G$, there exists an equilibrium of $G'$ which is closeby to $Q$.
\medskip
\item[(P2)] \textbf{Normal form Consistency}: Let $q^{e}: \S \to C$ be the reduction map from mixed to enabling strategies. Let $X \equiv (q^e)^{-1}(Q)$ be the component in normal form. Then $\text{deg}_{G}(Q) = \text{deg}^{KM}_{\mG}(X)$. 
\medskip
\item[(P3)] \textbf{Nash-Maps Computation}: Given a Nash-map $f: \times_{n}A(C) \times C \to C$, the Nash-map index defined by $f$ on $Q$ is identical to $\text{deg}_{G}(Q)$. 
\medskip
\item[(P4)] \textbf{Independence of Approximation}: $\text{deg}_{G}(Q)$ is independent of which polytope $C^{\e}$ is used for the limit argument. 
\medskip
\item[(P5)] \textbf{+1 property}: The sum of $\text{deg}_{G}(Q)$ over $Q$ is $+1$.
\end{itemize}

\medskip 

To see that (P1) holds, let $\text{deg}_{G}(Q) = m$. For $\e>0$ sufficiently small, Theorem \ref{stabilization} implies that $\sum_{Q'} \text{deg}^{\e}_G(Q') = m \neq 0$, where the sum is over the $\e$-restricted equilibrium components $Q'$ of $G$ in $W^{\e}$. There exists then $Q'$ such that $\text{deg}^{\e}_{G}(Q') \neq 0$. Propositions \ref{robustness degree} now implies the result. Property (P2) and property (P4) are immediate consequences of Theorem \ref{stabilization}. To see that property (P3) is satisfied, observe that there is an equivalence between the Nash map index in polytope and normal forms (cf. Proposition \ref{identicaltheories} and Theorem \ref{PFNFequivalence}), which is in turn equal to the degree in normal form. Theorem \ref{stabilization} then gives the result. Lastly, given terminal payoffs $G$, for $\e>0$, the $GW$-structure theorem assigns global degree $+1$ to the projection over $\mathcal{G}$ composed with $(\Theta^{GW}_{\e})^{-1}$. Since this composition is a proper map, its local degree over any $G$ is $+1$ (cf. \cite{D1972}, Section VIII, 4.4-5). For $\e>0$ sufficiently small, Theorem \ref{stabilization} gives $\sum_{Q'} \text{deg}^{\e}_G(Q') = +1$, where the sum is over all $\e$-restricted equilibrium components $Q'$ of $G$. Therefore, $\sum_{Q \in E(G)}\text{deg}_{G}(Q) = +1$.

\begin{example}\label{computation example}We present two examples in order to illustrate Definition \ref{indth} and the use of its properties. The first example we present is an example of \cite{GW2002} (see Figure \ref{FIG:GW}). This example shows how there is no structure theorem for the graph 
$$\mathcal{E} = \{(G,p) \in \Re^6 \times C \mid p \text{ is an equilibrium of }G \},$$ 
\medskip
where $C = C_1 \times C_2 = \D(\{T,B\}) \times \D(\{L,R\})$ is the enabling strategy set and $\Re^{6}$ is the space of terminal payoffs of the extensive-form game. We recall the reason why, for completeness: assume that there exists a structure theorem, where $H: \mathcal{E} \to \Re^6$ is a homeomorphism and $\text{proj}_{\mathcal{G}}: \mathcal{E} \to \Re^6$ is the natural projection to the payoff space. The game $G_{y} = G_{3}$ has a unique equilibrium path $BL$ that persists in a neighborhood of $G_{3}$, because $B$ remains a strictly dominant strategy for player $1$ and $L$ remains the unique best reply for player $2$. As the local degree of $\text{proj}_{\mathcal{G}} \circ H^{-1}$ at $G_{3}$ can be seen as counting the number of solutions to the equation $\text{proj}_{\mathcal{G}} (G_{3},p) = G_{3}$ (with the correct sign), where $p = (B,L)$, the local degree at $G_{3}$ must be $+1$ or $-1$ . The game $G_{1}$ has two equilibrium paths $BL$ and $T$, and again all games in a neighborhood (in $\Re^6$) of $G_{1}$ have these same two outcomes. Therefore, the local degree of $\text{proj}_{\mathcal{G}} \circ H^{-1}$ at $G_{1}$ must be $-2, 0,$ or $+2$. This is a contradiction with $\text{proj}_{\mathcal{G}} \circ H^{-1}$ being a proper map, which implies that the local degree at any point in $\Re^6$ is the same. Therefore, there cannot be a structure theorem for $\mathcal{E}$.

\begin{center}
\begin{figure}
\centering{}%
\caption{$G_{y}$}\label{FIG:GW}
\tikzstyle{hollow node}=[circle,draw,inner sep=1.5]
                         \tikzstyle{solid node}=[circle,draw,inner sep=1.5,fill=black]
                         \tikzset{
                           red node/.style={circle,draw=red,fill=red,inner sep=1.2},
                           blue node/.style={rectangle,draw=blue,inner sep=2.5}
                         }

\begin{tikzpicture}[scale=0.78]
 
\tikzstyle{level 1}=[level distance=15mm,sibling distance=35mm]
    \tikzstyle{level 2}=[level distance=15mm,sibling distance=25mm]
 	
   \node(0)[hollow node]{}
     child{node (0-1)[red node]{}
	edge from parent node[left,xshift=-2]{$T$}    
     }    
     child{node(0-2)[solid node]{}     
       child{node(2-0)[red node]{} 
		edge from parent node[left,xshift=-2]{$L$}      
       }
       child{node(3-1)[red node]{}
		edge from parent node[right,xshift=-2]{$R$}      	
      	}
     edge from parent node[right,xshift=-2]{$B$} 	
     };
     \node[above] at (0){$I$};
     \node[below] at (0-1){$(2,2)$}; 

	 \node[above] at (0-2){$II$};
	 \node[below] at (2-0){$(3,3)$};
	 \node[below] at (3-1){$(y,1)$};

\end{tikzpicture}
\end{figure}
\end{center}

We now show how the perturbations to the interior solve this matter and show how we can compute the degree of the equilibrium components in $G_{1}$: perturbing $C_1$ and $C_2$ to some polytope $C^{\e}$ in the interior in the game $G_{1}$ and computing the $\e$-restricted equilibria of this pertubed game, one sees that there is no $\e$-restricted equilibria in any neighborhood of $T$, as the $\e$ vanishes. This is because for a sufficiently small perturbation of the strategy sets, player $1$ would prefer to play $B$ with the highest probability possible, and $L$ is a dominant strategy for player $2$ (given the fixed perturbation to the interior). This implies that $T$ is not an equilibrium for any sufficiently small perturbation of $C_1$ and $C_2$ to the interior. By Definition \ref{indth}, $T$ has degree of $0$ (which is the same as its $KM$-degree). The same does not happen with $BL$. The  degree of $BL$ must then be $+1$ (due to (P5)).

\begin{center}
\begin{figure}
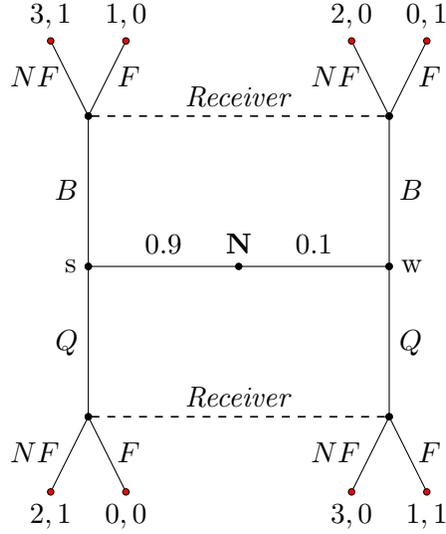

\caption{Beer-and-Quiche Game}\label{BandQ}
\bigskip
\bigskip
\begin{istgame}
\xtShowEndPoints[solid node,fill=red]
\xtdistance{0mm}{40mm}
\setistmathTF{1}{1}{1}
\istroot(o)(0,0){\textbf{N}} 
\istb{0.9}[a]{$s$}[l] 
\istb{0.1}[a]{$w$}[r] 
\endist 

\istroot[left](w)(o-1)
\istb{B}[l]
\istb{Q}[l]
\endist

\istroot'[right](s)(o-2)
\istb{B}[r]
\istb{Q}[r]
\endist

\xtInfoset[dashed](w-1)(s-1){$\mathit{Receiver}$}[a]
\xtInfoset[dashed](w-2)(s-2){$\mathit{Receiver}$}[a]

\xtdistance{10mm}{10mm}

\istroot[down](a)(w-2)
\istb{NF}[l]{2,1}[south]
\istb{F}[r]{0,0}[south]
\endist

\istroot[down](b)(s-2)
\istb{NF}[l]{3,0}[south]
\istb{F}[r]{1,1}[south]
\endist

\istroot[up](c)(w-1)
\istb{F}[r]{1,0}[north]
\istb{NF}[l]{3,1}[north]
\endist

\istroot[up](d)(s-1)
\istb{F}[r]{0,1}[north]
\istb{NF}[l]{2,0}[north]
\endist

\end{istgame}
\end{figure}
\end{center}
Consider now the  Beer-and-Quiche game (see \cite{CK1987}) depicted in Figure \ref{BandQ}. Recall that the decision nodes $w$ and $s$, after Nature's ($\textbf{N}$) move at the root belong to the Sender. The last action set of the sender can be described as $\{B_w, B_s, Q_w, Q_s\}$ ( ``$w$'' meaning weak, ``$s$'' meaning strong; ``$Q$'' meaning quiche and ``$B$'' meaning beer). The enabling strategy set of the sender is then $\D(\{B_w, Q_w\}) \times \D(\{B_s,Q_s\})$. The last action set of the receiver is $\{F_B, NF_{B}, F_Q, NF_{Q}\}$. The enabling strategy set of the receiver is then $\D(\{F_B, NF_B\}) \times \D(\{F_Q, NF_Q\})$.  There are two equilibrium components in enabling strategies in this game, each one associated with a distinct outcome: the first of these equilibrium components is described as follows: the sender chooses $(B_w, B_s)$. The receiver chooses $(NF_B, p \geq 1/2)$, where $p$ denotes the probability of $F_Q$. In the second component, the sender chooses $(Q_w, Q_s)$. The receiver chooses $(q \geq 1/2, NF_Q)$, where $q$ denotes the probability of $F_B$.

We show the second component of equilibria has degree $0$, and the first has degree $+1$. The second component is the one excluded by refinements such as the Intuitive Criterion (\cite{CK1987}) and Kohlberg-Mertens stability (\cite{KM1986}). In order to show that the second component has degree $0$, it suffices to show that for a particular sequence of strategy set perturbations converging to the unperturbed enabling set, no equilibrium is closeby to the component. Let us first set a neighborhood of the component. The strategy set of the Sender can be written as  $[0,1] \times [0,1]$, where a typical element is denoted by $(q_w, q_s)$, $q_w$ beeing the probability of playing Quiche after weak, and $q_s$ the probability of playing Quiche after strong. The strategy set of the Receiver can be written also as $[0,1] \times [0,1]$, where a typical element is denoted by $(f_B, f_Q)$, $f_B$ beeing the probability of fighting after Beer, and $f_Q$ the probability of fighting after Quiche. For $\d>0$, let $U_S = (1-\d,1] \times (1-\d,1]$ and $U_R = (1/2-\d, 1] \times [0,\d)$ and we now choose an appropriate $\d>0$. First observe that for the weak Sender, beer is a strictly inferior reply to any strategy in the component, as the equilibrium payoff is $3$, and deviating to beer gets him at most $2$. This strictness allows us to conclude that for $\d>0$ sufficiently small, the best reply of the weak Sender to any strategy in $U_R$ sets $q_w =1$. Fix such a $\d>0$. We now choose a perturbation for the enabling strategy set. Given $\e>0$, consider for the Sender the perturbed strategy set as $C^{\e}_S \equiv [\e, 1-\e] \times [\e, 1- \e] \subset [0,1] \times [0,1]$. For the Receiver, consider similarly the perturbed strategy set as $C^{\e}_R \equiv [\e, 1-\e] \times [\e, 1- \e] \subset [0,1] \times [0,1]$. Let $C^{\e} = C^{\e}_S \times C^{\e}_R$ and $U = U_S \times U_R$. Taking $\e>0$ sufficiently small, any best reply of the weak type to a strategy in $C^{\e}_R \cap U_R$ puts $q_w = (1-\e)$. This implies that the probability that the Receiver assigns to the Sender being weak (in the $\e$-perturbed game), in any $\e$-restricted equilibrium located in $U \cap C^{\e}$, is at most $0.1$. Therefore, for such $\e>0$, it implies that in any $\e$-restricted equilibrium in $U \cap C^{\e}$, the best reply of the Receiver is $f_B = \e$. Therefore, it must be that there is no equilibrium in $U \cap C^{\e}$. Hence, for a sufficiently small perturbation, there is no equilibrium closeby to the second component. This implies the degree of this component is zero. By (P5), the degree of the first component is $+1$. 

\end{example}

\section{Appendix}

\subsection{Proof of Lemma \ref{ST}}\label{ProofST}Let $\sigma \in \times_n P^{s}_n$ be an equilibrium of the polytope-form game $V^s$ and $r_n: \mathbb{R}^{d_n} \rightarrow P^{s}_n$ the nearest-point retraction. Define $z_{nx_n} \equiv \bar{V}_n(x_n,\sigma_{-n}) + g_{nx_n} + \sigma_{nx_n}$. Then $r_{nx_n}(z_n) = \sigma_{nx_n}$. Indeed, the variational inequality \ref{projection} characterizes a unique $r_n(z_n)$:
\begin{equation}\label{projection}
 \langle\tau_n - r_n(z_n), z_{n} - r_{n}(z_n) \rangle \leq 0, \forall \tau_n \in P^{s}_n.
 \end{equation} 

Then, if $\sigma$ is an equilibrium, by definition it implies $\langle \tau_n - \sigma_n, V^{s,n}(\sigma_{-n}) \rangle \leq 0$, for all $\tau_n \in P^{s}_n$. Rewriting this inequality as $\langle\tau_n - \sigma_n, \bar{V}^{s,n}(\sigma_{-n}) + g_{n} + \sigma_n - \sigma_n \rangle \leq 0$ shows $r_{n}(z_n) = \sigma_n$. 

Using the decomposition of $V^s = \bar{V}^s \bigoplus g$, we rewrite the mapping $\theta^s$ from the equilibrium graph $\mathcal{E}^s =\{(V^s,\sigma) \in \Re^{ND} \times P^s \mid \sigma$ is an equilibrium of $V^s$ $\}$ to $\Re^{ND}$ by $\theta^s(\bar{V}^s,g,\sigma)= (\bar{V}^s ,z)$. We show that $\theta^s$ is a homeomorphism. First, $\theta^s$ is clearly continuous. Also, the inverse homeomorphism can be defined explicitly: $h: \Re^{ND} \rightarrow \mathcal{E}^s$ with $h(\bar{V}^s, z) = (\bar{V}^s, g, r(z))$, where $g \equiv (g_1,..,g_N)$, $g_n \equiv z_{n} - \sigma_n - \bar{V}^{s,n}(\sigma_{-n})$ and $\sigma_m \equiv r_m(z_m), \forall m \in \N$. It follows that $h \circ \theta = id_{\mathcal{E}^s}$ and  $\theta \circ h = id_{\Re^{ND}}$.

Let $\mathbb{S}^{ND}$ denote the $ND$-dimensional sphere and recall that $\mathbb{S}^{ND}$ is homeomorphic to $\Re^{ND}\cup \{\infty\}$ with the one-point compactification topology. Let $\overline{\mathcal{E}^s}$ denote the one-point compactification of $\mathcal{E}^s$.

We define a homotopy $H: [0,1] \times \mathbb{S}^{ND} \rightarrow \mathbb{S}^{ND}$ by $H(t,V^s) = H(t,\bar{V}^s, z) = (\bar{V}^s, tz + (1-t)g)$, if $V^s \in \Re^{ND}$, and $H(t,\infty) = \infty$. Since $\text{proj}_{\Re^{ND}}$ and $h$ are both continuous and proper mappings, they have continuous extensions $\overline{\text{proj}}_{\Re^{ND}}: \overline{\mathcal{E}^s} \rightarrow \mathbb{S}^{ND}$  and $\overline{h}: \mathbb{S}^{ND} \rightarrow \overline{\mathcal{E}^s}$ to the one-point compactifications, both taking $\infty$ to $\infty$. Notice now that $H(0, \cdot) = (\overline{\text{proj}}_{\Re^{ND}} \circ \overline{h}): \mathbb{S}^{ND} \rightarrow \mathbb{S}^{ND}$ and $ H(1,\cdot ) = id_{\mathbb{S}^{ND}}$. We now show that $H$ is continuous, which shows $H$ is indeed a homotopy. Continuity at points $(t,V^s)$ where  $V^s \neq \infty$ is immediate from the definition, since the homotopy is linear. It remains to show the continuity of $H$ at all points $[0,1] \times \{\infty\}$ or equivalently $\forall R >0$, $\exists S >0$ such that if $||(\bar{V^s}, z)||_{\infty} \geq S$, it implies $\forall t$, $|| H(t,\bar{V}^s,z) ||_{\infty} \geq R$. 

Note that the definition of $g$ implies that $|z_{n x_n} - g_{n x_n}| \leq |\sigma_{n x_n}| + |\bar{V}^s_{n}(x_n, \sigma_{-n})|$, where for all $n$ $\sigma_{nx_n} \equiv r_{nx_n}(z_n)$. Because $\sigma_{nx_n} = \text{sign}(\sigma_{nx_n})|\sigma_{nx_n}|$, it implies $$\bar{V^s}_n(x_n,\sigma_{-n}) =  \sum_{x_j : j \neq n}\bar{V^s}_n(x_n, x_{-n})\Pi_{j \neq n}\sigma_{jx_j} =  \sum_{x_j : j \neq n}\bar{V^s}_n(x_n, x_{-n})\frak{h}_n(x_{-n})\Pi_{j \neq n}|\sigma_{jx_j}|$$ 

with $\frak{h}_n(x_{-n}) = \Pi_{m \neq n}\text{sign}(\sigma_{mx_m})$, which implies that $|\frak{h}_n(x_{-n})| = 1$. Let $|| \Pi_{m \neq n}|\sigma_m| ||_{\infty} \equiv \text{sup}\{\Pi_{m \neq n}|\sigma_{mx_m}| : m \in \N, x_m \in \mathbb{R}^{d_m} \}$. Because $P^s_m$ is a polytope, there exists  $\alpha > 0$ such that for all $\sigma_m \in P^s_m$, we have  $||\Pi_{m \neq n}|\sigma_m| ||_{\infty} \leq \alpha$, for all $n$. We can assume without loss of generality $\alpha > 1$. Also, there exists $B>1$ such that for all $m \in \N$ and $\sigma_m \in P^{s}_m$ it implies $|| \sigma_m ||_{\infty} \leq B$. Therefore, we have that for some $C>1$:
\begin{equation}\label{structure equation} 
||z_n - g_n||_{\infty} \leq ||\bar{V}||_{\infty}C\alpha + B. 
\end{equation} 
Let  $R>0$. Set $S = 2RC\alpha +B$. If $||(\bar{V},z)||_{\infty} \geq S$, then either $||\bar{V}||_{\infty} \geq R$ (in which case $||H(t,\bar{V},z)||_{\infty} \geq R$), or $||\bar{V}||_{\infty} < R$ and $||z||_{\infty} \geq 2RC\alpha + B$.  Using \ref{structure equation} we have  $|| tz + (1-t)g ||_{\infty} \geq R$, which finishes the proof. 
\medskip
\medskip

Throughout the rest of the Appendix we will use singular homology. Given $Y$ a topological space, we denote by $H_m(Y)$ the $m$-th dimensional singular homology group of $Y$. Given a continuous map $f: X \to Y$, $f_*$ denotes the morphism between singular homology groups of $X$ and $Y$.

\subsection{Degree Theory: Formal Definitions and Auxiliary Results}\label{degreetheory}
Let $P_n$ be a polytope of dimension $d_n -1$, for each player $n$, with $P \equiv \times_n P_n$. Let $D \equiv d_1...d_N$ and $\mathcal{E}^{PF} \equiv \{(V, p) \in \times_n A(P) \times P \mid $ $p$ is an equilibrium of $V \}$, where $A(P)$ is the $D$-dimensional linear space of multiaffine functions from $P$ to $\mathbb{R}$ -- where the linear space structure is given by pointwise addition and scalar multiplication. The linear space $\times_n A(P)$ is a $ND$-dimensional Euclidean space and we denote its one-point compactification by $\overline{\times_{n} A(P)}$. Recall that the one-point compactification $\overline{\times_{n} A(P)}$ is homeomorphic to the sphere $\mathbb{S}^{ND}$.

Let $(V, Q) \in \mathcal{E}^{PF}$. Let $U \subset \mathcal{E}^{PF}$ be an open neighborhood of $(V,Q)$ whose closure in $\E^{PF}$ contains no pair $(V,p)$ not already in $(V,Q)$. The \textit{local degree of $proj|_{U}: U \rightarrow \times_nA(P)$ over $V$} is the integer $\text{deg}_{V}(\text{proj}|_{U})$ that defines the following homomorphism in singular homology: 
$$(\text{proj}|_U)_{*} : H_{ND}(U,U \setminus \{(V,Q)\}) \rightarrow H_{ND}(\overline{\times_{n} A(P)}), \overline{\times_{n} A(P)} \setminus \{V\}),$$ 
where $H_{ND}(U,U \setminus \{(V,Q)\})$ is oriented according to the following composition of homomorphisms:
$$\mathbb{Z} = H_{ND}(\overline{\times_{n} A(P)}) \rightarrow H_{ND}(\overline{\times_{n} A(P)}, \overline{\times_{n} A(P)} \setminus K) \rightarrow H_{ND}(W,W \setminus K) \rightarrow H_{ND}(U,U \setminus \{(V,Q)\}), $$
	where $K = \theta^{PF}(V,Q)$ and $W = \theta^{PF}(U)$; the first and second arrows correspond to inclusion, where the second is an isomorphism by excision, and the third is the isomorphism $(\theta^{PF})^{-1}_{*}$.


\begin{definition}\label{def degree}
Let $V$ be a polytope-form game and $Q$ be a component of equilibria of $V$. Let $U \subset \mathcal{E}^{PF}$ be a neighborhood of $(V,Q) \in \E^{PF}$ whose closure (in $\E^{PF})$ contains no pair $(V,p)$ not already in $(V,Q)$. Then the \textit{degree of $Q$ w.r.t. $V$}, denoted $\text{deg}^{PF}_V(Q)$, is defined as $\text{deg}_{V}(\text{proj}|_{U})$. 
\end{definition}

The next proposition tells us why the degree is relevant for identifying components of equilibria which are robust to payoff perturbations. 

\begin{proposition}\label{robustness degree}
Let $V=(\N, (P_n)_{n \in \N}, (V_n)_{n \in \N})$ be a polytope-form game and $Q$ an equilibrium component of $V$. Fix a $U \subset \mathcal{E}^{PF}$ a neighborhood of $(V,Q)$ whose closure (in $\E^{PF})$ contains no pair $(V,p)$ which is not already in $(V,Q)$. Assume $\text{deg}_{V}(Q) \neq 0$. Then there exists a neighborhood $W \subset \times_n A(P)$ of $V = (V_n)_{n \in \N}$ such that for any $V' \in W$, there exists an equilibrium $p'$ of $V'$ with $(V',p') \in U$.
\end{proposition}
\begin{proof} See Proposition 5.4 , Ch. IV, \cite{D1972}.\end{proof}

\subsection{Index Theory: Definitions and Auxiliary Results}\label{indextheorygeneral}

\begin{definition}
A \textit{Nash map} is a continuous function $f: \times_n A(P) \times P \rightarrow P$ such that for each $V \in \times_n A(P)$, the fixed points of its restriction $f_{V}$ to $\{V\}\times P$, viewed as a map from $P$ to itself, are the Nash equilibria of $V$.
\end{definition}

\begin{definition}
Fix $Q \subset W \subset \mathbb{R}^{m}$ where $Q$ is compact and $W$ is open. Recall that the one-point compactification of $\mathbb{R}^{m}$ is (homeomorphic to) an $m$-dimensional sphere. The \textit{fundamental class} $\mathcal{O}_{Q} \in H_m(W,W \setminus Q)$ is the image of $1 \in \mathbb{Z}$ under the composition
$$\mathbb{Z} = H_m(\mathbb{S}^m) \rightarrow H_m(\mathbb{S}^{m}, \mathbb{S}^{m}  \setminus Q) \rightarrow H_m(W,W \setminus Q)$$

where the first and second arrows are the homomorphisms induced by inclusion and the second one is an isomorphism, by excision.  
\end{definition}

\begin{remark}
The fundamental class $\mathcal{O}_Q$ does not depend on $W$: if $U$ is any other neighborhood of $Q$ in $\mathbb{R}^{m}$, then the two inclusion maps from $(W\cap U, (W \cap U) \setminus Q)$ into $(W,W \setminus Q)$ and $(U,U \setminus Q)$ send $\mathcal{O}_Q$ to itself.
\end{remark}

Let $P_n \subset \mathbb{R}^{d_n}$ be a polytope. Define $J \equiv \times_n (P_n) $, where $(P_n)$ is the affine space generated by $P_n$, and let $J_0 \equiv J - J$, where the symbol ``$ - $'' denotes the subtraction in the vector space $\times_n \mathbb{R}^{d_n}$. The space $J_0$ is the unique linear subspace of $\times_n \mathbb{R}^{d_n}$ that is parallel to $J$ and of the same dimension as $J$.  Both $J$ and $J_0$ are homeomorphic to a Euclidean space. Therefore the definition of fundamental class applies in the obvious way to compact subsets of these spaces. Let $r$ be a retraction of $J$ onto $P \equiv \times_n P_n$. Then every Nash map $f$ extends to the map $f \circ (\text{id}_{\times_n A(P)}\times r)$ on $\times_n A(P) \times J$. If $r'$ was another retraction of $J$ onto $P$ then, using a linear homotopy, $r$ is homotopic to $r'$. This implies that the induced homotopy between $f \circ (\text{id}_{\times_n A(P)} \times r)$ and $f \circ (\text{id}_{\times_n A(P)} \times r')$ preserves the set of fixed points.

\medskip
Let $E(V)$ denote the set of equilibria of $V$. 

\begin{definition}\label{def index}
Fix a Nash map $f$, a  polytope-form game $V$, and a component $Q$ of $E(V)$. Choose an open neighborhood $W$ of $Q$ in $J$ disjoint from $E(V) \setminus Q$. Let $d:(W,W \setminus Q) \rightarrow (J_0, J_0 \setminus 0)$ be the displacement map given by $d(p) = p - (f_{V} \circ r)(p)$. The \textit{index of $Q$ under $f$}, denoted $\text{ind}(Q,f)$, is the unique integer $i$ for which $d_*(\mathcal{O}_Q) = i \cdot \mathcal{O}_0$, where $d_*$ is the singular homology homomorphism induced by $d$. \end{definition}

\begin{remark} Note that it is implicit in the notation $\text{ind}(Q,f)$ that $Q$ is a component of fixed points of $f_V$.\end{remark}

One example of a Nash map is the map used in \cite{JN1951} to prove existence of equilibria. Another is the $GPS$-map in \cite{GPS1993}. As seen from Definition \ref{def index}, the Nash-map index of a component of equilibria $Q$ apparently depends on the specific Nash map used to assign the index. We show below in Proposition \ref{inv of nashmap index} that the dependence is just apparent:  a Nash map used to assign indices to a certain equilibrium component \emph{assigns the same index as any other Nash map}. 

Before that we introduce a different concept of index of equilibria, defined directly from the best-reply correspondence. The best-reply index is defined from the best-reply correspondence of a polytope-form game. This notion of index will play an important role in the proof of Proposition \ref{identicaltheories}.

Let $Q$ be a component of equilibria of the polytope-form game $V$ and let $BR^{V}: P \rightrightarrows P$ be the best-reply correspondence of $V$. Let $U$ be open in $P$ and a neighborhood of $Q$ such that its closure $\text{cl}(U)$ in $P$ satisfies $cl(U) \cap E(V) = Q$. Let $O$ be an open neighborhood of $\text{Graph}(BR^{V})$ such that $O \cap \{(\sigma, \sigma) \in P \times P | \sigma \in \text{cl}(U) \setminus U\} = \emptyset$. We call such a neighborhood $O$ an \textit{adequate neighborhood of $\text{Graph}(BR^{V})$ for $Q$}. By Corollary 2 in \cite{AM1989}, there exists $A \subset O$ a neighborhood of $\text{Graph}(BR^{V})$ such that any two continuous functions $f_0$ and $f_1$ from $P$ to $P$ whose graphs are in the neighborhood $A$ are homotopic by a homotopy $G: [0,1] \times P \rightarrow P$ with $\text{Graph}(G(t,\cdot)) \subset O$ for all $t \in [0,1]$. The neighborhood $A$ is called an \textit{adequate homotopy neighborhood for $Q$}.  By Corollary 1 of \cite{AM1989}, there exists a continuous map $f : P \rightarrow P $ with $\text{Graph}(f) \subset A$. We define the \textit{best-reply index of component $Q$}, denoted $\text{Ind}_{BR^{V}}(Q)$, as the fixed point index of the continuous map $f|_{U}: U \to  P$. The choice of the neighborhood $O$ and the homotopy property of the index (see \cite{D1972}, Chapter VII, 5.15) imply that the index of the component is the same for any continuous map with graph in the neighborhood $A$. 

As a result of Proposition \ref{identicaltheories}, the best-reply index and the Nash-map index assigned to a component of equilbria are identical: this identity is shown explicity in Claim \ref{eqindex} for standard polytope-form games and the general identity follows from that Proposition immediately. 

We now show the invariance of the index to the choice of the Nash map. 

\begin{Proposition}\label{inv of nashmap index}
Let $f^{1}, f^2: \times_n A(P) \times P \rightarrow P$ be two Nash maps. Then for any equilibrium component $Q$ of the polytope-form game $V \in \times_n A(P)$, it follows that $\text{ind}(Q,f^1) = \text{ind}(Q,f^2)$.
\end{Proposition}

The proof of Proposition \ref{inv of nashmap index} is performed in steps. First, Lemma \ref{invariance of NM index} establishes the result for standard polytope-form games. Then the proof of Proposition \ref{inv of nashmap index} is presented using this lemma. From now on we fix a standard polytope-form game $V^s = (\N, (P^s_n)_{n \in \N}, (V^s_n)_{n \in \N})$. The space of payoffs is $\Re^{ND}$, where dim$(P^s_n) = d_n -1, \forall n \in \N$, and $D = d_1....d_N$.

\begin{remark}
Given a Nash-map $f: \mathbb{R}^{ND} \times P^s \rightarrow P^s$, we abuse notation slightly and use $f$ to denote the extension $f \circ (\text{id}_{\mathbb{R}^{ND}} \times r)$, where $r \equiv \times_n r_n$ and $r_n : \mathbb{R}^{d_n} \rightarrow P^s_n$ is the nearest-point retraction.
\end{remark}

Let $P^{\e}_n \subset (\Delta_n) = J_n$ be a polytope containing $P^s_n$ in its relative interior. Since $P^s_n$ is standard, $P^{\e}_n$ is also full-dimensional in $J_n$. Let $\Delta^{\e} \equiv \times_{n \in \N} P^{\e}_n$. Denote by $\partial \D^{\e}$ the topological boundary of $\D^{\e}$ in $J \equiv \times_n J_n$. We view the graph $\mathcal{E}^{s}$ as a subset of $\Re^{ND} \times \Delta^{\e}$. Recall that $J_0 = J - J$. The proof of Lemma \ref{invariance of NM index} is an adaptation of an unpublished proof of Govindan and Wilson.

\begin{lemma}\label{invariance of NM index}
Let $f^1$ and $f^2$ be two Nash maps. Then, the two displacement maps $d^1$ of $f^1$ and $d^2$ of $f^2$ are homotopic as maps between the triads $(\mathbb{R}^{ND} \times \Delta^{\e}, \mathcal{E}^{s}, (\mathbb{R}^{ND} \times \Delta^{\e}) \setminus \mathcal{E}^s)$ and $(J_0, 0, J_0 \setminus \{0\})$. Consequently, $\text{ind}(Q, f^1) = \text{ind}(Q, f^2)$.
\end{lemma}

\begin{proof}The proof of this Lemma is obtained from a series of claims. 

\begin{claim}\label{d1-d2}
$d^1 : \mathbb{R}^{ND} \times \partial \Delta^{\e} \rightarrow J_0 \setminus \{0\}$ is homotopic to $d^2: \mathbb{R}^{ND} \times \partial \Delta^{\e} \rightarrow J_0 \setminus \{0\}$. 
\end{claim}

\begin{proof}
Since $f^1$ and $f^2$ map $\mathbb{R}^{ND} \times \partial \Delta^{\e}$ into $P^s$, $d^1$ and $d^2$ are homotopic via the linear homotopy. 
\end{proof}

\begin{claim}\label{def ret}
$\Re^{ND} \times \partial \Delta^{\e}$ is a deformation retract of $(\mathbb{R}^{ND} \times \Delta^{\e}) - \mathcal{E}^s$
\end{claim}

\begin{proof}As in the structure theorem proved in Theorem \ref{ST}, we reparametrize the space of games as $V^s = (\bar{V}^s, g)$. Let $E_0 \subset \Re^{ND}$ be the linear subspace of payoffs vectors containing $\bar{V}^s$. Then $\Re^{ND} = E_0 \times \mathbb{R}^{m}$, where $m = d_1 + ... + d_n$.  Define the function $h$ from $\Re^{ND} \times \Delta^{\e}$ to itself by $h(\bar{V}^s, g, \s) = (\bar{V}^s ,z,\s)$, where $$ z_{nx_n} = \s_{nx_n} + V^s_n(x_n, \s_{-n}),$$

 where we maintain the notation of subsection \ref{polytope-form degree}. It easily follows that $h$ is a homeomorphism that maps $\Re^{ND} \times \partial \Delta^{\e}$ onto itself. Let $r: \mathbb{R}^{m} \rightarrow P^{s}$ be the nearest-point retraction. Denoting $R$ the graph of $r$, we have that $h(\mathcal{E}^s) = E_0 \times R$. It is enough to prove therefore that $\Re^{ND} \times \partial \Delta^{\e}$ is a deformation retract of $(\Re^{ND} \times \Delta^{\e}) \setminus (E_0 \times R)$. We can construct a retraction $\psi$ of $(\Re^{ND} \times \Delta^{\e}) \setminus (E_0 \times R)$ onto $\Re^{ND} \times \partial \Delta^{\e}$ as follows. First, given a pair $(z,\s) \in (\mathbb{R}^{m} \times \Delta^{\e}) \setminus R$, let $\eta (z,\s)$ be the unique point in $\partial \Delta^{\e}$ that lies on the ray emanating from $r(z)$ and passing through $\s$. Then define $\psi(\bar{V}^s,z,\s) = (\bar{V}^s,z,\eta(z,\s))$. The map $\psi$ is easily seen to be a retraction. Let $i_R: \Re^{ND} \times \partial \Delta^{\e} \rightarrow (\Re^{ND} \times \Delta^{\e}) \setminus (E_0 \times R)$ be the inclusion map. Then $i_R \circ \psi$ is homotopic to the identity map using the linear homotopy. Therefore, $\Re^{ND} \times \partial \Delta^{\e}$ is a deformation retract of $(\Re^{ND} \times \Delta^{\e}) \setminus (E_0 \times R)$.\end{proof}

\begin{claim}\label{homotopyclaim}
$d^1: (\Re^{ND} \times \Delta^\e) \setminus \mathcal{E}^s \rightarrow J_0 \setminus \{0\}$ is homotopic to $d^2: (\Re^{ND} \times \Delta^\e) \setminus \mathcal{E}^s \rightarrow J_0 \setminus \{0\}$.
\end{claim}

\begin{proof}Let $\text{id}$ be the identity map on $(\Re^{ND} \times \Delta^\e) \setminus \mathcal{E}^s$ and let $j_{\mathcal{E}^s}: \Re^{ND} \times \partial \Delta^\e \subset (\Re^{ND} \times \Delta^\e) \setminus \mathcal{E}^s$. By Claim \ref{def ret}, there exists a retraction $\phi$ from $(\Re^{ND} \times \Delta^\e) \setminus \mathcal{E}^s$ to $\Re^{ND} \times \partial \Delta^\e$ such that $\text{id}$ is homotopic to $j_{\mathcal{E}^s} \circ \phi$. Therefore, for $i = 1,2$, $(d^i \circ \text{id})$ is homotopic to $(d^i \circ j_{\mathcal{E}^s} \circ \phi): (\Re^{ND} \times \Delta^\e) \setminus \mathcal{E}^s  \rightarrow J_0 \setminus \{0\}$. By Claim \ref{d1-d2}, the restrictions of $d^1$ and $d^2$ to $\Re^{ND} \times \partial \Delta^\e$ are homotopic. Therefore, $d^1$ is homotopic to $d^2: (\Re^{ND} \times \Delta^\e) \setminus \mathcal{E}^s \rightarrow J_0 \setminus \{0\}$. \end{proof}

We now construct the homotopy of Lemma \ref{invariance of NM index}. Let $\Phi$ be the homotopy of Claim \ref{homotopyclaim} between the restrictions of $d^1$ and $d^2$ to $(\Re^{ND} \times \Delta^\e) \setminus \mathcal{E}^s$. It is readily checked from the constructions above that $\Phi((\Re^{ND} \times \Delta^\e) \setminus \mathcal{E}^s \times [0,1])$ is a bounded subset of $J_0$. By Urysohn's Lemma, there exists a continuous function $\alpha: \Re^{ND} \times \Delta^\e \rightarrow [0,1]$ such that $\alpha^{-1}(0) = \mathcal{E}^s$. Define then $G: (\Re^{ND} \times \Delta) \setminus \mathcal{E}^s \times [0,1] \rightarrow J_0 \setminus \{0\}$ by $G(x,t) = \alpha(x)\Phi(x,t)$. The image of $\Phi$ being bounded, $G$ has a continuous extension to a map from $\Re^{ND} \times \Delta^\e$ to $J_0$ that maps $\mathcal{E}^s$ to $0$. The result then follows from the observation that for $i =1,2$, $d^i$ is homotopic to $G(\cdot, i-1)$.\end{proof}

\begin{proof}[Proof of Proposition \ref{inv of nashmap index}]
Let $P^s_n$ be a standard polytope such that there exists an affine bijection $e_n: P_n \rightarrow P^s_n$ and let $e \equiv \times_n e_n$. Let $T_n : A(P) \rightarrow \mathbb{R}^{D}$ be defined from $e$ as in subsection \ref{polytope-form degree} and $T \equiv \times_n T_n$. Let $f'_i \equiv e \circ f_i \circ (T^{-1} \times e^{-1}): \Re^{ND} \times P^s \rightarrow P^s, i = 1,2$. Then $f'_i$ is a Nash-map and Lemma \ref{invariance of NM index} shows that the indices of equilibrium component $e(Q)$ according to $f'_1$ and $f'_2$ are the same. Let $V^s$ be a standard polytope-form game such that $T^{-1}(V^s) = V$. Considering the restriction $(f'_i)_{V^s}: P^s \rightarrow P^s$ we have that  $(f'_i)_{V^s} = e \circ (f_i)_{V} \circ e^{-1}$, by definition. Now define $h_i   \equiv (f_i)_{V} \circ e^{-1}$. The commutativity property of the index (see \cite{D1972}, Chapter VII, Theorem 5.14) shows that the fixed-point sets of $h_i \circ e$ and $e \circ h_i$ are homeomorphic and the index of each fixed-point component is the same under these two maps.  Since $e \circ h_i = (f'_i)_{V^s}$ and $h_i \circ e = (f_i)_{V}$, this shows that the index of $Q$ under $(f_i)_{V}$ is the same as the index of $e(Q)$ under $(f'_i)_{V^{s}}$. Since $i$ is arbitrary,  it follows therefore that $\text{ind}(Q,f_1) = \text{ind}(Q, f_2)$.\end{proof}

\subsection{Proof of Proposition \ref{identicaltheories}}The proof is a direct consequence of four claims, which we now prove. The notation of subsection \ref{equivdegindtheories} is maintained: let $e_n: P_n \to P^s_n$ be an affine isomorphism and $T$ defined from $e$ as in subsection \ref{polytope-form degree}. Recall the definition of the map $e^{PF}: \E^{PF} \to \E^s$, defined from $T$ and $e$ from the proof of Proposition \ref{NSST}. Recall $\theta^s: \E^s \to \Re^{ND}$ is the homeomorphism from Lemma \ref{ST}; $\theta^{PF}: \E^{PF} \to \times_n A_n(P)$ is the homeomorphism of Proposition \ref{NSST}. 

\begin{claim}\label{standarddeg}
The following equation holds: $\text{deg}^{PF}_{V}(Q) = \text{deg}^{PF}_{T(V)}(e(Q))$.
\end{claim}

\begin{proof}Note first that $e(Q) \subset P^s$ is an equilibrium component of $T(V)$. Let $W$ be an open neighborhood in the graph $\mathcal{E}^{PF}$ of $K \equiv \{V\} \times Q$ such that the closure $\text{cl}(W)$ in $\E^{PF}$ contains no other pair $(V, p)$ besides those in $K$. Then it follows by definition that the degree of $Q$ with respect to $V$ is the local degree of the mapping $\text{proj}_V|_{W}: W \rightarrow \overline{\times_nA(P)} $ over $V$. Let $A \equiv e^{PF}(W)$, which is open because $e^{PF}$ is a homeomorphism. Notice that $K^s \equiv \{T(V)\} \times e(Q)$ is a compact subset of $A$ such that the closure $\text{cl}(A)$ w.r.t. to $\E^s$ contains no pair $(T(V), \s) \in \mathcal{E}^s$ in its boundary. Therefore, the degree of $e(Q)$ of game $T(V)$ is the local degree of  $\text{proj}_{\Re^{ND}}|_{A}: A \rightarrow \mathbb{S}^{ND}$ over $T(V)$. 

Let $A^s_z = \theta^s(A)$, $W_z = \theta^{PF} (W)$, $K^s_z = \theta^s(K^s)$ and $K_z = \theta^{PF}(K)$. Since $T$ is a homeomorphism, we can orient $\overline{\times_nA(P)}$ according to $T$ from $\mathbb{S}^{ND}$. This gives that the first vertical homomorphism in the diagram below has degree $+1$, by definition. Using the long exact sequences of the pair to choose orientations for $(\overline{\times_nA(P)}, \overline{\times_nA(P)} \setminus \{V\}$ and $(\mathbb{S}^{ND}, \mathbb{S}^{ND} \setminus \{T(V)\})$, the natural property of the long exact sequence now implies that $T: (\overline{\times_nA(P)}, \overline{\times_nA(P)} \setminus \{V\}) \to  (\mathbb{S}^{ND}, \mathbb{S}^{ND} \setminus \{T(V)\})$ has degree $+1$. Observe now that the horizontal sequences of the diagram below come from the definition of the local degree. Let $\mu$ be the element of $H_{ND}(A, A \setminus K^s)$ obtained as the image of the generator of $H_{ND}(\mathbb{S}^{ND})$ under the horizontal sequence of homomorphisms, and let $\nu$ be the analogous element of $H_{ND}(W, W \setminus K)$. The map $j$ denotes the inclusion. We show $(e^{PF})_* \nu = \mu$. For that, it is sufficient to show that the diagram below commutes: the first square of the diagram commutes by naturality of the long exact sequence; the second square commutes by functoriality of homology (since the isomorphisms are inclusions), and the third square commutes because of functoriality of homology and $(e^{PF})^{-1} \circ (\theta^s){-1} = (\theta^{PF})^{-1} \circ T^{-1}$. This gives that $(e^{PF})_* \nu = \mu$. 
\medskip
\begin{tikzcd}
 H_{ND}(\mathbb{S}^{ND}) \arrow[d, " (T)^{-1}_{*}"] \arrow[r, "j_*"] &  H_{ND}(\mathbb{S}^{ND}, \mathbb{S}^{ND} \setminus K^s_z) \arrow[r, "\simeq"] \arrow[d, "(T)^{-1}_{*}"] & H_{ND}(W, W \setminus K^s_z) \arrow[d, "(T)^{-1}_{*}"]  \arrow[r, "(\theta^s)^{-1}_*"] & H_{ND}(A, A \setminus K^s) \arrow[d, "(e^{PF})^{-1}_{*}"]  \\
 H_{ND}(\overline{\times_nA(P)}) \arrow[r, "j_*"]  & H_{ND}(\overline{\times_nA(P)}, \overline{\times_nA(P)} \setminus K_z) \arrow[r, "\simeq"] & H_{ND}(W_z, W_z \setminus K_z) \arrow[r, "(\theta^{PF})^{-1}_{*}"] & H_{ND}(W, W \setminus K)
\end{tikzcd}
\medskip
Notice now that $T \circ \text{proj}_{V}|_{W} = \text{proj}_{\Re^{ND}}|_{A} \circ e^{PF}|_{W}$. Hence the diagram below commutes by functoriality of homology. 
\medskip
\begin{center}
\[
\begin{tikzcd}
 H_{ND}(A, A \setminus K^s) \arrow{r}{(\text{proj}_{\Re^{ND}})_{*}} \arrow{d}{(e^{PF})^{-1}_{*}} & H_{ND}(\mathbb{S}^{ND}, \mathbb{S}^{ND} \setminus \{T(V)\}) \arrow{d}{(T)^{-1}_{*}} \\
  H_{ND}(W, W \setminus K)  \arrow{r}{(\text{proj}_{V})_{*}} & H_{ND}(\overline{\times_nA(P)}, \overline{\times_nA(P)} \setminus \{V\})
\end{tikzcd}
\]
\end{center}
Notice now that map $T^{-1}: (\mathbb{S}^{ND}, \mathbb{S}^{ND} \setminus \{T(V)\}) \to (\overline{\times_n A(P)}, \overline{\times_n A(P)} \setminus \{V\})$ has degree $+1$, by construction. We have showed the homology-induced map from $e^{PF}$ sends $\nu$ to $\mu$. Therefore, $\text{deg}_V(Q) = \text{deg}_{T(V)}(e(Q))$.\end{proof}

\begin{claim}\label{eq index and deg}
The following equation holds: $\text{deg}^{PF}_{T(V)}(e(Q)) = \text{ind}_{T(V)}(e(Q))$.
\end{claim}

\begin{proof}Denote $V^s \equiv T(V)$. Let $g \in \times_n \mathbb{R}^{d_n}$.  Consider the polytope-form game $V^s \oplus g$ whose payoff function is defined by $$(V^s \oplus g)_ n(\sigma) = \sigma_n \cdot V^{s,n}(\sigma_{-n}) + \sigma_n \cdot g_n,$$ and player $n$ strategy set is $P^s_n$. Let $\mathcal{E}_{V^s} = \{(g,\sigma)  \in \times_n\mathbb{R}^{d_n} \times P^s \mid \sigma$ is an equilibrium of $V^s \bigoplus g \}$ be the graph of equilibria over the restricted class of perturbations $g$. Define $\Theta: \mathcal{E}_{V^s}\rightarrow \times_n\mathbb{R}^{d_n}$ by its coordinate functions $\Theta_n(g,\sigma) = \sigma_n + V^{s,n}(\sigma_{-n}) + g_n \in \mathbb{R}^{d_n}$ and $\text{proj}_{g}: \mathcal{E}_{V^s} \rightarrow \times_n \mathbb{R}^{d_n}$,  be the projection over the first coordinate.  Then $\Theta$ is a homeomorphism and $\Theta^{-1}(z) = (d_{\Psi_{V^s}}(z),r(z))$; note that $\text{proj}_g \circ \Theta^{-1} = id - w_{V^s} \circ r = id - \Psi_{V^s} = d_{\Psi_{V^s}}$. (Recall that $\Psi_{V^s}$ is the commuted $GPS$-map defined in subsection \ref{polytope-form index}). 

The map $\Theta$ allows us to provide an orientation to the one-point compactification $\overline{\mathcal{E}_{V^s}}$ according to which the degree of $\text{proj}_{g}: \mathcal{E}_{V^s} \rightarrow \times_n \mathbb{R}^{d_n}$ is $+1$. 
Let $e(Q)$ be an equilibrium component of $V^s$ and consider $U$ an open neighborhood in $\mathcal{E}_{V^s}$ containing $\{0\} \times e(Q)$ and no other pair $(0,\sigma)$ in the boundary of $U$. Then the local degree of $\text{proj}_{g}|_{U}$ over $\{0\}$ is well defined. Letting $U_z =\Theta(U)$, since $(\text{proj}_{g}|_{U} \circ \Theta^{-1})|_{U_z} = d_{\Psi_{V^s}}|_{U_z}$, the local degree of $(\text{proj}_{g}|_{U} \circ \Theta^{-1})|_{U_z}$ over $\{0\}$ equals ind$(e(Q), \Psi_{V^s}) = \text{ind}_{V^s}(e(Q))$.

Now we show that the local degree of $\text{proj}_{g}|_{U}$ over $0$ equals the degree of $e(Q)$ w.r.t. $V^s$. Recalling subsection \ref{polytope-form degree}, we can decompose $V^s$ and write $V^s = (\tilde{V}^s, g)$. Let $\mathcal{E}_{\tilde{V^s}}=\{ (g^{'}, \sigma) \mid ((\tilde{V},g^{'}), \sigma) \in \mathcal{E}^s \}$. From $\theta^s$ we can define another homeomorphism $\tilde{\Theta}: \mathcal{E}_{\tilde{V}^s} \rightarrow \times_n\mathbb{R}^{d_n} $ by $\tilde{\Theta}(g^{'},\sigma) = z$, where $z$ satisfies $\theta^s(\tilde{V}^s, g^{'}, \sigma) = (\tilde{V}^s,z)$. 

Let $\tilde U$ be an open neighborhood in $\mathcal{E}^s$ of $(\tilde{V}^s,g,e(Q)) \in \mathcal{E}^s$ such that $\text{cl}_{\E^s}(\tilde U)$ has no other point $(\tilde{V}^s,g,p)$ besides those in $(\tilde{V}^s,g,e(Q))$. Let $\tilde{U}_z \equiv \theta^s(\tilde U)$. The local degree of $\text{proj}_{\Re^{ND}} \circ (\theta^s)^{-1}|_{\tilde{U}_z}$ over $ V^s = \tilde{V}^s \oplus g$ is then well defined. There exist $U_1$ an open neighborhood of $\tilde{V}^s$  and $U_{2}$ an open neighborhood $\tilde{\Theta}(g,e(Q))$ such that $U_1 \times U_2 \subset \tilde{U}_z$. The local degree of $\text{proj}_{\Re^{ND}} \circ (\theta^s)^{-1}|_{\tilde{U}_z}$ over $V^s$ is equal to the local degree of $\text{proj}_{\Re^{ND}} \circ (\theta^s)^{-1}|_{U_1 \times U_2}: U_1 \times U_2 \rightarrow \mathbb{S}^{ND}$ over $V^s$ - according to Proposition 5.5, Chapter IV in \cite{D1972}. Consider the map $(\text{id} \times \text{proj}_{g} \circ \tilde{\Theta}^{-1})|_{U_1 \times U_2}: U_1 \times U_2 \rightarrow \mathbb{S}^{ND}$ defined by $(\text{id} \times \text{proj}_{g} \circ \tilde{\Theta}^{-1})(\tilde{V}^s, z) = (\tilde{V}^s, \text{proj}_{g} \circ (\tilde{\Theta})^{-1}(z))$. We have therefore that $\text{proj}_{\Re^{ND}} \circ (\theta^s)^{-1}(\tilde{V}^{'}, z^{'}) = (\tilde{V}^{'}, (\theta^s_2)^{-1}(\tilde{V}^{'}, z^{'}))$, where $(\theta^s)^{-1}(\tilde{V}^{'},z') = (\tilde{V}^{'}, (\theta^s_2)^{-1}(\tilde{V}^{'},z'), r(z'))$ and similarly $(\text{id} \times \text{proj}_{g} \circ (\tilde{\Theta})^{-1}))(\tilde{V}^{'},z^{'}) = (\tilde{V}^{'}, (\theta^s_2)^{-1}(\tilde{V}^s, z^{'}))$. Note that in the above expression of $\text{id} \times \text{proj}_{g} \circ (\tilde{\Theta})^{-1}$ we have that the second coordinate function $(\theta^s_2)^{-1}$ fixes the argument $\tilde{V}^s$. 

Let $H: [0,1] \times U_1 \times U_2 \rightarrow \Re^{ND}$ be defined by $H(t, \tilde{V}^{'},z^{'})= (\tilde{V}^{'}, t(\theta^s_2)^{-1}(\tilde{V}^{'},z^{'}) + (1-t)(\theta^s_2){-1}(\tilde{V}, z^{'}))$. By the homotopy property of the degree (see \cite{RB1993}, Theorem 9.5), it follows that the local degree over $V^s$ of $\text{id} \times \text{proj}_{g} \circ (\tilde{\Theta})^{-1}|_{U_1 \times U_2}$ and that of $\text{proj}_{\Re^{ND}} \circ (\theta^s)^{-1}|_{U_1 \times U_2}$ is the same. Finally, Theorem 9.7 in \cite{RB1993} implies that the local degree of $\text{id} \times \text{proj}_{g} \circ \tilde{\Theta}^{-1}|_{U_1 \times U_2}$ equals the local degree of $\text{proj}_{g} \circ (\tilde{\Theta})^{-1}|_{U_2}$ over $g$. This proves that $\text{deg}_{V^s}(e(Q))$ is equal to the local degree of $\text{proj}_{g} \circ (\tilde{\Theta})^{-1}|_{U_2}$ over $g$.


We now finish the proof by showing that the local degree of $\text{proj}_{g} \circ \tilde{\Theta}^{-1}|_{U_2}$ over $g$ equals the local degree $\text{proj}_{g} \circ \Theta^{-1}|_{U_z}$ over $0$. This concludes the proof, since it immediately implies that deg$_{V^s}(e(Q)) = \text{ind}_{V^s}(e(Q))$.

Fix $W_2$ an open neighborhood of $\tilde{\Theta}(\{g\} \times e(Q))$ such that $W_2 \subset U_2$. Let $d^{'} = \text{proj}_{g} \circ \tilde{\Theta}^{-1}|_{W_2}$. We have that $\tilde{\Theta}(\{g\} \times e(Q)) = \Theta(\{0\} \times e(Q))$ and $d_{\Psi_{V^s}} = d^{'} - g$. Let $\frak{g}: \times_{n}\Re^{d_n} \rightarrow \times_n \Re^{d_n}$ be defined by $\frak{g}(z) = z-g$. Then $\frak{g}(d^{'}(x)) = d^{'}(x) - g = d_{\Psi_{V^s}}(x)$. Since the degree of $\frak{g}_*: H_{ND}(\Re^{ND}, \Re^{ND} - \{g\}) \rightarrow H_{ND}(\Re^{ND}, \Re^{ND} -\{0\})$ is +1, it follows that the local degree of $d_{\Psi_{V^s}}|_{W_2}$ over $0$ equals the local degree of $d^{'}|_{W_2}$ over $g$, which concludes the proof. \end{proof} 

\begin{claim}\label{eqindex}
The following equation holds: $\text{ind}_{BR^{T(V)}}(e(Q)) = \text{ind}_{T(V)}(e(Q))$.
\end{claim}

\begin{proof}Let $V^s \equiv T(V)$, with standard polytope strategy-set for player $n$ equal to $P^s_n$ of dimension $d_n -1$. Let $\mathbb{V}(P^s_n)$ be the set of vertices of $P^s_n$ and let $N_{P^s_n}(v_n)$ be the normal cone to $P^s_n$ at $v_n \in \mathbb{V}_n(P^s_n)$, defined by $N_{P^s_n}(v_n) := \{ d \in \mathbb{R}^{d_n} | d \cdot (v_j - v_n) \leq 0, \forall v_j \in \mathbb{V}(P^s_n) \}$.  The union of the normal cones over the vertices equals $\mathbb{R}^{d_n}$ and induces a polyhedral subdivision of $\mathbb{R}^{d_n}$ called the \textit{normal fan} of the polytope $P^s_n$. The maximal-dimensional cells of this subdivision are the normal cones (see \cite{GZ2012}, p. 206). We prove next an auxiliary Lemma in order to provide a proof of Claim \ref{eqindex}.

\begin{Lemma}\label{auxclaim} Fix $\s \in P^s$ and $n \in \N$. There exists $\lambda_0 >0$ such that for each $\lambda > \lambda_0$ the nearest-point retraction $r_n(\s_n + \lambda V^{s,n}(\s_{-n})) \in BR^{V^s}_n(\s_{-n})$.\end{Lemma}

\begin{proof}[Proof of Lemma \ref{auxclaim}] Notice first that  $V^{s,n}(\s_{-n})$ is the gradient of the affine function $f_n: \mathbb{R}^{d_n} \rightarrow \mathbb{R}$ defined by $f_n(\s_n) =  \s_n \cdot V^{s,n}(\s_{-n})$. If $v_n \in \mathbb{V}_n(P^s_n)$ is a maximum for the problem $\max_{\s_n \in P^s_n}f_n(\s_n)$, $v_n$ can be characterized as follows:
\begin{equation} v_n \in \text{argmax}_{\s_n \in P^s_n}f_n(\s_n) = BR^{V^s}_n(\s_{-n}) \iff V^{s,n}(\s_{-n}) \in N_{P^s_n}(v_n) \end{equation}

 

Because the union of the normal cones at the vertices of $P^s_n$ is $\mathbb{R}^{d_n}$, there exists $\tilde{v} \in \mathbb{V}(P^s_n)$ such that  $V^{s,n}(\s_{-n}) \in N_{P^s_n}(\tilde{v})$. Assume first that $V^{s,n}(\s_{-n}) \in \text{int}(N_{P^s_n}(\tilde{v}))$. Then for $\lambda > 0$ sufficiently large $\frac{\s_n - \tilde{v}}{\lambda} + V^{s,n}(\s_{-n}) \in N_{P^s_n}(\tilde{v})$. This implies that $\lambda(\frac{\s_n - \tilde{v}}{\lambda} + V^{s,n}(\s_{-n})) = \s_n - \tilde{v} + \lambda V^{s,n}(\s_{-n}) \in N_{P^s_n}(\tilde{v})$, which implies by definition of the normal cone that 
\medskip
\begin{equation}
\langle \s_n + \lambda V^{s,n}(\s_{-n}) - \tilde{v}, p' - \tilde{v}  \rangle \leq 0, \forall p' \in P^s_n. 
\end{equation}
 \medskip

Therefore, $r_n(\s_n + \lambda V^{s,n}(\s_{-n})) = \tilde{v} \in BR^{V^s}_n(\s_{-n})$. Now, if $V^{s,n}(\s_{-n})$ is not in the interior of any cone, then it belongs to the intersection of some cones: assume therefore $V^{s,n}(\s_{-n}) \in \bigcap^{k_n}_{i=1} N_{P^s_n}(\tilde{v}_i)$. We want to show that for $\lambda >0$ sufficiently large $r_n(\s_n + \lambda V^{s,n}(\s_{-n})) = \sum^{k_n}_{i =1} \alpha_i \tilde{v}_i$, where $\alpha_i \geq 0$, $\sum^{k_n}_{i =1} \alpha_i = 1$, which implies that $r_n(\s_n + \lambda V^{s,n}(\s_{-n})) \in BR^{V^s}_{n}(\s_{-n})$. For that purpose, we state two properties which can be easily checked:

\begin{enumerate}

\item $\langle V^{s,n}(\s_{-n}), \tilde{v}_i - \tilde{v}_j \rangle = 0$, $i,j \in \{1,...,k_n\}$.

\item If $V^{s,n}(\s_{-n}) \notin N_{P^s_n}(v)$, then for all $\tilde{v} \in \{\tilde{v}_1,...,\tilde{v}_{k_n} \}$, it implies that $\langle V^{s,n}(\s_{-n}), \tilde{v} - v \rangle >0$.

\end{enumerate} 

We now finish the proof of the Lemma. Let $z^{\l}_n \equiv \s_n + \lambda V^{s,n}(\s_{-n})$.  Write $r_n(z^{\l}_n) = \sum_i \alpha^{\l}_i \tilde{v}_i + \sum_{t} \beta^{\l }_t v_t$, with $\alpha^{\l}_i \geq 0, \beta^{\l}_t \geq 0$ and $\sum_i \alpha^{\l}_i + \sum_{t} \beta^{\l}_t =1$. Let $\tilde{v} \in \{\tilde{v}_1,...,\tilde{v}_{k_n} \}$. Property (1) implies $\langle z^{\l}_n - r_n(z^{\l}_n), \tilde{v} - r_n(z^{\l}_n) \rangle = \langle \s_n - r_n(z^{\l}_n), \tilde{v} - r_n(z^{\l}_n) \rangle +  \lambda \sum_t\beta^{\l}_t \langle V^{s,n}(\s_{-n}), \tilde{v} - v_t \rangle$. The first term of the previous sum is bounded; since $\langle  z^{\l}_n - r_n(z^{\l}_n), \tilde{v} - r_n(z^{\l}_n) \rangle \leq 0$, property (2) now implies that for sufficiently large $\l$, $\b^{\l}_t =0, \forall t$. This shows that $r_n(z^{\l}_n) = \sum_i \alpha^{\l}_i \tilde{v}_i$, which concludes the proof of the Lemma. \end{proof}

We now conclude the proof of Claim \ref{eqindex}. We show that for large enough $\lambda$ the map $g^{\lambda}:P^s \rightarrow P^s$ defined by $g^{\lambda} \equiv \times_n (r_n \circ w^{\lambda}_n)$ with $w^{\lambda}_n (\s) = \s_n + \lambda V^{s,n}(\s_{-n})$ satisfies $\text{Graph}(g^{\lambda}) \subset O$, where $O$ is an adequate homotopy neighborhood of $\text{Graph}(BR^{V^s})$ for $e(Q)$ (recall the definition of the adequate homotopy neighborhood from subsection \ref{indextheorygeneral}).

Suppose by contradiction the claim is not true. Then there exists a sequence $\lambda_k \rightarrow +\infty$ as $k \rightarrow \infty$ such that $g^{k} \equiv g^{\lambda_k}$ satisfies $\text{Graph}(g^{k}) \cap O^{c} \neq \emptyset$ for all $k$. Since $P$ is compact and $g^{k}$ continuous, $(\text{Graph}(g^{k}))_{k \in \mathbb{N}}$ is a sequence of non-empty compact subsets of $P^s \times P^s$. This implies we can extract a convergent subsequence (in the Hausdorff-distance) of $(\text{Graph}(g^k))_{k \in \mathbb{N}}$ to a nonempty compact subset of $P^s \times P^s$. Passing to a convergent subsequence if necessary, we can assume that $\text{Graph}(g^{k})$ converges to a nonempty compact set $\mathfrak{F}$ in the Hausdorff-distance. It follows that $\mathfrak{F} \cap O^{c} \neq \emptyset$. Let $z \in \mathfrak{F} \cap O^{c}$. Consider $B_z$ an open neighborhood of $z$ that does not intersect $\text{Graph}(BR^{V^s})$. Therefore we have that Graph$(g^{k}) \cap B_z \neq \emptyset$ for sufficiently large $k$.\footnote{This follows from the characterization of the Hausdorff limit $\mathfrak{F}$ as the closed limit of the sequence \text{Graph}$(g^k)$. See \cite{AB2006}.} This implies that there exists an open set $U$ in $P^s$ such that,  for $k$ large enough, $\forall \s \in U, (\s, g^{k}(\s)) \in B_z$. 

Fix $\s \in U$. By Lemma \ref{auxclaim}, it follows that for large enough $k$, $r_n(\s_n + \lambda_kV^{s,n}(\s_{-n})) \in BR^{V^s}_n(\s_{-n}), \forall n \in \N$. Therefore, for $k$ large enough, we have that $g^{k}(\s) \in BR^{V^s}(\s)$. This implies that for $k$ sufficiently large $(\s, g^{k}(\s)) \in (B_z)^{c}$. Contradiction. 

Hence there exists $\lambda_0 > 0$ such that $\forall \lambda \geq \lambda_0$  we have $\text{Graph}(g^{\lambda}) \subset O$. Now define the homotopy $H: [0,1] \times P^s \rightarrow P^s$ such that $H(t, \cdot) = g^{1 + t(\lambda -1)}(\cdot)$. Notice that the polytope-form games denoted by $V^{s,1 + t(\lambda -1)}$ with payoffs given by $[1 + t(\lambda -1)]V^{s}_n$, for each player $n$, all have the same equilibria, which implies that their associated GPS-maps $g^{1 + t(\lambda -1)}$ all have the same fixed points. Therefore the homotopy  $H$ preserves fixed points. This implies that the indices of a component of equilibria under $g^{1}$ and $g^{\lambda}$ are identical, by the homotopy property of the index (Theorem 5.15, Chapter VII in \cite{D1972}). Since $\text{Graph}(g^{\lambda})$ is contained in the homotopy neighborhood $O$ of $\text{Graph}(BR^V)$, this implies that the index of $e(Q)$ under $g^1$  (the GPS-map of $V^s$) is identical to $\text{Ind}_{BR^{V^s}}(e(Q))$, which concludes the proof. \end{proof}

Recall for the next claim that $q^V$ is the reduction map from $V$ to $V'$.

\begin{claim}\label{BRindex}The following equation holds: $\text{ind}_{BR^V}(Q) = \text{ind}_{BR^{V'}}(q^V(Q)) $\end{claim}

\begin{proof}Firstly, if $\sigma$ is an equilibrium of $V$, then $q^V(\sigma)$ is an equilibrium of $V^{'}$, so $q^V(Q)$ is an equilibrium of $V'$. Fix $U$ a neighborhood of $Q$ in $P$ with cl$_{P}(U) \cap E(V) = Q$. Letting $U' \equiv q^{V}(U)$, $U'$ is open in $P'$ and cl$_{P'}(U') \cap E(V') = q^{V}(Q)$.\footnote{This follows from the fact that the map $q^{V}_n$ is an affine and surjective mapping, so it is an open mapping. This plus the Closed Map Lemma implies that $q^{V}$ is an open and closed map, which implies cl$_{P'}(U') \cap E(V') = q^{V}(Q)$.} Let now $W$ be an open neighborhood of $\text{Graph}(BR^{V})$ of the best reply of $V$ such that the best-reply index of $Q$ can be computed from the fixed-point index at $U$ of any continuous function $h: P \rightarrow P$ with $\text{Graph}(h) \subset W$. Consider now an open neighborhood $W^{'}$ of $\text{Graph}(BR^{V^{'}})$ such that for each $(\sigma', \tau') \in W^{'}$, $(q^{V} \times q^{V})^{-1}(\sigma', \tau') \subset W$. By the definition of the best-reply index, there exists a function $h^{'}: P^{'} \rightarrow P^{'}$ with Graph$(h^{'}) \subset W^{'}$ such that the fixed-point index of $h^{'}|_{U^{'}}$ is well defined and is the best reply index of $q^{V}(Q)$ w.r.t. $V'$. Let $j_n$ be a right inverse of $q^V_n$. By construction, we have that Graph$(j \circ h^{'} \circ q^{V}) \subset W$. This implies that the fixed point index of $(j \circ h^{'} \circ q^V)|_{U}$  equals the best-reply index of $Q$ w.r.t. $V$. Let $h \equiv j \circ h^{'} \circ q^{V}$. Because of the commutativity property of the index in Theorem 5.16 in Chapter VII of \cite{D1972}, we have that $h$ and $h^{'}$ have homeomorphic sets of fixed points and their indices agree: indeed, defining $h_0 \equiv h^{'} \circ q^V$ we have that $j \circ h_0 = h$ and $h_0 \circ j = h^{'}$. This implies that the fixed point index of $h|_U$ equals the fixed-point index of $h'|_{U'}$, which concludes the result. \end{proof}

\begin{proof}[Proof of Proposition \ref{identicaltheories}] 
We first prove (1). Claims \ref{standarddeg}, \ref{eq index and deg}, \ref{eqindex} and \ref{BRindex} imply that $\text{deg}^{PF}_{V}(Q) = \text{ind}_{BR^V}(Q)$, since the standartization is a reduction. For the same reason we have $\text{deg}^{PF}_{V'}(Q') = \text{ind}_{BR^{V'}}(Q')$. Now, Claim \ref{BRindex} shows invariance of the best-reply index to reductions. Therefore, $\text{deg}^{PF}_{V'}(Q') =  \text{deg}^{PF}_{V}(Q)$. The exact same reasoning applied to $\bar V$ gives (1). 

We now show (2). Claim \ref{standarddeg} shows the degree is invariant to standartization. Claim \ref{eq index and deg} shows equality of the index and degree in standard polytope-form games. The commutativity property of the index now immediately gives that the index is invariant to standartizations. Given we have proved (1), we therefore have (2).

In order to obtain a proof of (3), observe again the the degree is invariant to standartization from Claim \ref{standarddeg}. From claims \ref{eq index and deg}, \ref{eqindex} and \ref{BRindex},  (3) now follows.\end{proof}

\subsection{Additional Results on Extensive-form Games}\label{addextensiveform} Throughout the subsection we fix a game tree $\Gamma$, without moves of Nature. This is only for simplicity of exposition, since all results could be straightforwardly generalized by considering Nature as a player (without payoffs), playing a fixed strategy. We maintain the notation of section \ref{extensiveform}. We start with an auxiliary proposition that characterizes the interior of the enabling strategy set $C_n$ of player $n$.

\begin{proposition}
Let $\mathring{C}_n \equiv \{p_n \in [0,1]^{L_n} | (\exists \sigma_n \in$ int$(\Sigma_n))$ s.t. $p_n(i) = \sum_{s \in s_n(i)}\sigma_n(s) \}$. Then $\mathring{C}_n = \text{int}(C_n)$.
\end{proposition}

\begin{proof} 
It it clear that $\mathring{C}_n \subset \text{int}(C_n)$. So we show the converse. First, $\text{int}(C_n)$ is open in $C_n$ so $q_n^{-1}(\text{int}(C_n))$ is open in $\Sigma_n$, by continuity of $q_n$. It implies $q_n^{-1}(\text{int}(C_n)) \subset \text{int}(\Sigma_n)$. Therefore if $p_n \in \text{int}(C_n)$, then $q_n^{-1}(p_n) \subset \text{int}(\Sigma_n)$. Hence there exists $\sigma_n \in \text{int}(\Sigma_n)$ such that $q_n(\sigma_n) = p_n \iff p_n(i) = \Sigma_{s \in s_n(i)}\sigma_n(s)$, for all $i \in L_n$. So $p_n \in \mathring{C_n}$.\end{proof}

\begin{definition}A profile of behavior strategies $b = (b_n)_{n \in \N}$ \textit{induces an enabling profile} $p = (p_n)_{n \in \N}$ if any mixed strategy profile $\sigma = (\sigma_n)_{n \in \N}$ that is equivalent to $b$\footnote{A mixed strategy profile $\sigma$ is \textit{equivalent} to a behavior profile $b$ if for each $n$,  $\sigma_n$ is equivalent to $b_n$. The mixed strategy $\sigma_n$ is \textit{equivalent} to $b_n$ if for any mixed/behavior profile $\sigma_{-n}$ the distribution over terminal nodes induced by $(\sigma_n,\sigma_{-n})$ and $(b_n, \sigma_{-n})$ is the same (see \cite{MSZ2013}, p. 223).} satisfies $p_n(i) = \sum_{s \in s_n(i)}\sigma_n(s), \forall i \in L_n$.\end{definition}

The next proposition establishes the relation between enabling and behavior strategies.

\begin{proposition}\label{ben}The following hold:

\begin{enumerate}
\item  Given a profile of behavior strategies $b = (b_n)_{n \in \N}$, there exists a unique profile of enabling strategies induced by $b$.
\item  Let $(p_n)_{n \in \N}$ be a profile of enabling strategies with $p_n \in \mathring{C}_n$. Then there exists a unique profile of behavior strategies that induces $(p_n)_{n \in \N}$. 
\end{enumerate}

\end{proposition}

\begin{proof}
We prove (1). Given a profile of behavior strategies $b$, define $p_n(i) \equiv \Pi_{(u^{'},i^{'}) \preceq (u,i)}b_n(i^{'}|u^{'})$ for each $i \in L_n$. Let $\sigma$ be a profile that is equivalent to $b$. Then equivalence implies that, for each $n \in \N, i \in L_n$,  $\sum_{s \in s_n(i)}\sigma_n(s) = \Pi_{(u^{'},i^{'}) \preceq (u,i)}b_n(i^{'}|u^{'})$. Therefore $p_n \in C_n$ and $b$ induces $p$. Uniqueness follows immediately.

	We prove (2). Let $(p_n)_{n \in \N}$ be an enabling profile with $p_n \in \mathring{C}_n$. Then there exists $\sigma_n \in \text{int}(\Sigma_n)$ satisfying $p_n(i) = \sum_{s \in s_n(i)}\sigma_n(s)$. Define $\beta_n(u,i) \equiv p_n(i)$, where $i \in A_n(u)$ is a last action of player $n$. Then, if $(u^{'},i^{'})$ is an immediate predecessor of $u$ among $n's$ information set, define $\beta_n(u^{'},i^{'}) \equiv \sum_{i \in A_n(u)}\beta_n(u,i)$. Proceeding in this manner, we define $\beta_n(u,i)$ for every pair $(u,i), u \in U, i \in A_n(u)$. Notice that because of the assumption $p_n \in \mathring{C}_n$, it follows that $\beta_n(u,i)>0$ for all $u \in U_n$ and $i \in A_n(u)$. Therefore we can define the behavior strategy $b_n(i | u) \equiv \frac{\beta_n(u,i)}{\beta_n(u^{'},i^{'})}$, where $(u^{'},i^{'}) \prec u$ and $u^{'}$ is an immediate predecessor of $u$.  Now, for any $\tilde{\sigma}_n$ equivalent to $b_n$, it must be that $\sum_{s \in s_n(i)}\sigma_n(s) = \sum_{s \in s_n(i)}\tilde{\sigma}_n(s), \forall i \in L_n$, otherwise we can construct $\sigma^{'}_{-n}$ such that $(\sigma_n, \sigma^{'}_{-n})$ and $(\tilde{\sigma}_n, \sigma^{'}_{-n})$ do not induce the same distributions over terminal nodes. This implies that $(b_n)_{n \in \N}$ induces $(p_n)_{n \in \N}$. 
	
	Now, for uniqueness, suppose $b'$ is a profile of behavior strategies inducing $(p_n)_{n \in \N}$. By the proof of (1), it follows that $p_n(i) = \Pi_{(u^{'},i^{'}) \preceq (u,i)}b'_n(i^{'}|u^{'}), \forall i \in L_n$. For each $u \in U_n$ and $i \in L_n$ such that $i \in A_n(u)$ set $\beta'_n(u,i) = p_n(i)$. Proceeding in the same fashion as we did for $\beta_n$, the numbers $\beta'_n(u',i')>0$ for each $u' \in U_n$ and $i \in A_n(u')$ are uniquely determined. This implies therefore that for each $u \in U$ and $i \in A_n(u)$, $b'_n(i|u) = b_n(i|u)$, which shows uniqueness.  \end{proof}

\begin{remark}Notice that for each $p_n \in C_n$, there exists a behavior strategy profile $(b_n)_{n \in \N}$ inducing $(p_n)_{n \in \N}$, but this behavior strategy need not be unique. This happens when certain last actions have probability 0 for a certain player. Still, whenever we have $p_n(i)>0$, it is possible to proceed as in the proof of (2) Proposition \ref{ben} and derive $b_n(i'|u')$ for each $(u',i') \preceq (u,i), i \in A_n(u)$. For the remaining pairs $(u,i)$, the probabilities $b_n(i|u)$ are undetermined. \end{remark}

\begin{definition}
An enabling profile $p = (p_n)_{n \in \N}$ induces a distribution over terminal nodes $F \in \Delta(Z)$ if for any profile of behavior strategies $(b_{n})_{n \in \N}$ inducing $p$, it implies that  $(b_{n})_{n \in \N}$ induces $F$.
\end{definition}

\begin{corollary} Let $(p_n)_{n \in \N}$ be a profile of enabling strategies with $p_n \in C_n$. Then there exists a unique distribution $F \in \Delta(Z)$ induced by this profile of enabling strategies. Conversely, given $F \in \text{int} (\Delta(Z))$, there exists a unique profile $(p_n)_{n \in \N}$, with $p_n \in \mathring{C}_n$, that induces $F$. \end{corollary}

\begin{proof}The first part of the statement is straightforward by an application of (1) of Proposition \ref{ben}. We prove the second part. Given $F \in \text{int}(\Delta(Z))$ we show there exists a unique behavior strategy profile $(b_n)_{n \in \N}$ that induces $F$. Let $i$ be an action of player $n$ at an information set $u \in U_n$ and define $$b_n(i|u) := \frac{\sum_{z: (u,i) \prec z, u \in U_n}F(z)}{\sum_{z: u \prec z, u \in U_n}F(z)}.$$ This defines the unique behavior strategy $b_n$ and the profile $(b_n)_{n \in \N}$ induces $F$.  Also, by the proof of (1) in Proposition \ref{ben}, $(b_n)_{n \in \N}$ induces a unique enabling profile $(p_n)_{n \in \N}$ with $p_n(i) = \Pi_{(u',i') \preceq (u,i)}b_n(i'|u') >0$ for each $i \in L_n, n \in \N$.\end{proof}

\begin{Proposition}\label{linearsubspaceofpayoffs}
Then there exists a linear subspace $A^{\circ}(\times_n C_n)$ of the multiaffine functions over $\times_n C_n$ such that: 
\begin{enumerate}

\item For any  $(V^{e}_n)_{n \in \N} \in \times_{n \in \N} A^{\circ}(\times_n C_n)$ there exists a unique $G \in \mathcal{G}$, such that for any profile of enabling strategies $(p_n)_{n \in \N}$ and induced distribution $F \in \Delta(Z)$, $V^{e}_n(p_n, p_{-n}) = \sum_{z \in Z}G_n(z)F(z)$. 

\item Conversely, for each $G \in \mathcal{G}$ and $n \in \N$ there exists a unique multiaffine function $V^e_n \in A^{\circ}(\times_m C_m)$ such that for any profile of enabling strategies $(p_m)_{m \in \N}$ and induced distribution $F \in \Delta(Z)$, $V^e_n(p_n, p_{-n}) = \sum_{z \in Z}G_n(z)F(z)$.

\end{enumerate}
\end{Proposition}

\begin{proof}


We first construct the linear subspace $A^{\circ}(\times_{n}C_n)$ of the statement. For each $z \in Z$, there exists an unique path in $\Gamma$ from the root to $z$. Therefore, for each $z \in Z$, there exists a unique pair of set $\N^*(z) \subset \N$ and vector $(i_n)_{n \in \N^*(z)}$ of last actions such that $\ell_n(z) = i_n \in L_n$, for $n \in \N^{*}(z)$ and $\ell_n(z) = \emptyset$, for $n \in \N \setminus \N^{*}(z)$. We call this vector \textit{the unique vector of last actions associated to z}. Let $W \equiv \{ (i_n)_{n \in \N^{*}} | \exists z \in Z$ s.t. $\N^* = \N^*(z) \}$. Define the multiaffine function $V^{e}_n$ over $\times_n C_n$ by: 
 
 $$V^{e}_n(p_1,...,p_N) \equiv \sum_{(i^{*}_{j_1},...,i^{*}_{j_{n}}) \in W} V^{e}_n(i^{*}_{j_1},...,i^{*}_{j_n})p_{j_1}(i^{*}_{j_1})...p_{j_n}(i^{*}_{j_n})$$


Notice that the set of affine functions satisfying the formula above forms a linear subspace of the space of multiaffine functions over $\times_n C_n$ (under pointwise addition and scalar multiplication). Call this subspace $A^{\circ}(\times_n C_n)$.

We prove (1). Let $V^e_n \in A^{\circ}(\times_n C_n)$. We now show that this function defines unique payoffs over terminal nodes of $\Gamma$ for player $n$. For $z \in Z$, consider the unique vector of last actions $(i_n)_{n \in \N^{*}(z)} \in W$ associated to $z$. Define, for each $m \in \N$, $G_m(z) \equiv V_m(i^{*}_{j_1},...,i^{*}_{j_n})$. Let $(p_n)_{n \in \N}$ be a profile of enabling strategies and $F$ the induced distribution over $Z$. For $z \in Z$, if $(i^{*}_{j_1},...,i^{*}_{j_n}) \in W$ is associated to $z$, then it implies that: $$F(z)  = F_{j_1}(z)...F_{j_n}(z) = p_{j_1}(i^{*}_{j_1})...p_{j_n}(i^{*}_{j_n}), $$ where $F_{j_i}(z)  \equiv p_{j_i}(i)$, for any $z \in Z_n(i)$  (cf. \cite{GW2002}). Then $G_n(z)F(z) = V^e_n(i^{*}_{j_1},...,i^{*}_{j_n})p_{j_1}(i^{*}_{j_1})...p_{j_n}(i^{*}_{j_n})$. This implies that 



\begin{equation}\label{equality of payoffs}
\sum_{z \in Z}G_n(z)F(z) = \sum_{(i^{*}_{j_1},...,i^{*}_{j_{n}}) \in W} V^e_n(i^{*}_{j_1},...,i^{*}_{j_n})p_{j_1}(i^{*}_{j_1})...p_{j_n}(i^{*}_{j_n}) 
\end{equation}

Now we show (2). For each $z \in Z$, consider the unique $(i^{*}_{j_1},...,i^{*}_{j_n}) \in W$ associated to $z$. Define $V^e_n(i^{*}_{j_1},...,i^{*}_{j_n}) \equiv G_n(z)$. Define the multiaffine function for player $n \in \N$ by:  $$V^e_n(p_1,...,p_N) \equiv \sum_{(i^{*}_{j_1},...,i^{*}_{j_{n}}) \in W} V^e_n(i^{*}_{j_1},...,i^{*}_{j_n})p_{j_1}(i^{*}_{j_1})...p_{j_n}(i^{*}_{j_n}).$$ Then $V^e_n$ belongs to $A^{\circ}(\times_n C_n)$ and expected payoffs agree. \end{proof}

An immediate consequence of the proof of the Proposition \ref{linearsubspaceofpayoffs} above is the following.

\begin{corollary}\label{representation}
Let $C = \times_n C_n$ and $R: \times_{n \in \N} A^{\circ}(C) \rightarrow \mathbb{R}^{N|Z|}$ defined by $R(V^e) = G$, where $R \equiv \times_n R_n$ and $R_n(V^e_n) = G_n \in \mathbb{R}^{|Z|}$, in which $G_n$ is the unique vector of terminal payoffs obtained in (1) of Proposition \ref{linearsubspaceofpayoffs}. Then $R$ is a linear isomorphism.  \end{corollary}

\subsection{Proof of Theorem \ref{stabilization}}

We start by observing that for the normal-form game $\mathbb{G}$, $\text{deg}^{KM}_{\mathbb{G}}(X) = \text{ind}_{BR^{\mathbb{G}}}(X)$ (cf. \cite{GW2005}).  Let $q^{e}: \S_n \to C_n$ be the reduction map from mixed to enabling strategies and define $U \equiv (q^{e})^{-1}(W)$. Let $P^{\e}_n \equiv (q^{e}_n)^{-1}(C^{\e}_n) \subseteq \S_n \setminus \partial \S_n$. Notice that the game $\mathbb{G}|_{P^{\e}}$ where we restrict the mixed strategy set of each player $n$ in $\mathbb{G}$ to $P^{\e}_n$ is a polytope-form game. We denote by $E(\mG|_{P^{\e}})$ the set of $\e$-restricted equilibria of $\mG|_{P^{\e}}$. We now claim:

\begin{claim}\label{preparation} There exists $\bar \e>0$, such that for each $\e \in (0, \bar \e)$, $U^{\e} \equiv U \cap P^{\e}$ satisfies (cl$_{P^{\e}}(U^{\e}) \setminus U^{\e}) \cap E(\mathbb{G}|_{P^{\e}}) = \emptyset$ and the following equality holds: $$\text{ind}_{BR^{\mathbb{G}}}(X) = \sum_{X' \in U^{\e}}\text{ind}_{BR^{\mathbb{G}|_{P^{\e}}}}(X'),$$ where the sum is over the connected components of fixed points $X' \in U^{\e}$ of $BR^{\mathbb{G}|_{P^{\e}}}$.   \end{claim}
 
\begin{proof}For $\e>0$, let $r^{\epsilon}_n : \Sigma_n \rightarrow P^{\epsilon}_n$ be the nearest-point retraction and $r^{\epsilon} \equiv \times_n r^{\epsilon}_n$. Let $i_n: P^{\epsilon}_n \rightarrow \Sigma_n$ be the inclusion map and $i \equiv \times_n i_n$. Define the correspondence $\Gamma_{\epsilon}: \Sigma \rightrightarrows \Sigma$ by $\Gamma_{\epsilon}(\sigma) = (i \circ BR^{\mG|_{P^{\epsilon}}} \circ r^{\epsilon})(\sigma)$. Notice that Graph$(\Gamma_{\epsilon})$ converges (in the Hausdorff-distance) to $\text{Graph}(BR^{\mG})$ as $\epsilon \rightarrow 0$. Let $O$ be an open neighborhood of the Graph$(BR^{\mathbb{G}})$ that does not intersect $\{ (\sigma, \sigma) \in \Sigma \times \Sigma | \sigma$ cl$(U) \setminus U\}$ and according to which the index of the best reply $BR^{\mathbb{G}}$ at $U$ can be computed from any continuous map $h: \Sigma \rightarrow \Sigma $ with $\text{Graph}(h) \subset O$. Then, for $\epsilon>0$ sufficiently small we have that Graph$(\Gamma_{\epsilon}) \subset O$. Taking further $\epsilon >0$ sufficiently small, then $U^{\epsilon}$ contains no equilibria of $\mG|_{P^{\epsilon}}$ in its boundary (in $P^{\epsilon}$). Then there exists a continuous function $h^{\epsilon}: P^{\epsilon} \rightarrow P^{\epsilon}$ such that the fixed point index of $h^{\epsilon}|_{U^{\epsilon}}$ is well defined and equals the index of the best reply of $\mG|_{P^{\epsilon}}$ at $U^{\epsilon}$ (see \cite{AM1989}). Moreover, we can assume  Graph$(i \circ h^{\epsilon} \circ r^{\epsilon}) \subset O$. This implies that the fixed point index of $(i \circ h^{\epsilon} \circ r^{\epsilon})|_{U}$ is equal to the local index of the best reply of $\mG$ at $U$. Finally, the commutativity property of the fixed point index (Theorem 5.16 in Chapter VII of \cite{D1972}) shows that the fixed point index of $(i \circ h^{\epsilon} \circ r^{\epsilon})|_{U}$ equals the fixed point index of $h^{\epsilon}|_{U^{\epsilon}}$, which is the local index of the best reply of $\mG|_{P^{\epsilon}}$ at $U^{\epsilon}$. \end{proof}

Fix from now on $\e>0$ according to Claim \ref{preparation}. The reduction map $q^{\e} \equiv q^{e}|_{P^{\e}}: P^{\e} \to C^{\e}$ defines a reduction of $\mathbb{G}|_{P^{\e}}$, which we denote by $V^{\e}$. From the proof of Proposition \ref{identicaltheories}, we have  $\text{ind}_{BR^{\mathbb{G}}|_{P^{\e}}}(X') = \text{ind}_{BR^{V^{\e}}}(q^{\e}(X')) = \text{deg}^{PF}_{V^{\e}}(q^{\e}(X'))$. We now show $\text{deg}^{PF}_{V^{\e}}(q^{\e}(X')) = \text{deg}^{\e}_{G}(q^{\e}(X'))$, which concludes the proof. 

Let $\mathfrak{r}^{\e}_n:  \mathbb{R}^{L_n} \rightarrow C^{\epsilon}_n$ be the nearest-point retraction to $C^{\epsilon}_n$. Let $\mathfrak{r}^{\e} \equiv \times_n \mathfrak{r}^{\e}_n$. Consider also the map $\omega^{\e}_n: C^{\epsilon}_n \rightarrow \mathbb{R}^{L_n}$ defined by $\omega^\e_n(p_n) \equiv p_n + \nu_{n}(p_{-n})$ with $\omega^{\e} \equiv \times_n \omega^{\e}_n$ and let $\Phi^\e_G: \times_n C^{\epsilon}_n \rightarrow \times_n C^{\epsilon}_n$ be given by $\Phi^\e_G(p) = \mathfrak{r}^{\e} \circ \omega^{\e}$. Lemma 5.1 in \cite{GW2002} shows that a profile of enabling strategies $(p_n)_{n \in \N}$ is an equilibrium of the extensive-form game $G$ with perturbed enabling strategies $C^{\e}$  if and only if it is a fixed point of the map $\Phi^\e_G$. The map $\Phi^\e_G$ is the analogous version of the GPS-map formulated to the extensive-form game $G$. As the map $\Phi^\e_G$ is jointly continuous on $G$ and $p$, it follows that $(G,p) \mapsto \Phi^\e_G(p)$ is a Nash-map.  Let  $\phi^\e: \mathcal{G} \times C^{\epsilon} \rightarrow C^{\epsilon}$ be defined as $\phi^\e(G,p) = \Phi^\e_G(p)$.

It is without loss of generality to assume that $A^{\circ}(C^{\epsilon}) = A^{\circ}(C)$, by taking $\e>0$ smaller if necessary. Using then the linear isomorphism $R: \times_n A^{\circ}(C^{\epsilon}) \rightarrow \mathbb{R}^{N|Z|} $ from Corollary \ref{representation}, let $\tilde{\phi}^\e: A^{\circ}(C^{\epsilon}) \times C^{\epsilon} \rightarrow C^{\epsilon}$ be defined by $\tilde{\phi}^\e = \phi \circ (R \times id_{C^{\epsilon}})$. Then $\tilde{\phi}^\e$ is also a Nash-map for polytope-form games with payoff functions in $A^{\circ}(C^{\epsilon})$ for each player $n$. Notice that $\phi^\e(G, \cdot) = \tilde{\phi}^\e(V^e, \cdot)$, where $R^{-1}(G) = V^{e}$, so the fixed points and the indices assigned to these fixed points are the same according to the two Nash maps $\phi^\e$ and $\tilde{\phi}^\e$. 

Now recall from subsection \ref{polytope-form degree} that $T: \times_n A^{\circ}(C^{\epsilon}) \rightarrow \Re^{ND}$ is a linear isomorphism from the multiaffine functions over $C^{\epsilon}$ to $\Re^{ND}$, where $D$ is appropriately defined. Let $P^{s,\e}_n$ be the standard polytope resulting from an affine map $e^\e_n$ of $C^{\epsilon}_n$ to $\mathbb{R}^{d_n}$, where the dimension of $C^{\epsilon}_n$ is $d_n -1$.  Let $V^{s,\e} = (\N, (P^{s,\e}_n)_{n \in \N}, (V^s_n)_{n \in \N})$ be the standard polytope-form game obtained from $V^{\e} = (\N, (C^{\epsilon}_n)_{n \in \N}, (V^{e}_n)_{n \in \N})$. Now let $\bar{\phi}^\e: \Re^{ND} \times P^{s,\e} \rightarrow P^{s,\e}$ be defined by $\bar{\phi}^\e =  e^\e \circ \tilde{\phi}^\e \circ  (T^{-1} \times (e^{\e})^{-1})$. Notice that $\bar{\phi}^\e(V^{s,\e}, \cdot) = e^{\e} \circ \tilde{\phi}^\e(V^{\e}, \cdot) \circ (e^{\e})^{-1}$ which implies, by the comutativity property of the index, that the fixed point sets of $\tilde{\phi}^{\e}_{V^{\e}}$ and $\bar{\phi}^\e_{V^{s,\e}}$ are homeomorphic and their indices agree. Also, by construction $\bar{\phi}^{\e}: \Re^{ND} \times P^{s,\e} \rightarrow P^{s,\e}$ is a Nash map.

It now follows that: $$\text{deg}_{V^{\e}}(q^{\e}(X')) =  \text{deg}_{V^{s,\e}}(e^{\e}(q^{\e}(X'))) = \text{ind}_{V^{s,\e}}(e^{\e}(q^{\e}(X'))) =$$ $$=\text{ind}(e^{\e}(q^{\e}(X')), \bar{\phi}^\e_{V^{s,\e}}) = \text{ind}(q^{\e}(X'), \tilde{\phi}^\e_{V^{\e}}) = \text{ind}(q^{\e}(X'), \Phi^{\e}_G),$$ where the first and second follow from Proposition \ref{identicaltheories}, the third, fourth and fifth from our reasoning above. 

We claim $\text{ind}(q^{\e}(Q), \Phi^\e_G) = \text{deg}^{\e}_G(q^{\e}(Q))$, which concludes the proof. We denote by $G \bigoplus g$, where $g = (g_n)_{n \in \N}$ and $g_n \in \mathbb{R}^{L_n}$, the extensive-form game where, for each player $n$, the payoffs over terminal nodes are given by $ (G_n(z) + g_n(\ell_n(z)))$, if $\ell_n(z) \neq \emptyset$ and $G_n(z)$, if $\ell_n(z) = \emptyset$ . Let $\mathcal{E}^\e_{G} = \{(g,p) \in \times_n \mathbb{R}^{L_n} \times  C^{\epsilon} \mid$ $p$ is an $\e$-restricted equilibrium of the game $G \bigoplus g \}$. Define now $\theta^\e_G: \mathcal{E}^\e_G \rightarrow \times_{n} \mathbb{R}^{L_n}$ by $(\theta^\e_{G})_n(g,p) = (p_n(i) + \nu_n(i,p_{-n}) + g_n(i))_{i \in L_n}$. Using the same reasoning as in Claim \ref{eq index and deg} we can show that $\theta^\e_G$ is a homeomorphism and $(\theta^{\e}_G)^{-1}(g) = (f^\e(g), r(g))$, where $f^\e$ is the displacement of the permuted GPS map of $\Phi^\e_G$. Define $\text{proj}_{g}: \mathcal{E}^{\e}_{G} \rightarrow \times_n \mathbb{R}^{L_n}$ as proj$_g(g,p) = g$. Then it implies that $\text{proj}_g \circ (\theta^{\e}_G)^{-1}(g) = f^\e(g)$. 

Let now $\Theta^{GW}_{\e}: \mathcal{E}^{GW}_{\e} \rightarrow \mathcal{G}$ be the homeomorphism presented in Theorem 5.2 in \cite{GW2002}. As in \cite{KM1986}, we can reparametrize the graph of equilibria $\mathcal{E}^{GW}_{\e}$ writing an element $(G,p) \in \mathcal{E}^{GW}_{\e}$ as $(\tilde{G},\tilde g,p)$, where $G$ is uniquely written as $G = \tilde{G} \bigoplus \tilde g$, with $\tilde{G}_n(z) \equiv G_n(z) - \mathbb{E}[G| Z_n(\ell_n(z))]$ and $\tilde g_n(z) \equiv \mathbb{E}[G| Z_n(\ell_n(z))]$,  if $\ell_n(z) \neq \emptyset$ and $\tilde{G}_n(z) = G_n(z)$, if otherwise. Notice that such a vector $\tilde g$ of the decomposition can be assumed to be in $\times_n \mathbb{R}^{L_n}$ since, for $z,z' \in Z$ with $\ell_n(z) = i = \ell_n(z')$ we have that $\tilde g_n(z) = \tilde g_n(z')$ and if $z \in Z$ is such that $\ell_n(z) = \emptyset$ we have that $\tilde g_n(z) = 0$. Then $\Theta^{GW}_{\e}$ is defined by $\Theta^{GW}_{\e}(\tilde{G}, \tilde g, p) = (\tilde{G}, t)$, where $\tilde{G} \equiv (\tilde{G}_n)_{n \in \N}$, $\tilde{G}_n \equiv (\tilde{G}_n(z))_{z \in Z}$  and $t \in \times_n \mathbb{R}^{L_n}$ with $t_n(i) \equiv p_n(i) + \nu_n(i,p_{-n}), i \in L_n$. The payoff over terminal nodes represented in the vector $(\tilde{G}, t)$ is given by $\tilde{G}_n(z) + t_n(\ell_n(z))$. Notice that $\Theta^{GW}_{\e}$ is the identity in the first coordinate, similarly to what happened with the homeomorphism $\theta^s$ in Lemma \ref{ST}. 

For the fixed game $G \in \mathcal{G}$ consider now its decomposition $(\tilde{G}, \tilde g)$. Define the graph $\mathcal{E}^\e_{\tilde{G}} \equiv \{ (g',p) \in \times_n \mathbb{R}^{L_n} \times (\times_n C^{\epsilon}_n) \mid$ $p$ is an $\e$-equilibrium of $\tilde{G} \bigoplus g' \}$ and $\theta^{'}_\e: \mathcal{E}^\e_{\tilde{G}} \rightarrow \times_n \mathbb{R}^{L_n}$ defined by $\theta^{'}_\e(g',p) = t$, where $t$ satisfies $\Theta^{GW}_{\e}(\tilde{G},g',p) = (\tilde{G}, t)$. Therefore, $\theta^{'}_\e$ is a homeomorphism. Let $\text{proj}^{'}: \mathcal{E}^\e_{\tilde{G}} \rightarrow \times_n \mathbb{R}^{L_n}$ be the projection over the $g'$-coordinate. We can now define the local degree of the $\e$-restricted equilibrium component $q^{\e}(X')$ of game $G$ as the local degree of the projection $\text{proj}^{'}|_{W}$, where $W$ is an open neighborhood in $\mathcal{E}^\e_{\tilde{G}}$ of $\{g\} \times q^{\e}(X')$ that contains no other pair $(g,p)$ in its boundary. This local degree, by the same argument as in the proof of  Claim \ref{eq index and deg}, agrees with the degree of $q^{\e}(X')$ computed from $\text{proj}_{\mathcal{G}} : \mathcal{E}^{GW}_{\e} \rightarrow \mathcal{G}$.

Note now that $\theta^\e_G(\{0\} \times q^{\e}(X')) = \theta^{'}_\e(\{g\} \times q^{\e}(X'))$. Letting $f^{'}_\e \equiv \text{proj}' \circ (\theta'_{\e})^{-1}$, we have therefore that $f^\e = f^{'}_\e - g$. Therefore $f^\e$ and $f'_\e$ have the same local degrees. Since $f^\e$ is the displacement map of the permuted $GPS$ map of $\Phi^\e_G$, it follows that the local degree of $f^\e$ and the index of the $GPS$ map are equal, which implies ind$(q^{\e}(X'),\Phi^\e_G) = \text{deg}^{\e}_G(q^{\e}(X')).$

\renewcommand{\bibfont}{\small}
\bibliographystyle{abbrvnat}
\bibliography{references}

\end{document}